\documentclass[12pt]{article}

\usepackage[letterpaper,margin=1in]{geometry}
\usepackage[onehalfspacing]{setspace}
\usepackage{authblk}
\usepackage{amsthm,amsmath,amsfonts,amssymb}
\usepackage[authoryear]{natbib}
\usepackage[colorlinks,allcolors=blue]{hyperref}
\usepackage{graphicx}
\usepackage[ruled,vlined,linesnumbered]{algorithm2e}
\usepackage{bbm}
\usepackage{thmtools}
\usepackage[nameinlink]{cleveref}
\usepackage{subcaption} 
\usepackage[normalem]{ulem}
\usepackage{xparse}
\usepackage{xcolor}
\usepackage{tabularx,booktabs}
\usepackage{multirow}
\usepackage[section]{placeins} 

\SetAlFnt{\footnotesize\linespread{1.2}\selectfont}
\SetNlSty{textnormal}{\scriptsize}{}
\SetKwInOut{Input}{Input}
\SetKwInOut{Output}{Output}

\SetCommentSty{mycommfont}
\newcommand{\algoStepSep}{0.4em}
\renewcommand{\DontPrintSemicolon}{%
  \def\;{\relax\ifhmode\unskip\fi\par\vspace{\algoStepSep}}%
}
\renewcommand{\PrintSemicolon}{%
  \def\;{;\relax\ifhmode\unskip\fi\par\vspace{\algoStepSep}}%
}

\theoremstyle{plain}

\newtheorem{theorem}{Theorem}[section]

\theoremstyle{remark}
\newtheorem{definition}[theorem]{Definition}

\usepackage{bm}

\DeclareMathOperator*{\argmin}{arg\,min}

\newcommand{\iid}{\stackrel{\mbox{\scriptsize{iid}}}{\sim}}
 
\makeatletter
\newcommand{\multiline}[1]{%
  \begin{tabularx}{\dimexpr\linewidth-\ALG@thistlm}[t]{@{}X@{}}
    #1
  \end{tabularx}
}
\makeatother
\newcommand{\textqt}[1]{``#1''}
\newcommand{\dd}[1]{\mathrm{d}#1}
\newcommand{\card}[1]{\left\lvert#1\right\rvert}

\title{\textbf{Scalable Bayesian Clustering for Integrative Analysis of Multi-View Data}}
\author[1]{Matteo Gianella}
\author[1]{Rafael Cabral}
\author[1]{Andrew Harris}
\author[1,2]{Maria De Iorio}

\affil[1]{Department of Paediatrics, Yong Loo Lin School of Medicine, National University of Singapore}
\affil[2]{Institute for Human Development and Potential, A*STAR}
\date{October 2, 2026}

\begin{document}

\maketitle


\begin{abstract}
  \noindent In the era of Big Data, scalable and accurate clustering algorithms for high--dimensional data are essential. We present new Bayesian Distance Clustering (BDC) models and inference algorithms with improved scalability while maintaining the predictive accuracy of modern Bayesian nonparametric models. Unlike traditional methods, BDC models the distance between observations rather than the observations directly, offering a compromise between the scalability of distance-based methods and the enhanced predictive power and probabilistic interpretation of model-based methods. However, existing BDC models still rely on performing inference on the partition model to group observations into clusters. The support of this partition model grows exponentially with the dataset's size, complicating posterior space exploration and leading to many costly likelihood evaluations. Inspired by $k$-medoids, we propose using tessellations in discrete space to simplify inference by focusing the learning task on finding the best tessellation centers, or \textqt{medoids}. Additionally, we extend our models to effectively handle multi-view data, such as data comprised of clusters that evolve across time, enhancing their applicability to complex datasets. The real data application in numismatics demonstrates the efficacy of our approach. \bigskip

\end{abstract}

\section{Introduction}
\label{sect:introduction}
Clustering is a fundamental tool for uncovering latent structure in complex data. Modern datasets are increasingly large, high-dimensional, or multi-view, creating challenges for scalability and statistical inference. Distance-based methods such as $k$-means and $k$-medoids are computationally efficient but lack a probabilistic framework and principled uncertainty quantification \citep{hastie2009elements}. In contrast, Bayesian model-based clustering provides flexible modeling and uncertainty assessment, but posterior inference can become computationally prohibitive, particularly for large and multi-view datasets.

Bayesian Distance Clustering (BDC) offers a promising compromise by defining a likelihood on pairwise distances rather than on the observations themselves \citep{duan2021bayesian}. By defining a likelihood on inter-point distances, BDC substantially reduces model complexity, mitigates sensitivity to kernel misspecification, and enables clustering in settings where the original data may be high-dimensional, structured, or only implicitly observed through dissimilarities. This perspective aligns naturally with composite likelihood ideas and connects Bayesian inference with classical distance-based algorithms, offering a principled probabilistic interpretation of widely used heuristic methods \citep{natarajan2024cohesion}. 

This positions Bayesian Distance Clustering as a viable alternative to classical Bayesian nonparametric mixtures in regimes where full likelihood inference is infeasible, while retaining principled uncertainty quantification and model-based interpretability. By modeling pairwise dissimilarities directly, our approach provides greater flexibility in settings where a fully specified generative model for the raw data is unavailable, difficult to justify, or where the data are naturally represented only through dissimilarities. Furthermore, BDC methods are  complementary to recent work on scalable Bayesian nonparametric inference, which largely focuses on distributed computation \citep{lovell2013clustercluster}, data subsampling \citep{bardenet2017markov}, and multi-stage or consensus Monte Carlo algorithms to approximate otherwise intractable posteriors  \citep{ni2020scalable}. While these approaches improve scalability by modifying the inference procedure, BDC achieves scalability primarily through model reformulation, replacing high-dimensional likelihood evaluations with structured distance-based components that substantially reduce computational burden without sacrificing a coherent probabilistic interpretation. This flexibility makes Bayesian Distance Clustering particularly well-suited to contemporary applications in areas such as digital humanities, genomics, and integrative biomedical analysis.

We propose a scalable Bayesian Distance Clustering framework for multi-view data, combining within-cluster cohesion and repulsion to improve cluster identifiability. A structured distance-based prior enables information sharing across views, while medoid-based tessellations provide computationally efficient inference that scales favorably with both sample size and number of views. A crucial aspect of BDC is eliciting a prior on the partition $\mathcal{T} = \{\mathcal{C}_1, \ldots, \mathcal{C}_K\}$ of  $N$ objects in $K$  clusters. Popular choices in the Bayesian nonparametric literature involve prior on partitions implied by a species sampling model where the probability of \(\mathcal{T}\) is given by \(p(n_1, \ldots, n_K)\), where \(n_k = |\mathcal{C}_k|\) denotes the number of elements in \(\mathcal{C}_k\). The function \(p\) is a symmetric probability density function of its arguments and is known as the exchangeable partition probability function (EPPF). %

However, the previous models for the partition $\mathcal{T}$ have a large support which grows according to the Bell number of $N$, making posterior exploration through MCMC difficult and leading to inference algorithms that are not scalable. We build upon the BDC model in \cite{natarajan2024cohesion} and introduce new partition models for BDC with accompanying inference algorithms. We take inspiration from  $k$-medoids \citep{kaufman2009finding} and infer probabilistically which observations are the \textqt{medoids} instead of performing inference on the random partition model $\mathcal{T}$ directly. In more detail, each medoid is associated with a cluster, and then the clusters are defined by assigning the objects that are closest to the respective medoid in terms of some distance. This procedure improves computational efficiency since inference on the medoids is often easier and more scalable. This is due to the fact that medoids are selected from the observed data and that the induced space of medoid configurations is typically far more constrained than the space of all possible partitions, resulting in simpler posterior exploration and improved computational efficiency. It also permits the specification of likelihoods on the distance matrix that can be evaluated with linear computational complexity instead of quadratic complexity. Our simulations show that the new models yield comparable, sometimes better, predictive performance while significantly reducing computational time compared to the BDC models of \cite{natarajan2024cohesion}, while also allowing for uncertainty quantification on $K$, the number of clusters. On the other hand, the predictive performance is considerably higher than that of $k$-medoid implementations, especially in scenarios where the data is noisy, and there is overlap in the cluster groups. We discuss the relationship with $k$-medoids by showing that the $k$-medoid estimate can be seen as a maximum a posteriori estimate of the class of BDC models with the tessellation priors we propose. The proposed use of medoids bears some conceptual resemblance to the Anchored Bayesian Gaussian mixture model \citep{kunkel2020anchored}. Both approaches address identifiability and scalability challenges in Bayesian clustering by breaking exchangeability; however, anchored mixture models operate within a likelihood-based mixture framework that models the data directly and enforce identifiability by fixing the labels of a small number of \textit{anchor} observations. In contrast, our method is distance-based and achieves identifiability and scalability by shifting inference away from the full partition space to a lower-dimensional representation defined by medoids and distance-induced tessellations, resulting in substantial computational gains. The AddiVortes method proposed by \cite{stone2025addivortes} uses additive Voronoi tessellations in the context of nonparametric Bayesian regression, where the covariate space is partitioned to approximate a conditional mean function. In contrast, our work develops a distance-based Bayesian unsupervised clustering framework in which Voronoi tessellations over medoids are used to define random partitions and perform inference directly on cluster structure.

The second main contribution is introducing BDC models for multi-view data to cluster objects that contain several features, where there is a possibly distinct partition for each feature,  for instance, when clustering changes at different measurement times \citep[see][and references therein]{duan2020latent,franzolini2026conditional}. These data are frequently found in applications, and we extend our  BDC construction to model the distance matrices of all features jointly, accounting for dependencies between the partition of each feature. %

This paper is structured as follows. In Section \ref{sect:BTM}, we describe  BDC and introduce partition models defined through tessellations in discrete space. In Section \ref{sect:multiview}, we present BDC for multi-view data. Section \ref{sect:simulations} presents the main simulation study, Section \ref{section:numismatics} applies the methodology to digital numismatics, and  in Section \ref{sect:discussion}, we conclude the paper. 

\section{Bayesian distance clustering with tessellation priors}
\label{sect:BTM}
\cite{duan2021bayesian}, \cite{rigon2023generalized}, and \cite{natarajan2024cohesion} introduce a new class of Bayesian Distance Clustering (BDC) models to address the previous issues. Let $\mathcal{T} = \{\mathcal{C}_1,\dotsc, \mathcal{C}_K\}$ denote the partition of $N$ objects into $K$ clusters, where $\mathcal{C}_k$ is a set containing the indices of the objects in cluster $k$. BDC consists in eliciting a likelihood on the distance matrix $\mathbf{D} = [d(x_i,x_j)]_{i,j}$ for some distance metric $d$ instead of eliciting a likelihood on the data $\mathbf{x}=\{x_1,\dotsc, x_N\}$ directly.  The choice of distance metric is problem-dependent and should reflect the structure and objectives of the specific application. By modeling pairwise dissimilarities rather than the observations directly, BDC avoids specification of a complete cluster-specific distribution for the original data. Consequently, clustering is less sensitive to misspecification of features such as cluster shape, skewness, tails and covariance structure. Robustness here is therefore understood specifically as reduced sensitivity to misspecification of the within-cluster data-generating distribution; it does not imply robustness to all modeling choices, since inference may still depend on the chosen distance metric and on the model adopted for the distances \citep{duan2021bayesian}.

\cite{natarajan2024cohesion} specify a likelihood for pairwise distances as follows
\begin{equation}\label{eq:mainequation}
    \pi\left(\mathbf{D} \mid \bm{\theta}, \bm{\lambda}, \mathcal{T}\right) = \underbrace{\prod_{k=1}^K \prod_{\substack{i, j \in \mathcal{C}_k \\ i<j}} f\left(D_{ij} \mid \lambda_k\right)}_{\text{cohesion}} \ \underbrace{\prod_{(k,t) \in A} \prod_{\substack{i \in \mathcal{C}_k \\ j \in \mathcal{C}_t}} g\left(D_{ij} \mid \theta_{kt}\right)}_{\text{repulsion}}
\end{equation}
where  $A = \{(k,t): 1 \leq k < t \leq K\}$, and $f$ and $g$ are probability densities. The likelihood incorporates cohesion and repulsion terms to ensure clusters are composed of objects with small internal dissimilarities (cohesion) and clear separation from other clusters (repulsion). The repulsion term also enhances cluster identifiability by penalizing partitions in which clusters are weakly separated in the distance space. Distance-based likelihoods based solely on within-cluster cohesion can suffer from identifiability issues, particularly in high-dimensional settings, as distinct partitions may yield similar within-cluster distance distributions and lead to ambiguous inference on the number of clusters. By incorporating a repulsive component based on cross-cluster distances, the likelihood explicitly favors configurations in which clusters are not only internally cohesive but also well separated from one another, sharpening posterior modes and discouraging poorly separated solutions. See also Section~\ref{subsec:sensitivity_analysis} of Supplementary Material.

\subsection{Model specification}
\label{sect:likelihood}
The first term of the likelihood in \eqref{eq:mainequation} should favor clusters composed of objects with small dissimilarities among themselves (cohesion),  while the second term should favor large dissimilarities to objects in other clusters (repulsion). Following \cite{natarajan2024cohesion}, we adopt Gamma specifications for the cohesion and repulsion functions, motivated by their connection with pairwise distances under multivariate Normal models and by their computational convenience. More in detail, we set $f$ in \eqref{eq:mainequation} to be the density function of a $ \text{Gamma}(\delta_1,\lambda_k)$ distribution, with $\delta_1<1$, so that the density is decreasing and thus penalizing large distances among objects in the same cluster. Conversely, to achieve the repulsion behavior, we set $g$ in \eqref{eq:mainequation} to be the density function of a
$\text{Gamma}(\delta_2,\theta_{kt})$, with $\delta_2 > 1$ so that the $g(0 \mid \theta_{kt}) = 0$ and increasing near 0 in order to penalize small distances among objects in different clusters. Further, for computational convenience, we elicit conjugate priors, $\lambda_k \iid \operatorname{Gamma}(\mu,\beta)$ and $\theta_{kt} \iid \operatorname{Gamma}(\zeta,\gamma)$ where $\lambda_k$ is a rate parameter of $f$ specific for each cluster and $\theta_{kt}$ is the rate parameter of $g$ specific for each pair of clusters. These parameters control the dispersion in the pairwise distances of objects in the same cluster and objects in different pairs of clusters, respectively.

To set the hyperparameters $\delta_1$, $\delta_2$, $\mu$, $\beta$, $\zeta$ and $\gamma$, we follow the empirical Bayes procedure laid out in Algorithm 1 of \cite{natarajan2024cohesion}. 
Let $\boldsymbol{\lambda} = (\lambda_1,\ldots,\lambda_K)$ and $\boldsymbol{\theta}=(\theta_{kt}, 1 \leq k < t \leq K)$. Finally, we can integrate out $\lambda_k$ and $\theta_{kt}$ from $\pi\left(\mathbf{D} \mid \bm{\theta}, \bm{\lambda}, \mathcal{T}\right)$ and calculate a closed-form expression for the marginal likelihood $\pi\left(\mathbf{D} \mid \mathcal{T} \right)$ conditionally  on the partition $\mathcal{T}$. Performing inference directly on $\pi\left(\mathbf{D} \mid \mathcal{T}\right)$  speeds up computations by avoiding estimating $\lambda_k$ and $\theta_{kt}$. The density $\pi\left(\mathbf{D} \mid \mathcal{T} \right)$ is derived in Section \ref{sect:marginallik} of Supplementary Material. 

The model is completed by specifying a prior on the partition $\mathcal{T}$. A common choice is the EPPF of the Dirichlet process \citep{ferguson1973bayesian}. We will compare the new partition models introduced in the next section with the EPPF derived from the Pitman-Yor (PY) process  \citep[see][]{pitman1995exchangeable}. The EPPF of the PY process is
\begin{equation}\label{eq:EPPF}
    \pi(n_1, \ldots, n_K) = \frac{[M+\theta]_{K-1 ; \theta}}{[M+1]_{N-1;1}} \prod_{i=1}^K[1-\varphi]_{n_i-1;1}
\end{equation}
where for real numbers $x$ and $a$ and non-negative integer $m$, we define $[x]_{m ; a}$ as $[x]_{m ; a} = 1$ for $m=0$, and $[x]_{m ; a} = x(x+a) \dots (x+(m-1)a)$ for $m = 1,2,\dots$.

In \eqref{eq:EPPF}, \(K\) denotes the number of clusters in \(\mathcal{T}\), \(M > 0\) is a concentration parameter, and \mbox{\(0 \leq \varphi < 1\)} is a discount parameter. When \(\varphi = 0\), the Pitman-Yor process simplifies to the Dirichlet process. Performing inference in BDC with the previous partition priors or the exchangeable sequence of cluster partition prior proposed in \cite{natarajan2024cohesion} can lead to issues in terms of scalability and poor mixing and convergence for high-dimensional data. To overcome these challenges, we introduce a prior on the partition defined through an extension of Voronoi tessellation to discrete space. %

\subsection{Voronoi tessellation models in discrete space}
\label{sect:tessellation}
Voronoi tessellation as defined in \cite{denison2002bayesian2} splits a continuous space $\mathcal{X}$ into $K$ regions, with tessellation centers $t_1,\dotsc, t_K$. Points are assigned to a region $R_i$ based on their \textqt{distance} from the centers, so that $R_i = \{x \in \mathcal{X}: d(x,t_i) < d(x,t_j) \text{ for all } j\neq i \}$, where $d$ is a predefined distance metric defined for all points $x_1,x_2 \in \mathcal{X}$. See Figure \ref{fig:subpartition} for a visual illustration where $d$ is the Euclidean distance.

We extend the previous definition of Voronoi tessellation to discrete space. Unlike tessellations in continuous space, where any point in $\mathcal{X}$ could be a tessellation center, we assume that the tessellation centers must correspond to one of the $N$ objects, thus easing the search for the optimal tessellation centers. The remaining objects are assigned to clusters based on the pairwise distances with the tessellation centers. We also refer to the tessellation centers $t_i$ as \textqt{medoids} to highlight similarities with the standard $k$-medoids clustering. Let $\bm{\gamma} = (\gamma_1, \dots, \gamma_K) \in \mathcal{G}_{K}$ be the vector containing the distinct indices of the $K$ medoids, where $\mathcal{G}_{K} = \{\gamma_k \in [N],\ k = 1,\dots, K: \gamma_l \neq \gamma_m \text{ for } l \neq m \}$ and $[N]=\{1,\dots,N\}$. For example, if $\bm{\gamma} = (5,10,20)$, then the objects with indices $5$, $10$ and $20$ are the medoids. Similarly to Voronoi tessellations in continuous space, we assign all the objects nearest to the $k$-th medoid $\gamma_k$ to cluster $k$:
\begin{align}
  \mathcal{C}_k(\mathbf{D}, \bm{\gamma}) &\ = \left\{i \in [N]: \argmin_{l \in [K]} D_{i\gamma_l} = k\right\} \label{eq:gamma}
  \intertext{which define the partition model}
  \mathcal{T}(\mathbf{D}, \boldsymbol{\gamma}) &= \{\mathcal{C}_1(\mathbf{D}, \boldsymbol{\gamma}), \mathcal{C}_2(\mathbf{D}, \boldsymbol{\gamma}),\dots, \mathcal{C}_K(\mathbf{D}, \boldsymbol{\gamma}) \} \notag
\end{align}
We emphasize that the partition $\mathcal{T}$ is a deterministic function of both the medoid set $\boldsymbol{\gamma}$ and the distance matrix $\mathbf{D}$, $\mathcal{T}(\mathbf{D}, \boldsymbol{\gamma})$, but to simplify notation, we will often refer to it as $\mathcal{T}$. We also assume that the minimization in \eqref{eq:gamma} always has a unique solution. This assumption is made for simplicity and notational convenience, and any ties can be resolved through probabilistic assignment without affecting the proposed methodology.

Since the partition $\mathcal{T}$ is deterministic given $\bm{\gamma}$ and $\mathbf{D}$, we elicit the likelihood in Section~\ref{sect:likelihood} conditionally on the medoid set $\boldsymbol{\gamma}$. With abuse of notation, we let $\pi(\mathbf{D} \mid \bm{\gamma}) = \pi(\mathbf{D} \mid \mathcal{T}(\mathbf{D},\bm{\gamma}))$. We elicit a prior distribution on $\bm{\gamma}$ and perform inference on $\bm{\gamma}$ through the posterior density
\begin{equation}\label{eq:tessellation}
  \pi(\bm{\gamma} \mid \mathbf{D}) = \frac{\pi(\mathbf{D} \mid \bm{\gamma})\pi(\bm{\gamma})}{\int \pi(\mathbf{D} \mid \bm{\gamma})\pi(\bm{\gamma}) \mathrm{d}\bm{\gamma}}
\end{equation}
Lastly, when generating posterior samples of $\boldsymbol{\gamma}$, we also automatically obtain posterior samples of the partition $\mathcal{T}(\mathbf{D}, \boldsymbol{\gamma})$ which informs about the clustering structure of the data.

Performing inference on the medoid set $\boldsymbol{\gamma}$ instead of the partition model $\mathcal{T}$ leads to several computational advantages. First, the sample space being explored is much smaller: given $N$ objects, there are $2^{N}-1$ different allowed sets for  $\boldsymbol{\gamma}$ since each object can either be a medoid or not, and we exclude the case where no object is a medoid. On the other hand, in typical random partition models, the number of possible partitions for a set of size $N$ is the Bell number $B_N$, which grows much faster than $2^{N}-1$. Secondly, there are efficient and easy-to-implement algorithms to perform inference on these models. In this paper, we implement a birth-death algorithm, where at each step of the Metropolis-Hastings algorithm, we generate a new $\boldsymbol{\gamma}$ proposal by adding or removing a medoid in the medoid set, followed by a move step where a medoid is changed at random keeping $K$ fixed. It is also possible to consider a Gibbs sampler, and these algorithms are detailed in Section \ref{sect:Alg2} of Supplementary Material. %

Lastly, this formulation also allows for alternative likelihood specifications, which are more computationally efficient to evaluate. To improve scalability, we modify the cohesion term of \eqref{eq:mainequation}, and instead of modeling all pairwise distances between objects in the same cluster, we only model the distances between the objects in a cluster and the cluster medoid. Likewise, for the repulsion term, we only model the distances between the medoids, instead of all pairwise distances between objects in different clusters. The new likelihood is
\begin{equation}
  \label{eq:linear_simple}
  \pi^{\text{lin}}\left(\mathbf{D} \mid \bm{\gamma},\bm{\lambda}\right) = \underbrace{\prod_{k = 1}^{K} \prod_{i \in \mathcal{C}_k(\mathbf{D}, \bm{\gamma}) \setminus \{\gamma_k\}} f \left(D_{i\gamma_k} \mid \lambda_k\right)}_{\text{cohesion}} \ \underbrace{\prod_{(k,t) \in A\vphantom{\setminus}} g(D_{\gamma_k\gamma_t})}_{\text{repulsion}}
\end{equation}
The computational complexity of evaluating the previous likelihood is linear in $N$ (hence the superscript \textqt{lin} for linear), unlike \eqref{eq:mainequation}, which is quadratic. Similarly to \eqref{eq:mainequation}, we can also integrate out $\boldsymbol{\lambda}$, and express the likelihood conditionally only on $\boldsymbol{\gamma}$ (and the hyperparameters $\delta_1,\delta_2,\mu,\beta$ and $\theta$). The likelihood is given by
\begin{equation*}
  \pi^{\text{lin}}\left(\mathbf{D} \mid \boldsymbol{\gamma} \right) = \prod_{k = 1}^{K}\mathcal{L}_{k}^{(1)} \prod_{(k,t) \in A} \mathcal{L}_{kt}^{(2)}
\end{equation*}
where
\begin{align*}
  \mathcal{L}_{k}^{(1)} &= \frac{\Gamma(\delta_1)^{-n_k^\star} \Gamma(\mu+n_k^\star\delta_1)}{\Gamma(\mu)}\ \beta^\mu \left(\prod_{\ i \in \mathcal{C}_k(\mathbf{D}, \bm{\gamma}) \setminus \{\gamma_k\}} D_{i\gamma_k} \right)^{\delta_1-1} \left( \beta + \sum_{i \in \mathcal{C}_k(\mathbf{D}, \bm{\gamma}) \setminus \{\gamma_k\}}  D_{i\gamma_k}\right)^{-\mu-\delta_1 n_k^\star} \\ 
  \mathcal{L}_{kt}^{(2)} &= \text{Gamma}(D_{\gamma_k\gamma_t}; \delta_2,\theta)
\end{align*}
In the above equation, $n_k^\star = \lvert\mathcal{C}_k(\mathbf{D},\bm{\gamma})\rvert - 1$ and $\text{Gamma}(\cdot; \delta_2, \theta)$ is the density function of a Gamma distribution with shape parameter $\delta_2$ and rate parameter $\theta$. The parameter $\lambda_k$ controls the cohesion of each cluster. However, for the repulsion term, there is only a single hyperparameter, $\theta$, governing the overall repulsion between cluster centers. Section~\ref{sect:hyperpar_select} of Supplementary Material describes a procedure in the spirit of empirical Bayes, similar to the one in \cite{natarajan2024cohesion}, to select the hyperparameters $\delta_1,\delta_2,\mu,\beta$, and $\theta$. The empirical Bayes approach is used to calibrate hyperparameters governing the scale of the distance-based likelihood, as fully eliciting informative priors is particularly challenging in high-dimensional settings. This strategy improves numerical stability and posterior mixing while preserving the modeling flexibility \citep{kass1989approximate,carlin2000empirical}. 

\subsection{Prior on the medoid set}\label{section:priors}
We consider prior distributions on $\bm{\gamma}$ that reflect only the prior belief about the number of clusters $K$, assuming a uniform prior on the possible centers given $K$. These priors take the form $\pi(\bm{\gamma},K) = \binom{N}{K}^{-1} \pi(K),  \ 0 < K = \lvert\bm{\gamma}\rvert \leq N$, where $\pi(K)$   is the prior distribution on the number of clusters. In our context, it is natural to elicit a prior that penalizes the number of clusters, such as
\begin{equation}\label{eq:prior}
    \pi(\bm{\gamma} \mid K)\,\pi(K) = \binom{N}{K}^{-1} \times \text{TGeom}(K;p) \qquad 0 < K = \lvert\bm{\gamma}\rvert \leq N
\end{equation}
where $\text{TGeom}(\cdot;p)$ is the probability mass function of a truncated geometric distribution with parameter $p$ and support between $1$ and $N$, inclusive. The first term corresponds to a uniform prior over all possible choices of $K$ medoids from the $N$ objects, and the second term penalizes the number of medoids since it decreases with $K$. The hierarchical prior in \eqref{eq:prior} decouples the combinatorial selection of medoids from the prior on the number of clusters $\pi(K)$, yielding a modular formulation in which assumptions on cluster complexity and representative units are separated. In the analyses presented here, we adopt a truncated geometric prior $K$, which provides increased shrinkage towards smaller values of $K$ relative to a Poisson prior while retaining heavier tails. Conditional on $K$, a uniform prior is placed on the medoid set, reflecting lack of prior information on the choice of representative observations. This component may be replaced by alternatives that explicitly regulate cluster sizes, such as size-biased priors \citep{betancourt2022random}, or by joint priors on $\pi(K)$ induced, for example,  by exchangeable partition probability functions arising from Bayesian nonparametric constructions \citep{pitman1996some,pitman2006combinatorial}. While such priors typically favor a small number of large clusters, the proposed framework also accommodates priors centered on expert-elicited partitions \citep{paganin2021centered,paganin2025informed}, including informed  \citep{paganin2025informed} or anchor partition priors \citep{dahl2025dependent}, when such information is available.

\subsection{Computational complexity}
Table \ref{tab:comp_complexity} outlines the worst-case computational complexity of performing a likelihood evaluation based on \eqref{eq:mainequation} and \eqref{eq:linear_simple} and considering $K$ fixed. Additionally, we consider the scenario where only the cohesion term is present. In Big Data applications, a computational complexity of $\mathcal{O}(N^2)$ can be prohibitive, especially since such likelihood evaluations must be performed thousands or even millions of times during MCMC inference. In these cases, distance-based clustering based on \eqref{eq:linear_simple} is more desirable due to its linear computational complexity (assuming $K \ll N$ when considering the repulsion term).

\begin{table}[ht]
  \centering\small
  \begin{tabular}{ccc}
    \bottomrule[2pt]
    & Cohesion  & Cohesion + Repulsion \\ \midrule[2pt] 
    \textqt{Quadratic} Likelihood in \eqref{eq:mainequation} & $\mathcal{O}(N^2)$ & $\mathcal{O}(N^2)$ \\ 
    \textqt{Linear} Likelihood in \eqref{eq:linear_simple} & $\mathcal{O}(N)$ & $\mathcal{O}(N+K^2)$ \\ \bottomrule[2pt] 
    \end{tabular}
    \caption{Computational complexity of the likelihood evaluations in \eqref{eq:mainequation}  and \eqref{eq:linear_simple}, with and without the repulsion term.} %
    \label{tab:comp_complexity}
\end{table}

In the birth-death algorithm (Section \ref{sect:Alg2} of Supplementary Material), beside performing a likelihood evaluation at each MCMC iteration, we also need to find the tessellation configuration given a set of centers as in \eqref{eq:gamma}, and so the computational complexity of each MCMC iteration is \mbox{$\mathcal{O}(NK + L)$}, where $L$ is the computational complexity of one likelihood evaluation given in Table \ref{tab:comp_complexity}. In the Gibbs sampler MCMC algorithm, we cycle through each object, and for each object, we perform two tessellations and two likelihood evaluations, so the computational complexity per iteration is $\mathcal{O}(N(NK+L))$. On the other hand, the computational complexity of the $k$-medoids algorithm PAM (Partitioning Around Medoids) is $\mathcal{O}(N(N-K)^2)$. More scalable $k$-medoids algorithms are available, including CLARA and CLARANS \citep{schubert2019faster}, but they perform considerably worse when the data are noisy, and the clusters overlap, as shown in Section \ref{sect:simulations}.

We note that the stated complexity $\mathcal{O}(N+K^2)$ refers to the expected amortized cost per MCMC iteration, since the sampler induces only local updates to the induced partition $\mathcal{T}$ through birth–death proposals, allowing distance evaluations and assignments to be cached and reused across iterations. Although a full recomputation of $\mathcal{T}$ would incur a worst-case cost of $\mathcal{O}(KN)$, such global reassignment steps are not performed at each iteration and arise only in rare, degenerate scenarios; consequently, the dominant expected per-iteration cost remains $\mathcal{O}(KN)$.

\subsection[Relationship with k-medoids]{Relationship with $k$-medoids}
Traditional $k$-medoids clustering algorithms also aim at identifying which objects should be the \textqt{medoids}. For a review see \cite{kaufman2009finding} and \cite{jain2010data}. Formally, the $k$-medoids estimate for the vector of distinct indices $\bm{\gamma}$  is, for a fixed $K$, the solution of the following optimization problem:
\begin{equation*}
  \hat{\bm{\gamma}} = \argmin_{\bm{\gamma}} \sum_{k=1}^{K} \sum_{i \in \mathcal{C}_k(\mathbf{D},\bm{\gamma}) \setminus \{\gamma_k\}} D_{i\gamma_k}
\end{equation*}
We assign all the non-medoid objects nearest to medoid $\gamma_k$ to cluster $k$, also according to \eqref{eq:gamma}, and the optimization consists of searching for the set of $K$ medoids that minimize the distance between each medoid $\gamma_k$ and the non-medoid objects assigned to cluster $k$. The $k$-medoids estimate is the maximum a posteriori estimate of the BDC model in \eqref{eq:linear_simple} for specific choices of $f$, $g$ and prior $\pi(\bm{\gamma})$. More in detail, let us consider the density $f\left(D_{i\gamma_k} \mid \lambda_k\right) \propto \exp(-D_{i\gamma_k})$, ignore the repulsion term, and consider the prior
$\pi(\bm{\gamma}) \propto \mathbb{I}\{\lvert\bm{\gamma}\rvert = K\}$ (where $\mathbb{I}\{\cdot\}$ denotes the indicator function), so that the number of medoids is fixed to $K$. Then, the $k$-medoids estimate is the mode of the posterior density
\begin{equation*}
  \pi\left(\bm{\gamma} \mid \mathbf{D} \right) \propto \exp\left\{- \sum_{k = 1}^{K}\sum_{i \in \mathcal{C}_k(\mathbf{D}, \bm{\gamma}) \setminus \{\gamma_k\}} D_{i\gamma_k} \right\} \mathbb{I}\{\lvert\bm{\gamma}\rvert = K\}
\end{equation*}
where the resulting partition estimate is given by $\hat{\mathcal{T}}(\mathbf{D}, \hat{\boldsymbol{\gamma}} ) = \{\mathcal{C}_1(\mathbf{D}, \hat{\boldsymbol{\gamma}} ), \dotsc, \mathcal{C}_K(\mathbf{D},\hat{\boldsymbol{\gamma}} )\}$. The advantage of our Bayesian formulation is the ability to quantify uncertainty on the partition and number of clusters. Also, as shown in the simulations, our MCMC algorithms are more robust in locating the modal region. %

The robustness of the proposed tessellation-based model closely mirrors that of $k$-medoids as cluster representatives are restricted to observed data points and assignments depend on distances rather than squared deviations. As a result, isolated or extreme observations are unlikely to attract many assignments and typically form small or low-weight clusters, limiting their global influence on the partition. The Bayesian formulation further regularizes the clustering through priors on the partition and the number of clusters, naturally down-weighting outliers in the posterior while retaining the practical robustness of distance-based methods together with uncertainty quantification.

\section{Bayesian distance clustering for multi-view data}\label{sect:multiview}

Datasets frequently present multivariate information often collected across distinct domains and with different feature support spaces. An example is in longitudinal data analysis, where the underlying clustering structure of individuals is likely to change over time. Data of this type are often referred to as \textit{multi-view data}, and the number and shapes of clusters may need to vary across features to accurately capture dependence structures and heterogeneity \citep{franzolini2026conditional}. However, existing models for clustering multi-view data are based on modeling the data directly rather than the distance matrix, and here, we examine the latter. To this end, we consider $N$ objects observed across $V$ distinct \textit{feature spaces} or \textit{views}. This means that there are $V$ distance matrices $\mathbf{D}^{(v)} \in \mathbb{R}^{N \times N}$ for $v=1,\dotsc,V$.

Previous works, such as \cite{caron2017generalized}, \cite{page2022dependent}, and \cite{deiorio2023bayesian}, focus on modeling dependencies between random partition models. Since we model the medoid set directly rather than the partition, the models in the previous papers cannot be applied directly. In this section, we present two approaches to model dependencies between partitions in BDC with partitions defined through tessellations. The first framework assumes a strict hierarchical ordering across features and we call it \textit{nested tessellation}; the second is a flexible, \textit{joint tessellation} framework.

For each view $v \in [V]$, we define the associated vector of distinct medoid indices as $\bm{\gamma}_v = (\gamma_{v1}, \dots, \gamma_{vK_v}) \in \mathcal{G}^{(v)}_{K_{v}}$, where $\mathcal{G}^{(v)}_{K_{v}} = \{\gamma_{vk} \in [N],\ k = 1, \dots, K_v :  \gamma_{vl} \neq \gamma_{vm} \text{ for } l \neq m \}$ and $K_v$ is the number of clusters in view $v$. Each medoid set $\bm{\gamma}_v$ induces a partition $\mathcal{T}_{v} = \{\mathcal{C}_{v1}, \dots, \mathcal{C}_{vK_v}\}$, where $\mathcal{C}_{vk}$ is the $k$-th cluster of view $v$. Every cluster $\mathcal{C}_{vk}$ is defined similarly to \eqref{eq:gamma} by grouping together objects that are closest to the $k$-th medoid $\gamma_{vk}$ according to the distance matrix $\mathbf{D}^{(v)}$. We assume that the distance matrices $\mathbf{D}^{(1)}, \dots, \mathbf{D}^{(V)}$ are conditionally independent given the medoid sets $\bm{\gamma}_1, \dots, \bm{\gamma}_V$. This leads to the following joint posterior distribution:
\begin{equation}\label{eq:joint_partition}
    \pi(\bm{\gamma}_1,\dots,\bm{\gamma}_V \mid \mathbf{D}^{(1)}, \dots, \mathbf{D}^{(V)}) \propto \left[\ \prod_{v=1}^V \pi(\mathbf{D}^{(v)} \mid \mathcal{T}_v) \right] \pi(\bm{\gamma}_1, \dots, \bm{\gamma}_V)
\end{equation}
where $\pi(\boldsymbol{\gamma}_1, \dots, \boldsymbol{\gamma}_V)$ denotes a joint prior on the $V$ tessellations. The following subsections present two prior specifications for $\pi(\boldsymbol{\gamma}_1, \dots, \boldsymbol{\gamma}_V)$ that account for dependencies in the medoid sets and, consequently, in the induced partitions.

\subsection{Nested tessellation models}
\label{section:nested_partition}
Let $\bm{z}_v = (z_{v1}, \dots, z_{vN})$ be the cluster allocation vector for each view $v$, where $z_{vi} = k$ if object $i$ belongs to the $k$-th cluster of view $v$ (or, equivalently, $i \in \mathcal{C}_{vk}$). Let us also recall that $\bm{\gamma}_v = (\gamma_{v1}, \dots, \gamma_{vK_v})$ contains the distinct indices of the $K_v$ medoids in view $v$.
\begin{definition}\label{def:nested}
  In a nested partition model across $V$ ordered views, two objects belong to the same cluster in view $v$ only if they belong to the same cluster in view $v-1$:
  \begin{equation*}
    z_{vi} = z_{vj} \implies z_{v-1i} = z_{v-1j} \qquad \forall i,j \in [N],\ v = 2,\dots,V
  \end{equation*}
  while the reverse implication is not assumed.
\end{definition}
This definition leads to a sequence of progressively refined partitions where the number of clusters is monotonically non-decreasing ($K_1 \le K_2 \le \dots \le K_V$), and clusters in view $v$ are sub-clusters of those in view $v-1$ (see Figure~\ref{fig:subpartition} for a representation with $V = 2$ in continuous space). Nested partitions are appropriate when the clustering structure is expected to evolve cumulatively, with clusters refining or splitting over time without disappearing. This assumption is natural in settings such as longitudinal studies, hierarchical biological annotation, and historical or archaeological analyses, where increasing numbers of clusters reflect meaningful scientific refinement. Methodologically, the nested structure facilitates information sharing across views or time points, restricts the space of admissible partitions, and improves identifiability and posterior efficiency. We therefore present nested models as a specific option, while noting that multi-view clustering is preferable in many applications.

\begin{figure}[t]
    \centering
    \includegraphics[width=0.5\linewidth]{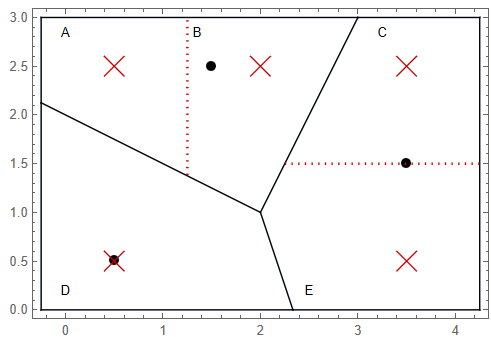}
    \caption{Example of a nested partition structure. The black lines delimit the three regions of the first view. Both the black lines and red dotted lines delimit the five regions of the second view. The black dots represent the medoids in the first view, while the red crosses represent the medoids in the second view.}
    \label{fig:subpartition}
\end{figure}

We now consider a natural extension of the Voronoi tessellation models of Section~\ref{sect:tessellation} that lead to nested partitions. We specify in view 1 the model described in Section~\ref{sect:tessellation}. Given the partition in view 1, we divide the observations in view 2 into groups corresponding to the clusters identified in view 1; that is, observations in view 2 are grouped together if they belong to the same cluster in view 1. Within each of these groups, we then specify a tessellation model as defined in Section~\ref{sect:tessellation}. This procedure is repeated sequentially for each of the remaining ordered views. Section~\ref{sect:Algnested} of Supplementary Material proposes an MCMC algorithm of to perform inference on these models.

\subsection{Joint tessellation models}
\label{section:joint_partition}
The nested tessellation model  defines a sequence of progressively refined partitions across views. Two objects cannot belong to the same cluster in the second view unless they also belong to the same cluster in the first view. Consequently, the number of clusters in the second view is always larger than in the first.  To allow flexible, unconstrained dependence among partitions without forcing nesting, we introduce a joint tessellation model. Intuitively, we aim to include a dependence parameter \(\alpha_{uv} \in [0,1]\) in our model, so that when $\alpha_{uv} = 0$, $\bm{\gamma}_u$ is independent of $\bm{\gamma}_v$, and as $\alpha_{uv}$ increases, the partitions $\mathcal{T}_u$ and $\mathcal{T}_v$ become increasingly similar. The joint prior can be expressed as 
\begin{equation}
  \label{eq:joint_partition_prior}
  \pi\left(\bm{\gamma}_1, \dots,  \bm{\gamma}_V \mid \{\alpha_{uv}\}_{1 \leq u < v \leq V} \right) \propto \Bigg[\prod_{v=1}^V \pi(\bm{\gamma}_{v})\Bigg] \Bigg[ \prod_{1\leq u < v \leq V} h(\bm{\gamma}_u,\bm{\gamma}_v \mid \alpha_{uv})\Bigg]
\end{equation}
where $h(\bm{\gamma}_u, \bm{\gamma}_v \mid \alpha_{uv})$ denotes a first-order interaction term between two medoid sets.

Previous work, such as that by \cite{dahl2025dependent}, introduces the concept of shrinking a random partition toward a reference partition. \cite{paganin2021centered} multiply the partition prior $\pi(\mathcal{T}_1)$ by a penalization term leading to the conditional density $\pi(\mathcal{T}_2 \mid \mathcal{T}_1, \phi) \propto \pi(\mathcal{T}_2) \exp\{-\phi d(\mathcal{T}_2, \mathcal{T}_1)\}$, where $d$ quantifies the distance between partitions. However, these models apply to partitions and not to medoid sets. Also, in the case of the random partitions in \cite{paganin2021centered}, the normalization constant of $\pi(\mathcal{T}_2 \mid \mathcal{T}_1, \phi)$, which is a function of $\mathcal{T}_1$ and $\phi$, is intractable, leading to the necessity of fixing parameters or relying on computationally expensive MCMC algorithms for doubly intractable problems.

The joint tessellation prior in \eqref{eq:joint_partition_prior} resembles priors used in repulsive mixture models, where a standard prior on the locations of the mixture components is augmented by an interaction term that can either encourage or discourage proximity between component locations \citep[see, e.g.,][]{cremaschi2025repulsion}. This construction is inspired by ideas from statistical mechanics, where pairwise interaction terms are used to describe attraction or repulsion between particles in a physical system. We extend this idea from interactions between mixture component locations to interactions between partitions, introducing a term that controls the degree of similarity or dissimilarity among the partitions associated with different views. In our context, we define the interaction term as: 
\begin{equation*}
    h(\bm{\gamma}_u, \bm{\gamma}_v \mid \alpha_{uv}) = \exp \left\{ -\phi(\alpha_{uv}) d(\bm{\gamma}_u, \bm{\gamma}_v) \right\} \quad\text{where}\quad d\left(\bm{\gamma}_u, \bm{\gamma}_v\right) = \frac{1}{\operatorname{RI}\left(\mathcal{T}_u, \mathcal{T}_v\right)} - 1
\end{equation*}
Here, $\operatorname{RI}$ stands for the Rand Index \citep{rand1971objective} and $d(\bm{\gamma}_u, \bm{\gamma}_v)$ is a pairwise partition dissimilarity. When the two induced partitions coincide, $\operatorname{RI}$ is 1, resulting in a dissimilarity of $0$. Conversely, if the partitions do not agree on any object, then $\operatorname{RI}$ is $0$ and the dissimilarity becomes infinite. We employ the transformation $\phi(\alpha_{uv}) = \alpha_{uv} / (1 - \alpha_{uv})$ to ensure that $\alpha_{uv} \in [0,1]$. When $\alpha_{uv} \to 0$, no penalty is induced, whereas larger values of $\alpha_{uv}$ lead to a larger penalization between the induced partitions in views $u$ and $v$. Section~\ref{sect:Algjoint} of Supplementary Material provides the MCMC algorithm to perform inference.

\section{Simulation study}
\label{sect:simulations}
To evaluate the performance of the proposed tessellation models for (multi-view) data, we conduct an extensive Monte Carlo simulation study. All simulations were performed on a dedicated workstation equipped with an AMD Ryzen 9 7950X processor (32 threads, clock speeds up to 5.88 GHz) and 128 GB of RAM.

We generate synthetic datasets with known ground truth cluster structures and systematically vary key parameters to assess the robustness and scalability of the proposed models. Each dataset is generated as follows: we consider multi-view data with $V = 2$ views. For each view, we generate $N = 100$ points from a mixture of $K^\star = 10$ multivariate Gaussian distributions with cluster centers at the vertices of a standard 10-simplex and covariance matrices $\sigma_s^2 I_d$. The cluster allocations for the first view are drawn from a multinomial distribution with size $K^{\star}$ and probability vector sampled from a Dirichlet distribution with concentration $M = 10$. For the second view, we randomly select $\lfloor 100\alpha_s \rfloor$ objects to have the same cluster allocations as the first view. The remaining $100 - \lfloor 100\alpha_s \rfloor$ objects in the second view are assigned to a different cluster through a random permutation of the first view's allocations. The simulation scenarios are generated considering different levels of $(i)$ cluster overlap through the parameter $\sigma_s$ and $(ii)$ inter-view agreement rate through the parameter $\alpha_s$. In particular, we consider three distinct  regimes: Low ($\sigma_s = 0.1$), Moderate ($\sigma_s = 0.2$) and High overlap ($\sigma_s = 0.3$). We also consider five levels of inter-view agreement rate ($\alpha_s \in \{0.00, 0.25, 0.50, 0.75, 1.00\}$). For each combination of these parameters, we generate 50 replicates to ensure statistical robustness.

Across all simulation scenarios, we evaluate two variants of our method: \textit{Tessellation Independent}, which fits the baseline model independently for each view (Section~\ref{sect:BTM}; Algorithm~\ref{algo:2} in Supplementary Material), and \textit{Tessellation Joint}, the joint multi-view formulation (Section~\ref{section:joint_partition}; Algorithm~\ref{algo:joint} in Supplementary Material). We compare our models against: $(i)$ \textit{PAM}: Partitioning Around Medoids algorithm; $(ii)$ \textit{CLARA}: Clustering Large Applications algorithm; $(iii)$ \textit{PY Independent}: Pitman-Yor mixture models fitted independently for each view; and $(iv)$ \textit{PY Joint}: joint multi-view Pitman-Yor mixtures based on prior~\eqref{eq:joint_partition_prior}.

\textit{PAM} and \textit{CLARA} algorithms are fitted through the \texttt{cluster} R package \citep{maechler2019finding}. The MCMC algorithm we adopt to make inference on the Pitman-Yor process mixture models is the split and merge algorithm by \cite{jain2004split}. We alternate between split/merge proposals and standard Gibbs steps to improve chain mixing, as suggested in the paper. We consider independent marginal hyper-priors for the concentration parameter $M$ and the discount parameter $\theta$ in \eqref{eq:EPPF}. In particular, we assume $M \sim \operatorname{Gamma}(a_M, b_M)$ and $\theta \sim \operatorname{Beta}(a_\theta, b_\theta)$. Sampling of the Pitman-Yor hyper-parameters is performed via adaptive Metropolis-Hastings steps using the Robbins-Monro dynamics \citep{garthwaite2016adaptive} and targeting optimal acceptance rates \citep{roberts1997weak}. To guarantee a fair comparison between the proposed tessellation models and the Pitman-Yor process mixture models, we consider the quadratic likelihood specification with repulsion term as baseline likelihood. All MCMC-based models (\textit{Tessellation Independent}, \textit{Tessellation Joint}, \textit{PY Independent}, \textit{PY Joint}) are fitted using 10,000 iterations, with the first 2,000 discarded as burn-in. The number of clusters for the PAM and CLARA algorithms is fixed to the true number of clusters $K^\star = 10$.

We evaluate the performance of each model in the simulation study through several metrics. We consider: $(i)$ the Adjusted Rand Index \citep[ARI, see][]{rand1971objective}, the Normalized Mutual Information \citep[NMI, see][]{mcdaid2011normalized} and the Jaccard Index to assess the accuracy of the cluster configuration w.r.t. the ground truth established at simulation time; $(ii)$ the difference $\Delta K$ between the retrieved number of clusters $\hat{K}$ and the true number of clusters $K^\star$; $(iii)$ the ARI between views as a measure to assess the alignment between the views; $(iv)$ the total computational time as a measure of scalability of the inference algorithms. To compute the ARI, NMI, Jaccard index and $\Delta K$, we first obtain a point estimate of the partition by minimizing the posterior expected variation of information (VI) loss \citep{meilua2007comparing,wade2018bayesian}, as implemented in the \texttt{salso} R package \citep{dahl2022search}, and then compare it with the ground truth.

After this baseline comparative study, we focus on the proposed Bayesian framework and compare alternative likelihood specifications to assess the individual impact of the key modeling choices, detailing possible advantages and disadvantages. In particular, we compare alternative specifications in terms of:
\begin{itemize}
  \item[--] \textbf{Presence / Absence of the repulsion term}: keeping the likelihood specification as in \eqref{eq:mainequation}, we compare the performance of the models with and without the repulsion term in the likelihood function. This allows us to assess the impact of the repulsion term on the identifiability of the cluster structure and the overall performance of the models.
  \item[--] \textbf{Likelihood specification}: keeping the repulsion term in each likelihood specification, we compare the use of the \textqt{quadratic} likelihood $\pi\left(\mathbf{D} \mid \bm{\theta}, \bm{\lambda}, \bm{\gamma} \right)$ in \eqref{eq:mainequation} and the \textqt{linear} likelihood $\pi^{\text{lin}}\left(\mathbf{D} \mid \bm{\lambda}, \bm{\gamma}\right)$ in \eqref{eq:linear_simple}. This comparison can give insights into the trade-off between model accuracy and computational efficiency, as the linear likelihood is expected to be more scalable but less accurate than the quadratic one. 
\end{itemize}
In the next section, we report the results of the comparative method evaluation. Complementary analyses are provided in Supplementary Material: the sensitivity analysis in Section~\ref{subsec:sensitivity_analysis} demonstrates that the repulsion term substantially improves cluster recovery and prevents over-splitting under moderate to high overlap, whereas the linear likelihood provides a computationally efficient alternative when clusters are well separated. Furthermore, the scalability analysis in Section~\ref{subsec:scalability} confirms that the tessellation models scale favorably with both sample size and number of views, incurring substantially lower computational cost than their Pitman-Yor counterparts.

\subsection{Method comparison}
\label{subsec:comparative_study}
In this section, we compare the single-view clustering performance of the proposed \textit{Tessellation Independent} model against classical non-Bayesian algorithms (\textit{PAM}, \textit{CLARA}) and the independent Pitman-Yor process mixture model (\textit{PY Independent}). Figure~\ref{fig:indep_performance} displays, in the top panel, the average Adjusted Rand Index (ARI) across both views as a function of the inter-view agreement rate $\alpha_s$. The dots represent the median ARI across 50 replicates, while the shaded bands indicate the 95\% empirical  quantile intervals. The second panel of Figure~\ref{fig:indep_performance} reports the distribution of the cardinality discrepancy $\Delta K = \frac{1}{2}\sum_{v=1}^2 (\hat{K}_v - K^\star)$ relative to the true cardinality $K^\star = 10$. Both panels are presented faceted horizontally across the three cluster overlap regimes: Low ($\sigma_s = 0.1$), Mid ($\sigma_s = 0.2$) and High ($\sigma_s = 0.3$).

\begin{figure}[t!]
  \centering
  \begin{subfigure}[b]{\linewidth}
    \includegraphics[width=\linewidth]{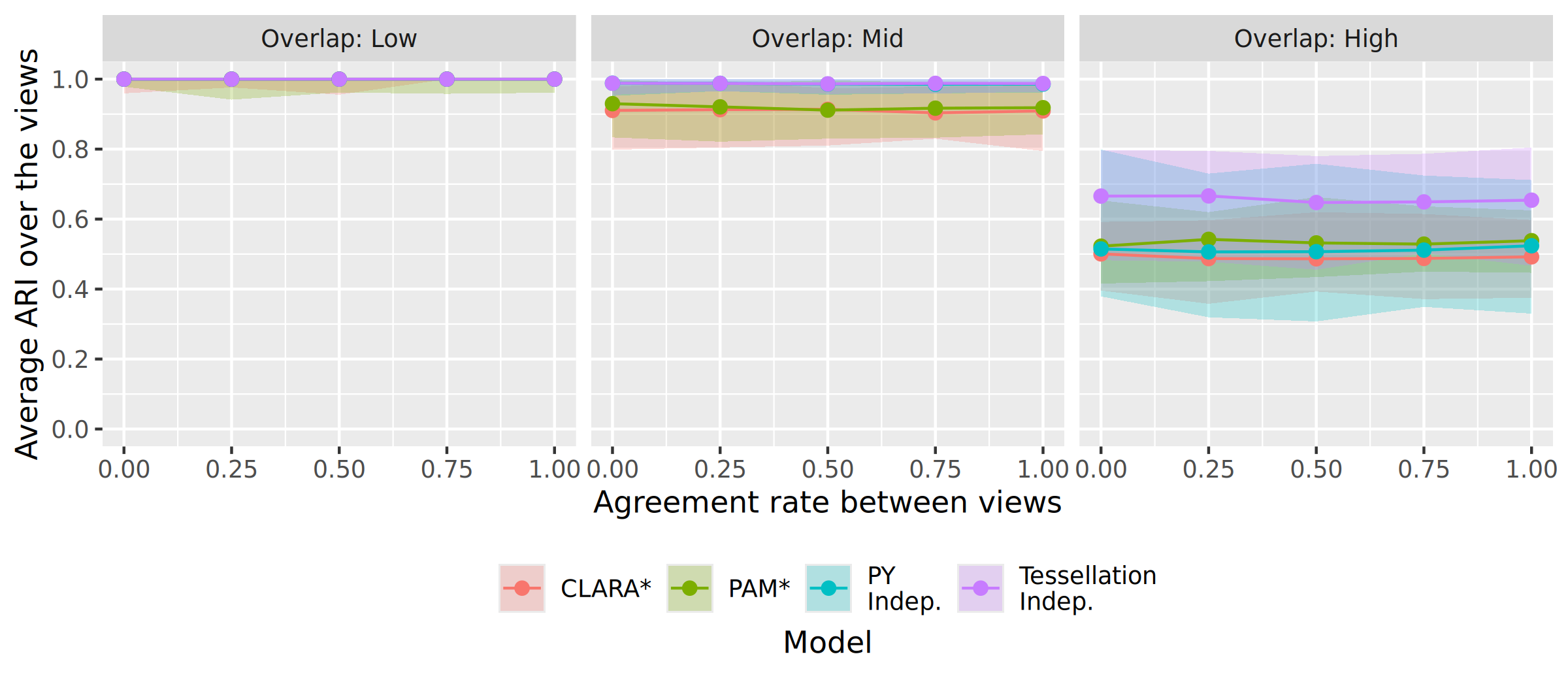}\vspace{-0.4cm}
    \caption{Average Adjusted Rand Index (ARI) across views as a function of $\alpha_s$.}
    \label{fig:mean_ari_indep}
  \end{subfigure}
  
  \vspace{0.3cm}
  
  \begin{subfigure}[b]{\linewidth}
    \includegraphics[width=\linewidth]{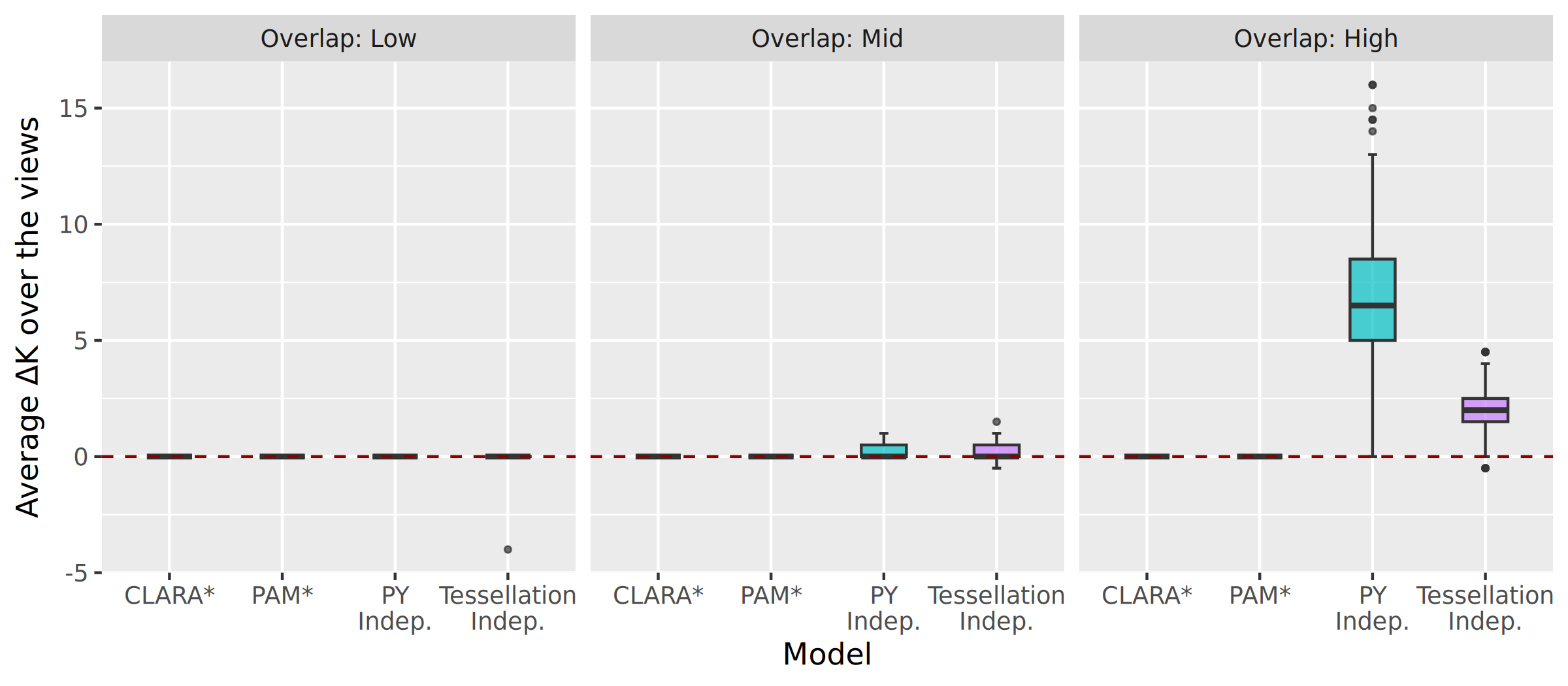}\vspace{-0.4cm}
    \caption{Distribution of average cardinality discrepancy $\Delta K = \frac{1}{2}(\hat{K}_1 - 10) + \frac{1}{2}(\hat{K}_2 - 10)$.}
    \label{fig:clust_recovery_indep}
  \end{subfigure}
  \caption{Performance comparison under independent model specifications across overlap regimes. (a) Median ARI with empirical 95\% bands. (b) Boxplots of cardinality error $\Delta K$ (red dashed line denotes exact recovery). Asterisks ($^\star$) indicate methods supplied with the true cardinality $K^\star = 10$.}
  \label{fig:indep_performance}
\end{figure}

Across all panels of Figure~\ref{fig:mean_ari_indep}, the horizontal trajectories show that performance is invariant to the inter-view agreement rate $\alpha_s$. Moving across the overlap regimes reveals distinct performance behaviors: when clusters are well separated (left panels), all four candidate models achieve near-perfect partition recovery. In addition, both non-parametric Bayesian specifications consistently recover the exact ground truth cardinality ($\Delta K \approx 0$). In the intermediate regime (center panels), the non-parametric Bayesian approaches (\textit{Tessellation Independent} and \textit{PY Independent}) maintain high accuracy with median ARI values exceeding $0.98$. Conversely, the classical algorithms (\textit{PAM} and \textit{CLARA}) exhibit slightly worse performance (median ARI $\approx 0.90$--$0.91$) and wider empirical dispersion, despite being provided with the true number of clusters $K^\star = 10$. Both \textit{Tessellation Independent} and \textit{PY Independent} reliably preserve cardinality calibration with a median $\Delta K \approx 0$ with minimal dispersion. Under considerable cluster overlap (right panels), \textit{Tessellation Independent} outperform all other competitors: while \textit{PY Independent}, \textit{PAM}, and \textit{CLARA} decline in performance (median ARI dropping to approximately $0.49$--$0.53$), the proposed model demonstrates higher performance consistency, sustaining a median ARI of approximately $0.65$ across all settings. Figure~\ref{fig:clust_recovery_indep}  shows that under heavy noise, \textit{PY Independent} overestimates cluster cardinality (median $\Delta K \approx 7$). In contrast, the penalizing prior on the medoid set in \textit{Tessellation Independent} (see \eqref{eq:prior}) effectively discourages redundant clusters, keeping cardinality controlled (median $\Delta K \approx 2$).

To assess the practical feasibility of our inferential framework, Figure~\ref{fig:ari_time_tradeoff} and Table~\ref{tab:benchmark_summary} present the trade-off between clustering accuracy (ARI) and computational runtime (elapsed CPU seconds, on a $\log_{10}$ scale) across all overlap regimes. While \textit{PAM} and \textit{CLARA} algorithms execute in fractions of a second, their clustering performance deteriorates under moderate and high overlap, highlighting the necessity of model-based formulations in noisy settings. Compared to the Pitman-Yor mixture baseline, \textit{Tessellation Independent} achieves both superior estimation accuracy and substantially reduced computational overhead. In well-separated and moderate regimes ($\sigma_s \in \{0.1, 0.2\}$), it matches the high accuracy of \textit{PY Independent} (median ARI $\ge 0.98$) while requiring an order of magnitude less computing time ($\approx 0.4$ s compared to $\approx 8$ s). Under high overlap, the split-and-merge MCMC sampler for \textit{PY Independent} suffers from marked computational slowdown due to the combinatorial burden of exploring over-partitioned cluster states. In contrast, the MCMC sampling scheme for the proposed tessellation model exhibits stable computational cost across all noise levels while simultaneously attaining higher clustering accuracy.

\begin{figure}[t]
  \centering
  \includegraphics[width=\linewidth]{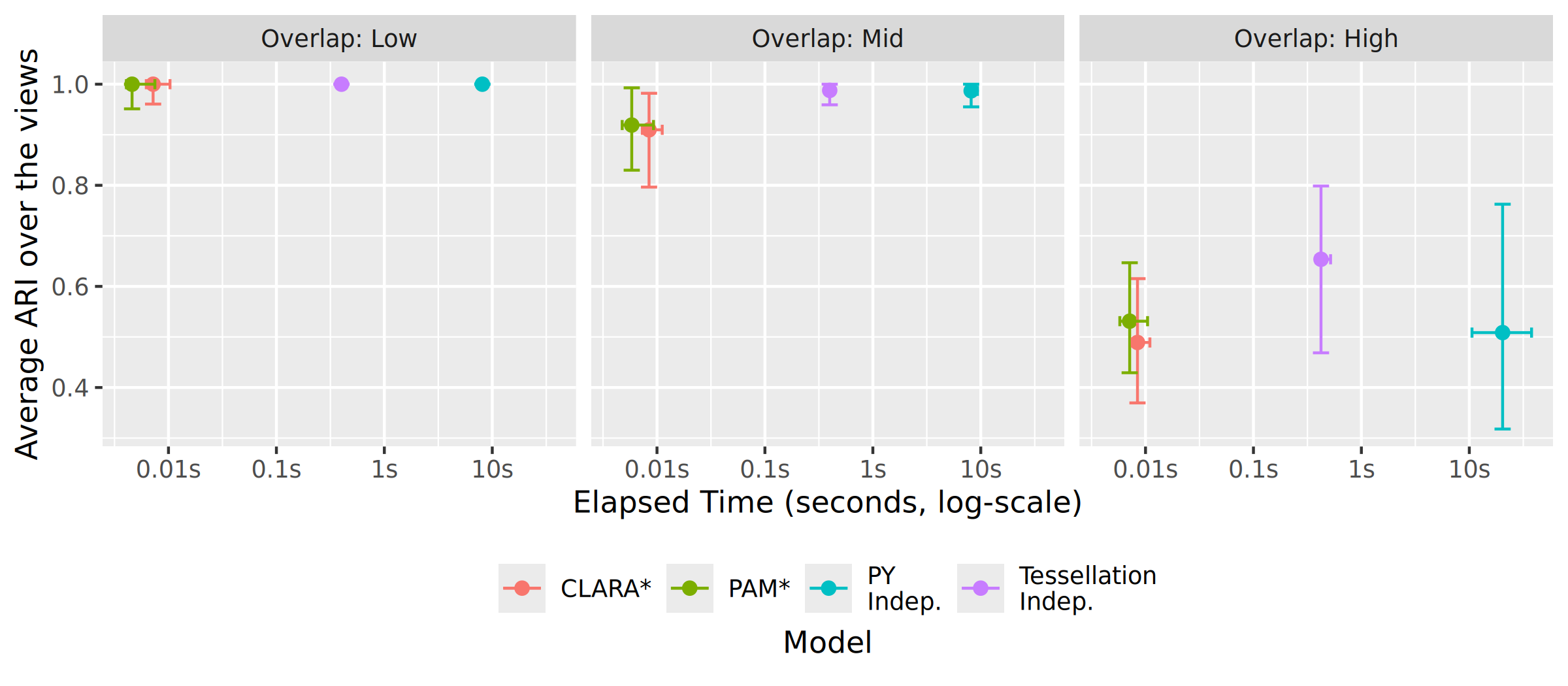}
  \caption{Trade-off between estimation accuracy and computational cost: average Adjusted Rand Index (ARI) against CPU elapsed time (in seconds, $\log_{10}$ scale) across 50 Monte Carlo replications, faceted by cluster overlap regime. Points display the bivariate medians; horizontal and vertical error bars denote the empirical $95\%$ intervals for runtime and clustering accuracy over the 50 replicates, respectively.}
  \label{fig:ari_time_tradeoff}
\end{figure}

\begin{table}[t]
  \centering
  \resizebox{\linewidth}{!}{%
  \begin{tabular}{ccccccc}
    \bottomrule[2pt]
    Overlap & Model & ARI (SD) & NMI (SD) & Jaccard Index (SD) & $\Delta K$ (SD) & Time (s)\\
    \midrule[2pt]
    \multirow{4}{*}{Low} & CLARA* & 0.998 (0.010) & 0.999 (0.006) & 0.997 (0.016) & +0.000 (0.000) & 0.007 \\
    & PAM* & 0.997 (0.012) & 0.998 (0.009) & 0.995 (0.020) & +0.000 (0.000) & 0.005 \\
    & PY Indep. & 1.000 (0.000) & 1.000 (0.000) & 1.000 (0.000) & +0.000 (0.000) & 8.083 \\
    & Tessellation Indep. & 0.998 (0.030) & 0.998 (0.027) & 0.998 (0.028) & -0.016 (0.253) & 0.403 \\
    \midrule
    \multirow{4}{*}{Mid} & CLARA* & 0.904 (0.054) & 0.931 (0.037) & 0.848 (0.078) & +0.000 (0.000) & 0.009 \\
    & PAM* & 0.919 (0.044) & 0.941 (0.031) & 0.869 (0.067) & +0.000 (0.000) & 0.006 \\
    & PY Indep. & 0.985 (0.012) & 0.987 (0.009) & 0.974 (0.021) & +0.208 (0.295) & 8.326 \\
    & Tessellation Indep. & 0.986 (0.011) & 0.988 (0.009) & 0.976 (0.019) & +0.148 (0.258) & 0.399 \\
    \midrule
    \multirow{4}{*}{High} & CLARA* & 0.492 (0.065) & 0.658 (0.040) & 0.379 (0.057) & +0.000 (0.000) & 0.009 \\
    & PAM* & 0.532 (0.055) & 0.685 (0.038) & 0.414 (0.050) & +0.000 (0.000) & 0.007 \\
    & PY Indep. & 0.525 (0.115) & 0.713 (0.058) & 0.429 (0.103) & +7.004 (3.018) & 21.645 \\
    & Tessellation Indep. & 0.651 (0.083) & 0.765 (0.051) & 0.537 (0.084) & +1.982 (0.981) & 0.430 \\
    \bottomrule[2pt]
  \end{tabular}}
  \caption{Performance comparison across cluster overlap regimes: mean (standard deviation) of the average ARI, NMI, Jaccard Index, cluster cardinality discrepancy $\Delta K$ and average CPU runtime (seconds) across 50 replicates. Asterisks ($^\star$) indicate methods supplied with the true number of clusters ($K^\star = 10$).}
  \label{tab:benchmark_summary}
\end{table}

We further assess joint multi-view estimation, showing that it avoids spurious alignment at low inter-view agreement and that \textit{Tessellation Joint} better recovers inter-view concordance under high noise (see Section~\ref{subsec:multiview_alignment} of Supplementary Material).

\section{Numismatics application}
\label{section:numismatics}
\cite{natarajan2024cohesion} and \cite{harris2024past} emphasize the necessity and importance of scalable clustering methods in die studies. A die is a stamp often made of a metal piece, used to strike impressions on heated metal to form a recognizable coin. Die studies are essential for determining the number of dies used to mint a specific issue of coinage. Since dies were typically destroyed by their makers after wearing out, die studies rely primarily on analyzing the coins produced by them. Thus, the primary task is clustering coins to identify those struck from the same die. This process is indispensable for understanding and correlating coinage with premodern historical chronology, including economic conditions, estimating mint outputs, and political history, especially in cases where corresponding historical sources are abundant such as ancient Rome \citep{howgego2002ancient}. Numismatic research, however, remains time-consuming, as each coin must be visually compared to every other coin to confirm die matches. For instance, the 521 coins obverses and reverses in this application would require about $550,000$ visual comparisons, which could require expert numismatics about 800 hours of visual inspection. This complexity hinders large-scale die studies, and thus, automatic workflows for clustering coins into different dies are needed (see Figure \ref{fig:workflow} for a visualization of the workflow).

\begin{figure}[t]
    \centering
    \begin{subfigure}[c]{0.39\linewidth}
      \centering
      \includegraphics[width=\linewidth]{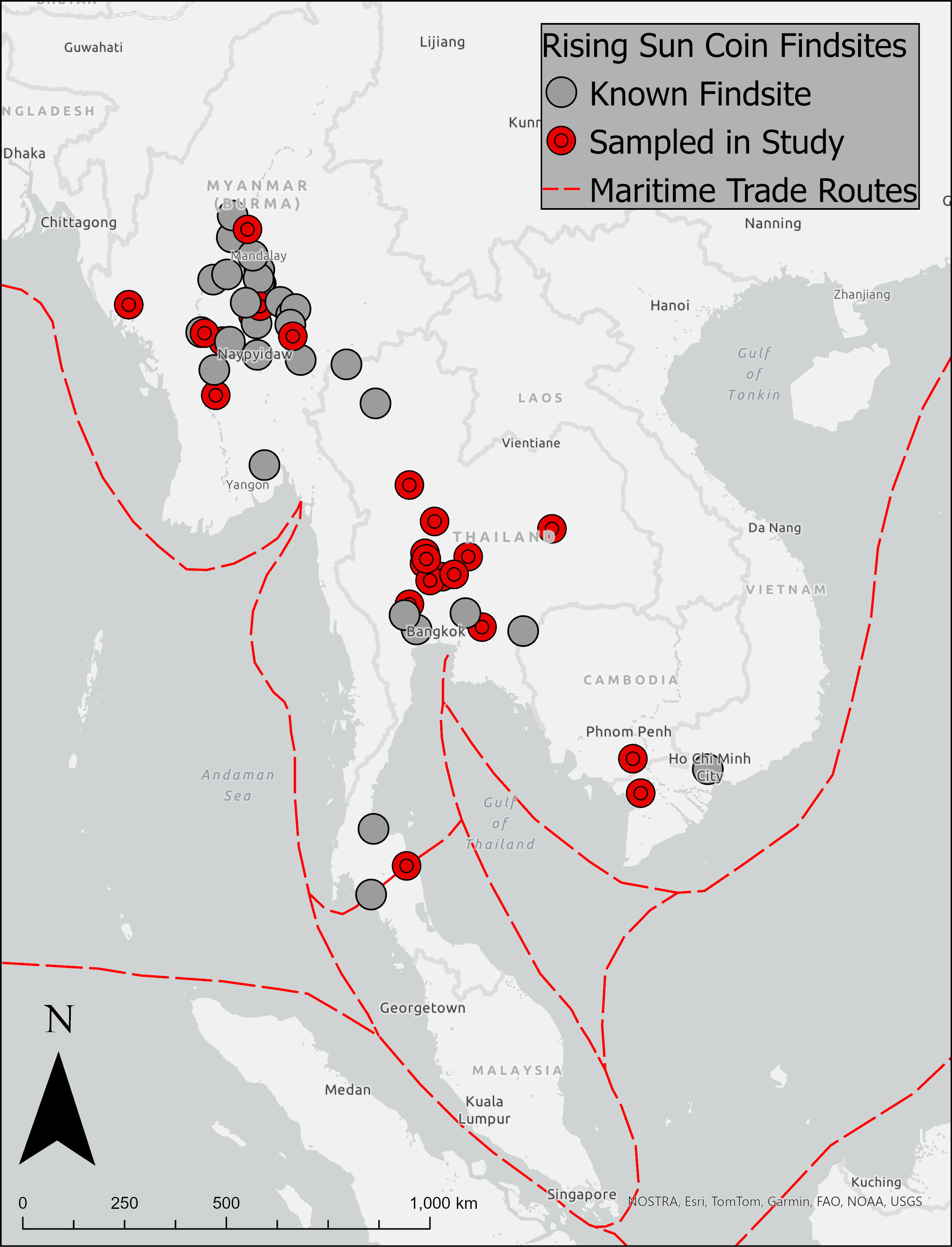}
    \end{subfigure}
    \hfill
    \begin{subfigure}[c]{0.58\linewidth}
      \centering
      \includegraphics[width=\linewidth]{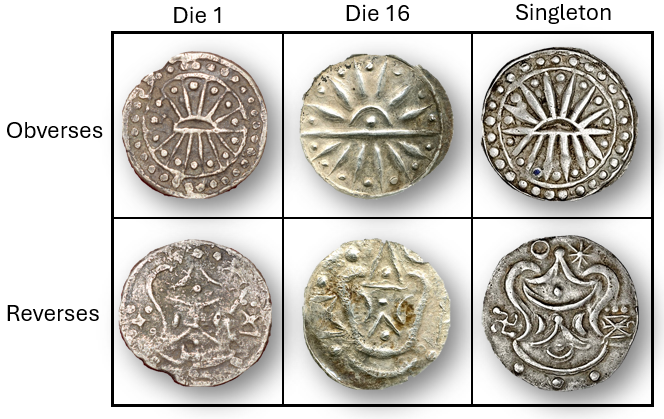}
    \end{subfigure}   
    \caption{Map with findsites and trade routes (left) and photographs of coin obverses and reverses for several identified dies (right).}
    \label{fig:coins}
\end{figure}
 
Figure \ref{fig:coins} shows photographs of the obverse and reverse of 98\% silver coins, which were originally minted in the Pyu-Mon polities of northern and central Myanmar between the 4$^{\text{th}}$ and 9$^{\text{th}}$ centuries AD \citep{wicks1992money,harris2025currents}. A total of 521 coins were collected. The obverse of each coin depicts the motif of a \textqt{Rising Sun} surrounded by a border of 27 zodiac-inspired beads, while the reverse layers an Indic symbol known as a Srivatsa and other supporting symbols such as a swastika, bhadrapittha (royal throne), and celestial symbols \citep{gutman1978ancient}. \textqt{Rising Sun} coins in full-unit ($\approx 9.2-9.4\text{g}$) and smaller denominations have been discovered in 1st millennium AD settlements across mainland and peninsular Thailand, the Mekong Delta areas of Cambodia and Vietnam, and numerous sites throughout Myanmar \citep{mahlo2012early}. Figure~\ref{fig:coins} illustrates known Rising Sun coin find-sites in relation to hypothesized maritime trade routes, while also highlighting specific find-sites. The Rising Sun coins featured in this study were found in Myanmar, Cambodia, Thailand, Vietnam, and the Ganges Delta in Bangladesh, and are either held in museum collections in each country or sourced from private collections and auction houses in Europe and North America, where provenance is not as reliable \citep{mahlo2012early,harris2025currents}.  %

\begin{figure}[t]
    \centering
    \includegraphics[width=1\linewidth]{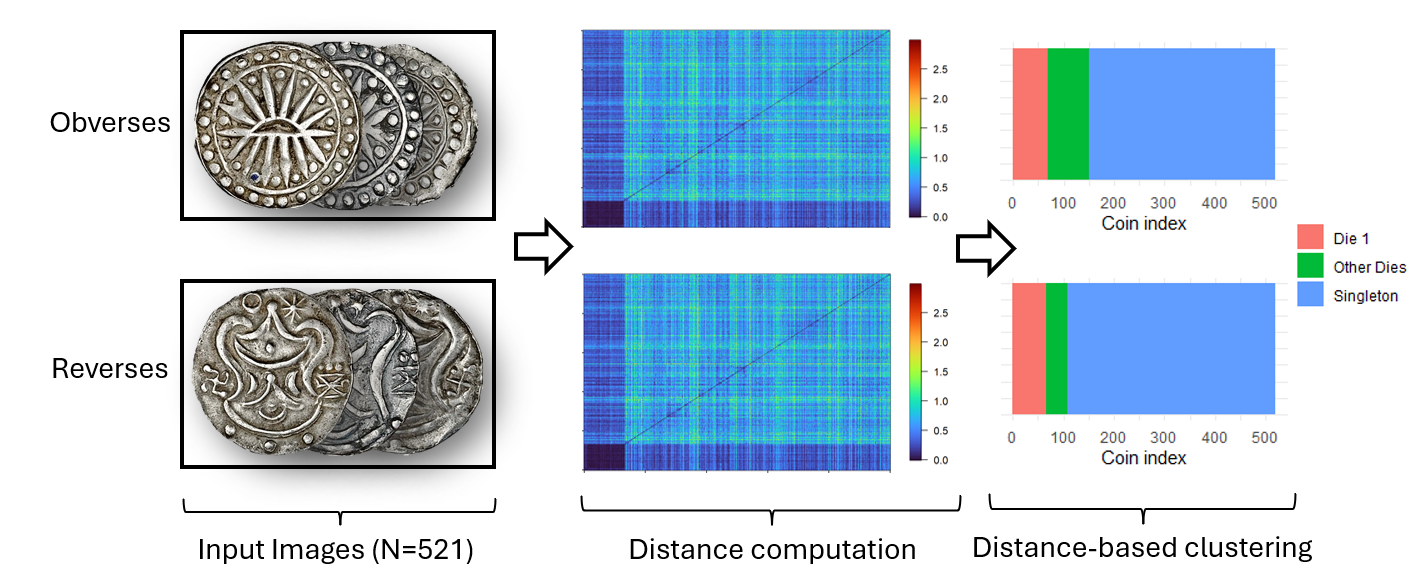}
    \caption{Workflow consisting of image preprocessing, computing the distance matrix between coins, and clustering the coins based on the distances (the cluster estimates on the right are obtained from the tessellation model with nested partitions). }
    \label{fig:workflow}
\end{figure}

\subsection{Results}
\label{sect:results}
We first construct pairwise distance matrices for both the obverse and reverse faces of the 521 coins using the photographic extraction workflow detailed in Section~\ref{app:workflow} of Supplementary Material, which extends the one in \cite{natarajan2024cohesion}. To avoid inflating the partition search space with isolated specimens, we pre-classify coins as singletons whenever the 1\% quantile of their pairwise distances exceeds $0.15$, a threshold that indicates the absence of any plausible die match (see Section~\ref{sect:numismatics_figures} of Supplementary Material for further details on the threshold choice). This screening pre-classify 284 coins as singletons, resulting in a final dataset of $N = 237$ coins. We fit the six models described in Section~\ref{sect:simulations} exclusively on these $237$ coins using the quadratic likelihood specification in \eqref{eq:mainequation}.

\begin{table}[t]
  \begin{subtable}[t]{0.47\textwidth}
    \resizebox{\textwidth}{!}{%
    \begin{tabular}{cccc}
      \bottomrule[2pt]
      Model Class & Formulation & ARI & $K$ (SD) \\
      \midrule[2pt]
      \multirow{2}{*}{$k$-medoids} & PAM & 0.24 & 398 (fixed) \\
      & CLARA & 0.24 & 398 (fixed) \\
      \midrule
      \multirow{2}{*}{PY} & Independent & 0.79 & 426 (0.53) \\
      & Joint & 0.93 & 435 (0.62) \\
      \midrule
      \multirow{2}{*}{Tessellation} & Independent & 0.92 & 405 (2.39) \\
      & Joint & 0.89 & 395 (2.82) \\
      \midrule[2pt]
    \end{tabular}}
    \caption{Coin obverses (view 1)}
  \end{subtable}
  \hfill
  \begin{subtable}[t]{0.47\textwidth}
    \resizebox{\textwidth}{!}{%
    \begin{tabular}{cccc}
      \bottomrule[2pt]
      Model Class & Formulation & ARI & $K$ (SD) \\
      \midrule[2pt]
      \multirow{2}{*}{$k$-medoids} & PAM & 0.31 & 398 (fixed) \\
      & CLARA & 0.27 & 398 (fixed) \\
      \midrule
      \multirow{2}{*}{PY} & Independent & 0.79 & 416 (2.12) \\
      & Joint & 0.86 & 435 (1.02) \\
      \midrule
      \multirow{2}{*}{Tessellation} & Independent & 0.80 & 405 (3.35) \\
      & Joint & 0.81 & 407 (2.15) \\
      \midrule[2pt]
    \end{tabular}}
    \caption{Coin reverses (view 2)}
  \end{subtable}
  \caption{Adjusted Rand Index (ARI) between the cluster estimates and true clustering, posterior mean (and standard deviation) of the number of clusters $K$ for coin obverses (left) and coin reverses (right).}
  \label{tab:table_numismatics}
\end{table}

Left and right panels of Table~\ref{tab:table_numismatics} summarize each model's posterior number of clusters and adjusted Rand index (ARI) between the true clustering structure (obtained from an expert numismatist) and the cluster estimate considering all coins ($N = 521$). We recall that the adjusted Rand index is a measure that quantify how well the estimated clusters match the clustering structure obtained from an expert numismatist. The findings indicate that tessellation-based models, when using medoids for inference, can perform similarly to traditional Bayesian random partition models in a fraction of time. %
In contrast, although extremely fast, %
$k$-medoids algorithms perform poorly in estimating the true clustering structure. It achieves a low ARI, around 0.2 to 0.3, compared to the higher ARI values of approximately 0.8 to 0.9 achieved by the other models. Additional figures and results are provided in Section \ref{sect:numismatics_figures} of Supplementary Material.
About 13\% of the coins belong to the same die, which we call Die 1; 14\%  of coins belong to other dies, which mostly contain two coins only, and the remaining 73\% of coins belong to die clusters with only one coin (singleton). Coins from Die 1 have now been confirmed to have been cast rather than struck, a detail not initially obvious from visual inspection \citep{harris2025currents}. %

\section{Discussion}
\label{sect:discussion}
The proposed tessellation-based Bayesian Distance Clustering models can be viewed as a Bayesian extension of $k$-medoids, performing probabilistic inference on the medoid set. They provide more scalable inference than existing BDC and other Bayesian clustering approaches while retaining comparable predictive performance. Unlike standard $k$-medoids, they also quantify posterior uncertainty, offering a useful compromise between computational efficiency and predictive accuracy.

Future work will focus on improving MCMC efficiency through parallelization and informed proposals. In particular, locally balanced proposals \citep{zanella2019informed} could exploit information in \(\pi(D,\gamma)\), improve mixing, and further enhance posterior exploration and scalability.

\section*{Supplementary material}
We provide a supplement, containing supporting material for Sections \ref{sect:BTM}, \ref{sect:multiview}, \ref{sect:simulations}, and \ref{section:numismatics}. %

\section*{Acknowledgments}
This research was supported by the National University of Singapore through the grant MOE-T2EP40121-0021.

\section*{Disclosure statement}
The authors report there are no competing interests to declare.

\section*{Data availability statement}
The data used and analyzed in Section~\ref{section:numismatics} are available upon request at \url{https://biadwiki.org/}.

\bibliographystyle{apalike}
\bibliography{ScalableVoronoi}

@article{paganin2025informed,
  title={Informed Partition Models for Dependent Random Partitions},
  author={Paganin, Sally and Page, Garritt L and Quintana, Fernando Andr{\'e}s},
  journal={Electronic Journal of Statistics},
  volume={19},
  number={2},
  pages={4973--5018},
  year={2025},
  publisher={The Institute of Mathematical Statistics and the Bernoulli Society}
}

@article{dahl2025dependent,
  title={Dependent random partitions by shrinking toward an anchor},
  author={Dahl, David B and Warr, Richard L and Jensen, Thomas P},
  journal={Journal of the American Statistical Association},
  volume={120},
  number={552},
  pages={2684--2696},
  year={2025},
  publisher={Taylor \& Francis}
}

@book{pitman2006combinatorial,
  title={Combinatorial stochastic processes: Ecole d'et{\'e} de probabilit{\'e}s de saint-flour xxxii-2002},
  author={Pitman, Jim},
  year={2006},
  publisher={Springer Science \& Business Media}
}

@article{carlin2000empirical,
  title={Empirical Bayes: Past, present and future},
  author={Carlin, Bradley P and Louis, Thomas A},
  journal={Journal of the American Statistical Association},
  volume={95},
  number={452},
  pages={1286--1289},
  year={2000},
  publisher={Taylor \& Francis}
}

@article{kass1989approximate,
  title={Approximate Bayesian inference in conditionally independent hierarchical models (parametric empirical Bayes models)},
  author={Kass, Robert E and Steffey, Duane},
  journal={Journal of the American Statistical Association},
  volume={84},
  number={407},
  pages={717--726},
  year={1989},
  publisher={Taylor \& Francis}
}

@article{ni2020scalable,
  title={Scalable Bayesian nonparametric clustering and classification},
  author={Ni, Yang and M{\"u}ller, Peter and Diesendruck, Maurice and Williamson, Sinead and Zhu, Yitan and Ji, Yuan},
  journal={Journal of Computational and Graphical Statistics},
  volume={29},
  number={1},
  pages={53--65},
  year={2020},
  publisher={Taylor \& Francis}
}

@article{bardenet2017markov,
  title={On Markov chain Monte Carlo methods for tall data},
  author={Bardenet, R{\'e}mi and Doucet, Arnaud and Holmes, Chris},
  journal={Journal of Machine Learning Research},
  volume={18},
  number={47},
  pages={1--43},
  year={2017}
}

@misc{lovell2013clustercluster,
  title={{Clustercluster: parallel Markov chain Monte Carlo for Dirichlet process mixtures}},
  author={Lovell, Dan and Malmaud, Jonathan and Adams, Ryan P and Mansinghka, Vikash K},
  journal={arXiv preprint arXiv:1304.2302},
  year={2013}
}

@article{stone2025addivortes,
  title={{AddiVortes: (Bayesian) Additive Voronoi Tessellations}},
  author={Stone, Adam J and Gosling, John Paul},
  journal={Journal of Computational and Graphical Statistics},
  volume={34},
  number={3},
  pages={859--871},
  year={2025},
  publisher={Taylor \& Francis}
}

@article{kunkel2020anchored,
  title={Anchored {B}ayesian {G}aussian mixture models},
  author={Kunkel, Deborah and Peruggia, Mario},
  journal = {Electronic Journal of Statistics},
  volume = {14},
  number = {2},
  pages = {3869-3913},
  year={2020}
}

@article{page2022dependent,
  title={{Dependent modeling of temporal sequences of random partitions}},
  author={Page, Garritt L and Quintana, Fernando A and Dahl, David B},
  journal={Journal of Computational and Graphical Statistics},
  volume={31},
  number={2},
  pages={614--627},
  year={2022},
  publisher={Taylor \& Francis}
}

@Article{pitman1995exchangeable,
  title={{Exchangeable and partially exchangeable random partitions}},
  author={Pitman, Jim},
  journal={Probability theory and related fields},
  volume={102},
  number={2},
  pages={145--158},
  year={1995},
  publisher={Springer}
}

@article{natarajan2024cohesion,
  title={{Cohesion and repulsion in Bayesian distance clustering}},
  author={Natarajan, Abhinav and De Iorio, Maria and Heinecke, Andreas and Mayer, Emanuel and Glenn, Simon},
  journal={Journal of the American Statistical Association},
  volume={119},
  number={546},
  pages={1374--1384},
  year={2024},
  publisher={Taylor \& Francis}
}

@Article{denison2002bayesian1,
  title={{Bayesian partition modelling}},
  author={Denison, DGT and Adams, NM and Holmes, CC and Hand, DJ},
  journal={Computational statistics \& data analysis},
  volume={38},
  number={4},
  pages={475--485},
  year={2002},
  publisher={Elsevier}
}

@Book{denison2002bayesian2,
  title={{Bayesian methods for nonlinear classification and regression}},
  author={Denison, David GT and Holmes, Christopher C and Mallick, Bani K and Smith, Adrian FM},
  volume={386},
  year={2002},
  publisher={John Wiley \& Sons}
}

@Article{betancourt2022random,
  title={{Random partition models for microclustering tasks}},
  author={Betancourt, Brenda and Zanella, Giacomo and Steorts, Rebecca C},
  journal={Journal of the American Statistical Association},
  volume={117},
  number={539},
  pages={1215--1227},
  year={2022},
  publisher={Taylor \& Francis}
}

@article{pitman1996some,
  title = {{Some Developments of the Blackwell-Macqueen URN Scheme}},
  author = {Jim Pitman},
  journal = {Lecture Notes-Monograph Series},
  volume = {30},
  pages = {245--267},
  year = {1996},
  publisher = {Institute of Mathematical Statistics}
}

@article{zanella2019informed,
  title={{Informed proposals for local MCMC in discrete spaces}},
  author={Zanella, Giacomo},
  journal={Journal of the American Statistical Association},
  volume={115},
  number={530},
  pages={852--865},
  year={2020},
  publisher={Taylor \& Francis}
}

@Book{kaufman2009finding,
  title={{Finding groups in data: an introduction to cluster analysis}},
  author={Kaufman, Leonard and Rousseeuw, Peter J},
  year={2009},
  publisher={John Wiley \& Sons}
}

@inproceedings{schubert2019faster,
  author={Schubert, Erich and Rousseeuw, Peter J.},
  editor={Amato, Giuseppe and Gennaro, Claudio and Oria, Vincent and Radovanovi{\'{c}} , Milo{\v{s}}},
  title={{Faster k-Medoids Clustering: Improving the PAM, CLARA, and CLARANS Algorithms}},
  booktitle={Similarity Search and Applications},
  year={2019},
  publisher={Springer International Publishing},
  address={Cham},
  pages={171--187},
  isbn="978-3-030-32047-8"
}

@article{franzolini2026conditional,
  title={{Conditional partial exchangeability: a probabilistic framework for multi-view clustering}},
  author={Franzolini, Beatrice and De Iorio, Maria and Eriksson, Johan},
  journal={Journal of the American Statistical Association},
  pages={1--13},
  year={2026},
  publisher={Taylor \& Francis}
}

@Book{abramowitz1948handbook,
  title={{Handbook of mathematical functions with formulas, graphs, and mathematical tables}},
  author={Abramowitz, Milton and Stegun, Irene A},
  volume={55},
  year={1948},
  publisher={US Government printing office}
}

@Article{duan2021bayesian,
  title={{Bayesian distance clustering}},
  author={Duan, Leo L and Dunson, David B},
  journal={Journal of Machine Learning Research},
  volume={22},
  number={224},
  pages={1--27},
  year={2021}
}

@Article{rigon2023generalized,
  title={{A generalized Bayes framework for probabilistic clustering}},
  author={Rigon, Tommaso and Herring, Amy H and Dunson, David B},
  journal={Biometrika},
  volume={110},
  number={3},
  pages={559--578},
  year={2023},
  publisher={Oxford University Press}
}

@article{harris2024past,
  title={{From past to future: digital methods towards artefact analysis}},
  author={Harris, Andrew and Cremaschi, Andrea and Lim, Tse Siang and De Iorio, Maria and Kwa, Chong Guan},
  journal={Digital Scholarship in the Humanities},
  volume={39},
  number={4},
  pages={1026--1042},
  year={2024},
  publisher={Oxford University Press}
}

@article{harris2025currents,
  title   = {{Currents of currency: utilising die studies to trace Rising Sun/Srivatsa coin distribution in first-millennium AD Southeast Asia}},
  author  = {Harris, Andrew and Cabral, Rafael and De Iorio, Maria and Krajaejun, Pipad and Kwa, Chong Guan},
  journal = {Antiquity},
  year    = {2025},
  volume  = {99},
  number  = {406},
  pages   = {1030--1048},
  publisher = {Cambridge University Press}
}

@Book{howgego2002ancient,
  title={{Ancient history from coins}},
  author={Howgego, Christopher},
  year={2002},
  publisher={Routledge}
}

@Article{gutman1978ancient,
  title={{The Ancient Coinage of Mainland Southeast Asia}},
  author={Gutman, P},
  journal={Journal of the Siam Society},
  volume={66},
  number={1},
  pages={18--21},
  year={1978},
}

@Article{mahlo2012early,
  title={{The early coins of Myanmar/Burma: messengers from the past: Pyu, Mon, and Candras of Arakan (first millenium AD)}},
  author={Mahlo, Dietrich and Margolis, Karen},
  journal={(No Title)},
  year={2012}
}

@inbook{wicks1992money,
 author = {Robert S. Wicks},
 booktitle = {Money, Markets, and Trade in Early Southeast Asia: The Development of Indigenous Monetary Systems to AD 1400},
 pages = {111--155},
 publisher = {Cornell University Press},
 title = {Money and Society in Ancient Burma: Mon, Pyu, and Pagan},
 year = {1992}
}

@Article{rudin1992nonlinear,
  title={{Nonlinear total variation based noise removal algorithms}},
  author={Rudin, Leonid I and Osher, Stanley and Fatemi, Emad},
  journal={Physica D: nonlinear phenomena},
  volume={60},
  number={1-4},
  pages={259--268},
  year={1992},
  publisher={Elsevier}
}

@Article{pizer1987adaptive,
  title={{Adaptive histogram equalization and its variations}},
  author={Pizer, Stephen M and Amburn, E Philip and Austin, John D and Cromartie, Robert and Geselowitz, Ari and Greer, Trey and ter Haar Romeny, Bart and Zimmerman, John B and Zuiderveld, Karel},
  journal={Computer vision, graphics, and image processing},
  volume={39},
  number={3},
  pages={355--368},
  year={1987},
  publisher={Elsevier}
}

@misc{Mathematica,
  author = {{Wolfram Research Inc.}},
  title = {Mathematica, {V}ersion 14.0},
  url = {https://www.wolfram.com/mathematica},
  year = {2024}
}

@inproceedings{dusmanu2019d2,
  title={{D2-net: A trainable CNN for joint description and detection of local features}},
  author={Dusmanu, Mihai and Rocco, Ignacio and Pajdla, Tomas and Pollefeys, Marc and Sivic, Josef and Torii, Akihiko and Sattler, Torsten},
  booktitle={{Proceedings of the IEEE/CVF conference on computer vision and pattern recognition}},
  pages={8092--8101},
  year={2019}
}

@Article{muja2009fast,
  title={{Fast approximate nearest neighbors with automatic algorithm configuration.}},
  author={Muja, Marius and Lowe, David G},
  journal={VISAPP (1)},
  volume={2},
  number={331-340},
  pages={2},
  year={2009}
}

@Article{gao2019gaussian,
  title={{Gaussian process landmarking on manifolds}},
  author={Gao, Tingran and Kovalsky, Shahar Z and Daubechies, Ingrid},
  journal={SIAM Journal on Mathematics of Data Science},
  volume={1},
  number={1},
  pages={208--236},
  year={2019},
  publisher={SIAM}
}

@inproceedings{shechtman2007matching,
  title={{Matching local self-similarities across images and videos}},
  author={Shechtman, Eli and Irani, Michal},
  booktitle={{IEEE conference on computer vision and pattern recognition}},
  pages={1--8},
  year={2007},
  organization={IEEE}
}

@inproceedings{chatfield2009efficient,
  title={{Efficient retrieval of deformable shape classes using local self-similarities}},
  author={Chatfield, Ken and Philbin, James and Zisserman, Andrew},
  booktitle={{2009 IEEE 12th International Conference on Computer Vision Workshops, ICCV Workshops}},
  pages={264--271},
  year={2009},
  organization={IEEE}
}

@misc{matlab,
  year = {2022},
  author = {{MathWorks Inc.}},
  title = {MATLAB version: 9.13.0 (R2022b)},
  publisher = {The MathWorks Inc.},
  address = {Natick, Massachusetts, United States},
  url = {https://www.mathworks.com}
}

@Article{dahl2022search,
  title={{Search algorithms and loss functions for Bayesian clustering}},
  author={Dahl, David B and Johnson, Devin J and M{\"u}ller, Peter},
  journal={Journal of Computational and Graphical Statistics},
  volume={31},
  number={4},
  pages={1189--1201},
  year={2022},
  publisher={Taylor \& Francis}
}

@article{wade2018bayesian,
author = {Sara Wade and Zoubin Ghahramani},
title = {{Bayesian Cluster Analysis: Point Estimation and Credible Balls (with Discussion)}},
journal = {Bayesian Analysis},
volume = {13},
number = {2},
pages = {559 -- 626},
year = {2018},
publisher = {International Society for Bayesian Analysis}
}

@Article{maechler2019finding,
  title={{Finding groups in data: Cluster analysis extended Rousseeuw et al}},
  author={Maechler, Martin and others},
  journal={R package version},
  volume={2},
  number={0},
  pages={242--248},
  year={2019}
}

@Article{duan2020latent,
  title={{Latent simplex position model: High dimensional multi-view clustering with uncertainty quantification}},
  author={Duan, Leo L},
  journal={Journal of Machine Learning Research},
  volume={21},
  number={38},
  pages={1--25},
  year={2020}
}

@Article{jain2010data,
  title={{Data clustering: 50 years beyond K-means}},
  author={Jain, Anil K},
  journal={Pattern recognition letters},
  volume={31},
  number={8},
  pages={651--666},
  year={2010},
  publisher={Elsevier}
}

@Book{hastie2009elements,
  title={{The elements of statistical learning: data mining, inference, and prediction}},
  author={Hastie, Trevor and Tibshirani, Robert and Friedman, Jerome H and Friedman, Jerome H},
  volume={2},
  year={2009},
  publisher={Springer}
}

@Article{caron2017generalized,
  title={{Generalized P{\'o}lya Urn for Time-Varying Pitman-Yor Processes}},
  author={Caron, Fran{\c{c}}ois and Neiswanger, Willie and Wood, Frank and Doucet, Arnaud and Davy, Manuel},
  journal={Journal of Machine Learning Research},
  volume={18},
  number={27},
  pages={1--32},
  year={2017}
}

@article{deiorio2023bayesian,
  title={Bayesian nonparametric mixture modeling for temporal dynamics of gender stereotypes},
  author={De Iorio, Maria and Favaro, Stefano and Guglielmi, Alessandra and Ye, Lifeng},
  journal={The Annals of Applied Statistics},
  volume={17},
  number={3},
  pages={2256--2278},
  year={2023},
  publisher={Institute of Mathematical Statistics}
}

@Article{paganin2021centered,
  title={{Centered partition processes: informative priors for clustering (with discussion)}},
  author={Paganin, Sally and Herring, Amy H and Olshan, Andrew F and Dunson, David B and Study, The National Birth Defects Prevention},
  journal={Bayesian analysis},
  volume={16},
  number={1},
  pages={301},
  year={2021},
  publisher={NIH Public Access}
}

@Article{ferguson1973bayesian,
  title={{A Bayesian analysis of some nonparametric problems}},
  author={Ferguson, Thomas S},
  journal={The annals of statistics},
  pages={209--230},
  year={1973},
  publisher={JSTOR}
}

@article{jain2004split,
  title={{A split-merge Markov chain Monte Carlo procedure for the Dirichlet process mixture model}},
  author={Jain, Sonia and Neal, Radford M},
  journal={Journal of computational and Graphical Statistics},
  volume={13},
  number={1},
  pages={158--182},
  year={2004},
  publisher={Taylor \& Francis}
}

@article{rand1971objective,
  title={{Objective criteria for the evaluation of clustering methods}},
  author={Rand, William M},
  journal={Journal of the American Statistical association},
  volume={66},
  number={336},
  pages={846--850},
  year={1971},
  publisher={Taylor \& Francis}
}

@article{meilua2007comparing,
  title={{Comparing clusterings — an information based distance}},
  author={Meil{\u{a}}, Marina},
  journal={Journal of multivariate analysis},
  volume={98},
  number={5},
  pages={873--895},
  year={2007},
  publisher={Elsevier}
}

@article{mcdaid2011normalized,
  title={Normalized mutual information to evaluate overlapping community finding algorithms},
  author={McDaid, Aaron F and Greene, Derek and Hurley, Neil},
  journal={arXiv preprint arXiv:1110.2515},
  year={2011}
}

@article{garthwaite2016adaptive,
  title={Adaptive optimal scaling of Metropolis--Hastings algorithms using the Robbins--Monro process},
  author={Garthwaite, Paul H and Fan, Yanan and Sisson, Scott A},
  journal={Communications in Statistics-Theory and Methods},
  volume={45},
  number={17},
  pages={5098--5111},
  year={2016},
  publisher={Taylor \& Francis}
}

@article{roberts1997weak,
  title={Weak convergence and optimal scaling of random walk Metropolis algorithms},
  author={Roberts, Gareth O and Gelman, Andrew and Gilks, Walter R},
  journal={The annals of applied probability},
  volume={7},
  number={1},
  pages={110--120},
  year={1997},
  publisher={Institute of Mathematical Statistics}
}

@article{cremaschi2025repulsion,
  title={{Repulsion, chaos, and equilibrium in mixture models}},
  author={Cremaschi, Andrea and Wertz, Timothy M and De Iorio, Maria},
  journal={Journal of the Royal Statistical Society Series B: Statistical Methodology},
  volume={87},
  number={2},
  pages={389--432},
  year={2025},
  publisher={Oxford University Press UK}
}

\clearpage\appendix
\begin{center}
    \LARGE \textbf{Supplementary Material for:\\\textqt{Scalable Bayesian Clustering for Integrative Analysis of Multi-View Data}}
\end{center}
\vspace{1cm}

\setcounter{section}{0}
\setcounter{equation}{0}
\setcounter{figure}{0}
\setcounter{table}{0}
\setcounter{page}{1}

\renewcommand{\thesection}{S\arabic{section}}
\renewcommand{\thefigure}{S\arabic{figure}}
\renewcommand{\thetable}{S\arabic{table}}
\renewcommand{\theequation}{S\arabic{equation}}

\makeatletter
    \renewcommand{\p@subfigure}{\thefigure}
    \renewcommand{\theHfigure}{S\arabic{figure}}
    \renewcommand{\theHtable}{S\arabic{table}}
    \renewcommand{\theHequation}{S\arabic{equation}}
\makeatother


\noindent This document provides supporting material for the results presented in the main manuscript.

\section{Marginal likelihood conditional on partition model}
\label{sect:marginallik}
We can integrate out the parameters $\bm{\theta}$ and $\bm{\lambda}$ in $\pi\left(\mathbf{D} \mid \bm{\theta}, \bm{\lambda}, \mathcal{T}\right)$ and avoid performing inference on these terms. First, we note that, for random variables $X_i$ and $i=1,\dots,N$, if $X_i \mid \lambda \iid \operatorname{Gamma}(\delta,\lambda)$, and $\lambda \sim \operatorname{Gamma}(\mu,\beta)$, then:
\begin{equation*}
  \pi_{\bm{X}}(\bm{x}) = \int \prod_{i=1}^n \pi(x_i \mid \lambda)\pi(\lambda)\dd{\lambda} = \frac{\Gamma(\delta)^{-n} \Gamma(\mu+n\delta)}{\Gamma(\mu)}\ \beta^\mu \left( \prod_{i=1}^n x_i\right)^{\delta-1} \left( \beta + \sum_{i=1}^n x_i\right)^{-\mu-\delta n}
\end{equation*}
Therefore, $\pi\left(\mathbf{D} \mid \mathcal{T}\right) =
\int \pi\left(\mathbf{D} \mid \bm{\theta}, \bm{\lambda}, \mathcal{T}\right)\pi(\bm{\lambda})\pi(\bm{\theta})\,\dd{\bm{\lambda}}\dd{\bm{\theta}}$ simplifies to:
\begin{align}
  \label{eq:Datasimple}
  \pi\left(\mathbf{D} \mid \mathcal{T}\right) &= \left[\prod_{k=1}^K \int \prod_{\substack{i, j \in \mathcal{C}_k \\ i<j}}  f\left(D_{ij} \mid \lambda_k\right) \pi(\lambda_k)\dd{\lambda_k}\right] \left[\prod_{(k, t) \in A} \int \prod_{\substack{i \in \mathcal{C}_k \\ j \in \mathcal{C}_t}}  g\left(D_{ij} \mid \theta_{kt}\right) \pi(\theta_{kt}) \dd{\theta_{kt}}\right] \\
  &= \prod_{k=1}^K \mathcal{L}_k^{(1)} \prod_{(k, t) \in A} \mathcal{L}_{kt}^{(2)} \notag
  \intertext{where:}
  \mathcal{L}_k^{(1)} &= \frac{\Gamma(\delta_1)^{-n_k^{\star}} \Gamma(\mu+n_k^{\star}\delta_1)}{\Gamma(\mu)}\ \beta^\mu \left(\,\prod_{i, j \in \mathcal{C}_k,\, i<j} D_{ij}\right)^{\delta_1-1} \left(\beta + \sum_{i, j \in \mathcal{C}_k,\, i<j} D_{ij}\right)^{-\mu-\delta_1 n_k^{\star}} \notag \\ 
  \mathcal{L}_{kt}^{(2)} &= \frac{\Gamma(\delta_2)^{-n_kn_t} \Gamma(\zeta + n_kn_t\delta_2)}{\Gamma(\zeta)}\ \gamma^\zeta \left(\,\prod_{i \in \mathcal{C}_k,\, j \in \mathcal{C}_t} D_{ij} \right)^{\delta_2-1} \left( \gamma + \sum_{i \in \mathcal{C}_k,\, j \in \mathcal{C}_t} D_{ij}\right)^{-\zeta-\delta_2 n_k n_t} \notag
\end{align}
In the above expressions, $n_k^{\star}= n_k(n_k-1)/2$, and $n_k$ is the number of objects in cluster $k$.

\section{Hyperparameter selection}
\label{sect:hyperpar_select}
The following algorithm outlines our steps to determine the hyperparameters in the \textqt{linear} distance matrix likelihood \eqref{eq:linear_simple} of the main manuscript. The procedure is adapted from the one in \cite{natarajan2024cohesion}, which is developed for the likelihood \eqref{eq:mainequation} of the main manuscript. Instead of performing maximum likelihood estimation (MLE), which requires numerical methods and is prone to numerical issues due to the potentially small dimension of $\mathbf{b}$ (the distances between medoids), we opt for the method of moments. Algorithm~\ref{algo:eb_linear} provides a detailed description of the steps involved in this process.

\begin{algorithm}[t!]
  \caption{Choosing $\delta_1,\delta_2,\mu, \beta$, and $\theta$ from the distance matrix $\mathbf{D}$}
  \label{algo:eb_linear}
  \DontPrintSemicolon
  \Input{Distance matrix $\mathbf{D}$}
  \BlankLine
  \Output{Estimated parameters $\delta_1, \delta_2, \mu, \beta, \theta$}
  \BlankLine\BlankLine
  Compute an initial value for $K$, denoted as $K_{\text{elbow}}$, using the elbow method. This is achieved by applying $k$-medoids clustering (through the PAM, CLARA, or CLARANS algorithms) on the distance matrix $\mathbf{D}$ for a range of $K$ values, using the within-cluster-sum-of-squares (WSS) score as the criterion.\;

  Perform $k$-medoids clustering with $K_{\text{elbow}}$ to obtain an initial cluster configuration.\;

  Define two sets of distances: $\mathbf{a}$, which includes distances between each non-medoid object and its assigned medoid, and $\mathbf{b}$, which includes distances between the medoids.\;

  Fit a Gamma distribution to the values in $\mathbf{a}$ and estimate $\delta_1$ using the method of moments, given by $\delta_1 = \Bar{\mathbf{a}}^2/S_\mathbf{a}^2$, where $\Bar{\mathbf{a}}$ and $S_\mathbf{a}^2$ represent the sample mean and variance of $\mathbf{a}$, respectively.\;

  Set $\mu = \delta_1 n_{\mathbf{a}}$ and $\beta = \sum_{i} a_i$, where $n_{\mathbf{a}}$ is the cardinality of $\mathbf{a}$. This corresponds to the conditional posterior of $\lambda$, obtained by specifying an improper prior $\pi(\lambda) \propto I(\lambda > 0)$ and treating $\mathbf{a}$ as a weighted set of observations from a Gamma($\delta_1$, $\lambda$) distribution.\;

  Fit a Gamma distribution to the values in $\mathbf{b}$ and estimate $\delta_2$ and $\theta$, the shape and rate parameters of this distribution. The method of moments yields $\delta_2 = \Bar{\mathbf{b}}^2/S_\mathbf{b}^2$ and $\theta = \Bar{\mathbf{b}}/S_\mathbf{b}^2$.\;
\end{algorithm}

\section{Posterior inference}
\label{sect:algorithms}
In this section, we present the MCMC algorithms employed to perform inference, starting with the tessellation models for one view in Section~\ref{sect:tessellation} of the manuscript, and then presenting the algorithms for the multi-view data based on nested tessellation models (Section~\ref{section:nested_partition} of the manuscript) and joint tessellation models (Section~\ref{section:joint_partition} of the manuscript).

\subsection{Tessellation models}
\label{sect:Alg2}
We perform inference on the Voronoi tessellation models of Section \ref{sect:tessellation} of the main man\-u\-script using the Metropolis-Hastings (MH) algorithm, detailed in Algorithm \ref{algo:2}. The MH algorithm involves birth, death, and move steps at each iteration. Specifically, we propose a new medoid set, $\boldsymbol{\gamma}^{\star}$, by adding, removing, or moving medoids from the medoid set of the previous iteration, $\boldsymbol{\gamma}^{(t-1)}$ \citep{denison2002bayesian1}. In our implementation, we randomly choose between a birth and move step, and at every second iteration, we perform a move step. In the birth step, we uniformly select a new index from $[N] \setminus \boldsymbol{\gamma}$, representing the indices of objects that are not medoids. Similarly, for the death step, we uniformly select an index from $\boldsymbol{\gamma}$ to be removed. The move step combines a death step followed by a birth step. The constant $\beta$ in Algorithm \ref{algo:2} is the ratio of the proposal probabilities (the probability of moving from the new state to the old state divided by the probability of moving from the old state to the new state).

\begin{algorithm}[t!]%
  \caption{Birth-Death and Move Algorithm for Tessellation Models}%
  \label{algo:2}
  \DontPrintSemicolon
  \Input{Data matrix $\mathbf{D}$, total iterations $T$}
  \BlankLine
  \Output{Posterior chains $\{\bm{\gamma}^{(t)}\}_{t=1}^T$, $\{\mathcal{T}^{(t)}\}_{t=1}^T$}
  \BlankLine\BlankLine
  Initialize parameters $\bm{\gamma}^{(0)}$ and compute $\mathcal{T}^{(0)}$ via \eqref{eq:gamma} in the manuscript\;

  \For{$t = 1, 2, \ldots, T$}{
    \eIf{$t \bmod 2 = 0$ \tcp*[h]{Move step}}{
      Remove one medoid uniformly from $\bm{\gamma}^{(t-1)}$ and add a new one uniformly from $[N] \setminus \bm{\gamma}^{(t-1)}$, yielding proposed set $\bm{\gamma}^{\star}$\;
      Compute $\mathcal{T}^{\star}(\mathbf{D}, \bm{\gamma}^{\star})$ via \eqref{eq:gamma} in the manuscript\;
      Set $\beta = \frac{\card{\bm{\gamma}^{(t-1)}}\,\left(N - \card{\bm{\gamma}^{(t-1)}}\right)}{\card{\bm{\gamma}^{\star}}\,\left(N - \card{\bm{\gamma}^{\star}}\right)} = 1$\;
      Calculate acceptance ratio $\alpha = \min \left( 1, \frac{\pi(\mathbf{D} \mid \mathcal{T}^{\star})\pi(\bm{\gamma}^{\star})}{\pi(\mathbf{D} \mid \mathcal{T}^{(t-1)})\pi(\bm{\gamma}^{(t-1)})} \times \beta \right)$\;
    }{
      Sample $u \sim \text{Uniform}(0, 1)$\;
      \eIf{$u < 0.5$ \tcp*[h]{Birth step}}{
        Add a new medoid uniformly from $[N] \setminus \bm{\gamma}^{(t-1)}$, yielding proposed set $\bm{\gamma}^{\star}$\;
        Compute $\mathcal{T}^{\star}(\mathbf{D}, \bm{\gamma}^{\star})$ via \eqref{eq:gamma} in the manuscript\;
        Set $\beta = \frac{N - \card{\bm{\gamma}^{(t-1)}}}{\card{\bm{\gamma}^{\star}}}$\;
        Calculate acceptance ratio $\alpha = \min \left( 1, \frac{\pi(\mathbf{D} \mid \mathcal{T}^{\star})\pi(\bm{\gamma}^{\star})}{\pi(\mathbf{D} \mid \mathcal{T}^{(t-1)})\pi(\bm{\gamma}^{(t-1)})} \times \beta \right)$\;
      }{
        \tcp{Death step}
        Remove a medoid uniformly from $\bm{\gamma}^{(t-1)}$, yielding proposed set $\bm{\gamma}^{\star}$\;
        Compute $\mathcal{T}^{\star}(\mathbf{D}, \bm{\gamma}^{\star})$ via \eqref{eq:gamma} in the manuscript\;
        Set $\beta = \frac{\card{\bm{\gamma}^{(t-1)}}}{N - \card{\bm{\gamma}^{\star}}}$\;
        Calculate acceptance ratio $\alpha = \min \left( 1, \frac{\pi(\mathbf{D} \mid \mathcal{T}^{\star})\pi(\bm{\gamma}^{\star})}{\pi(\mathbf{D} \mid \mathcal{T}^{(t-1)})\pi(\bm{\gamma}^{(t-1)})} \times \beta \right)$\;
      }
    }
      
    Sample $v \sim \text{Uniform}(0, 1)$\;
    \eIf{$v < \alpha$}{
      Accept: $\bm{\gamma}^{(t)} \gets \bm{\gamma}^{\star}$, $\mathcal{T}^{(t)} \gets \mathcal{T}^{\star}$\;
    }{
      Reject: $\bm{\gamma}^{(t)} \gets \bm{\gamma}^{(t-1)}$, $\mathcal{T}^{(t)} \gets \mathcal{T}^{(t-1)}$\;
    }
  }
\end{algorithm}

It is also straightforward to employ a Gibbs sampler for this model. Consider the auxiliary variables \(w_1, \dotsc, w_N\), where \(w_i = 1\) if object \(i\) is a medoid, and \(w_i = 0\) otherwise. The conditional probability is given by
\begin{equation*}
  w_i = 1 \mid rest \propto \pi(\mathbf{D} \mid w_i = 1, \bm{w}_{-i})\,\pi(w_i = 1, \bm{w}_{-i})
\end{equation*}
where $\bm{w}_{-i}$ is the medoid indicator vector after observation $i$ has been removed. The prior $\pi(w_i = 1, \bm{w}_{-i})$ follows from the prior in Section~\ref{section:priors} of the main manuscript, which only depends on $N$ and the number of clusters $K=\sum^{N}_{i=1} w_i$.
We add object $i$ to the medoid set with probability given by $\pi(w_i = 1 \mid rest)$, and otherwise remove it from the medoid set. At each step of the Gibbs sampler, this operation is repeated for all $N$ objects. Thus, we increase the computational cost of each MCMC iteration by a factor of $N$ compared to Algorithm~\ref{algo:2}. However, combining the Birth-Death and Move algorithms with the Gibbs sampler may be beneficial since the Gibbs sampler gives an acceptance rate of 1, while Algorithm~\ref{algo:2} will generally have a low acceptance rate. Future work aims to explore this approach in greater detail.

\subsection{Nested tessellation models}
\label{sect:Algnested}
Algorithm~\ref{algo:nested} generalizes the single-view MCMC strategy of Algorithm~\ref{algo:2} to the multi-view setting with $V$ views. We first update the first view ($v = 1$) by generating a proposal $\bm{\gamma}_1^{\star}$ via standard birth, death, or move steps detailed in Algorithm~\ref{algo:2}. We then ensure that $\bm{\gamma}_1^{\star}$ remains compatible with the current child medoid configuration $\bm{\gamma}_2^{(t-1)}$ (i.e., $\mathcal{C}_{1k}^\star \cap \bm{\gamma}_2^{(t-1)} \neq \emptyset$ for all $k = 1, \dots, |\bm{\gamma}_1^\star|$). When valid, we can compute the induced child likelihood $\pi(\mathbf{D}^{(2)} \mid \mathcal{T}_2^\star)$ and compute the Metropolis--Hastings ratio for acceptance.

Subsequently, we sequentially update each nested view $v \in \{2, \dots, V\}$. Rather than proposing a global configuration for $\bm{\gamma}_v$, we propose to perform $K_{v-1}$ localized Metropolis--Hastings updates for each parent cluster $\mathcal{C}_{v-1\,k}^{(t)}$ ($k = 1, \dots, K_{v-1}$). For each parent cluster $k$, the birth, death, and move steps are restricted to the subset $\mathcal{C}_{v-1\,k}^{(t)}$. This localized proposals guarantees compatibility between $v$ and $v-1$ by construction. Then, each local candidate $\bm{\gamma}^{\star}_{vk}$ can be tested for acceptance only if it satisfies the compatibility checks with the child view medoids $\bm{\gamma}_{v+1}^{(t-1)}$. Finally, notice that for $v = V$ the compatibility check is not required.

\begin{algorithm}[t!]%
  \caption{Birth-Death and Move Algorithm for Nested Tessellation Models}
  \label{algo:nested}
  \DontPrintSemicolon
  \SetKwInOut{Input}{Input}\SetKwInOut{Output}{Output}
  \Input{Multi-view data $\{\mathbf{D}^{(1)}, \dots, \mathbf{D}^{(V)}\}$, total iterations $T$}
  \BlankLine
  \Output{Posterior chains $\{\bm{\gamma}_1^{(t)}, \dots, \bm{\gamma}_V^{(t)}\}_{t=1}^T$, $\{\mathcal{T}_1^{(t)}, \dots, \mathcal{T}_V^{(t)}\}_{t=1}^T$}
  \BlankLine\BlankLine
  Initialize valid configurations $\bm{\gamma}^{(0)}_{1}, \dots, \bm{\gamma}^{(0)}_{V}$ and compute $\mathcal{T}^{(0)}_{1}, \dots, \mathcal{T}^{(0)}_{V}$ via \eqref{eq:gamma} in the manuscript\;
  \For{$t = 1, 2, \dots, T$}{
    \tcp{Update medoid set for the first view}
    Propose $\bm{\gamma}_1^\star$ via steps 3--19 of Algorithm~\ref{algo:2}\;
    Compute $\mathcal{T}^{\star}_{1} = \{\mathcal{C}^{\star}_{11}, \dots, \mathcal{C}^{\star}_{1K^{\star}_1}\}$ via \eqref{eq:gamma} in the manuscript\;
    \eIf{$\mathcal{C}_{1k}^\star \cap \bm{\gamma}_2^{(t-1)} \neq \emptyset \quad \forall k \in \{1, \dots, K^{\star}_1\}$}
    {
      Compute acceptance ratio $\alpha_1$\; %
      $\bm{\gamma}_1^{(t)} \leftarrow \bm{\gamma}_1^\star$ with probability $\min(1, \alpha_1)$, else $\bm{\gamma}_1^{(t-1)}$\;
    }
    {
      \tcp{Reject: incompatible proposal}
      $\bm{\gamma}_1^{(t)} \leftarrow \bm{\gamma}_1^{(t-1)}$
    }
    Update $\mathcal{T}^{(t)}_{1}$ from $\bm{\gamma}^{(t)}_1$ via \eqref{eq:gamma} in the manuscript\;

    \BlankLine
    \tcp{Update medoid sets in the subsequent views}
    \For{$v = 2, \dots, V$}{
      \For{$k = 1, \dots, K_{v-1}$}{
        Propose $\bm{\gamma}^{\star}_{vk} \subset \mathcal{C}_{v-1\,k}^{(t)}$ via steps 3--19 of Algorithm~\ref{algo:2} restricted to $\mathcal{C}_{v-1\,k}^{(t)}$\;
        Compute local tessellation $\mathcal{T}^{\star}_{vk}$ via \eqref{eq:gamma} in the manuscript restricted to $\mathcal{C}_{v-1\,k}^{(t)}$\;
        \eIf{$v = V$ \textbf{or} $\left(\mathcal{C}_{vj}^\star \cap \bm{\gamma}_{v+1}^{(t-1)} \neq \emptyset \quad \forall j\right)$}
        {
          Compute local acceptance ratio $\alpha_{vk}$\;
          $\bm{\gamma}_{vk}^{(t)} \leftarrow \bm{\gamma}_{vk}^\star$ with probability $\min(1, \alpha_{vk})$, else $\bm{\gamma}_{vk}^{(t-1)}$\;
        }
        {
          \tcp{Reject: incompatible proposal}
          $\bm{\gamma}_{vk}^{(t)} \leftarrow \bm{\gamma}_{vk}^{(t-1)}$
        }
        Assemble global configuration: $\bm{\gamma}_v^{(t)} \leftarrow (\bm{\gamma}_{v1}^{(t)}, \dots, \bm{\gamma}^{(t)}_{vK_{v-1}})$\;
        Update $\mathcal{T}^{(t)}_{v}$ from $\bm{\gamma}_v^{(t)}$ via \eqref{eq:gamma} in the manuscript\;
      }
    }
  }
\end{algorithm}

\subsection{Joint tessellation models}
\label{sect:Algjoint}
Algorithm~\ref{algo:joint} is an adaptation of the MCMC algorithm in Algorithm~\ref{algo:2}. For computational efficiency, we avoid performing inference on $\alpha_{uv}$, $1 \leq u < v \leq V$, directly. We set independent marginal priors $\alpha_{uv} \iid \operatorname{Beta}(a,b)$ and then marginalize out each $\alpha_{uv}$ from the joint prior distribution in \eqref{eq:joint_partition_prior} of the manuscript, leading to
\begin{equation*}
  \pi\left(\bm{\gamma}_1, \dots, \bm{\gamma}_V \right) = \prod_{v = 1}^{V} \pi(\bm{\gamma}_v) \prod_{1\leq u < v \leq V} \int_{0}^{1} \exp\left\{-\phi(\alpha_{uv})d(\bm{\gamma}_{u}, \bm{\gamma}_{v})\right\} \pi(\alpha_{uv})\,\dd{\alpha_{uv}}
\end{equation*}
Given that each $\alpha_{uv}$ is independently and identically distributed as $\operatorname{Beta}(a,b)$ and assuming $\phi(\alpha_{uv}) = \alpha_{uv} / (1 - \alpha_{uv})$, we can evaluate the integral in closed form as
\begin{equation}
  \label{eq:marg_interaction}
  \int^{1}_{0} \exp \bigl\{ -\phi(\alpha_{uv})\, d(\bm{\gamma}_{u}, \bm{\gamma}_{v}) \bigr\} \pi(\alpha_{uv}) \,\dd{\alpha_{uv}} = \frac{\Gamma(a) \operatorname{T}(a, 1-b, d(\bm{\gamma}_u, \bm{\gamma}_v))}{\operatorname{B}(a,b)}
\end{equation}
where $\operatorname{T}(\cdot,\cdot,\cdot)$ stands for the Tricomi confluent hypergeometric function \citep{abramowitz1948handbook}, $\Gamma(\cdot)$ is the Gamma function, and $\operatorname{B}(\cdot, \cdot)$ is the Beta function. We call the right-hand side of \eqref{eq:marg_interaction} the \emph{interaction term} between views $u$ and $v$, denoted $C(\bm{\gamma}_u, \bm{\gamma}_v)$. Thus, the target distribution of Algorithm~\ref{algo:joint} simplifies to
\begin{equation*}
  \pi\left(\bm{\gamma}_1, \dots, \bm{\gamma}_V \mid \mathbf{D}^{(1)}, \dots, \mathbf{D}^{(V)} \right) \propto \prod_{v=1}^V \pi\left(\mathbf{D}^{(v)} \mid \bm{\gamma}_{v}\right) \pi(\bm{\gamma}_{v}) \prod_{1 \leq u < v \leq V} C(\bm{\gamma}_u, \bm{\gamma}_v)
\end{equation*}

\begin{algorithm}[t]
  \caption{Birth-Death and Move Algorithm for Joint Tessellation Models}
  \label{algo:joint}
  \DontPrintSemicolon
  \Input{Multi-view data $\{\mathbf{D}^{(1)}, \dots, \mathbf{D}^{(V)}\}$, total iterations $T$}
  \BlankLine
  \Output{Posterior chains $\{\bm{\gamma}_1^{(t)}, \dots, \bm{\gamma}_V^{(t)}\}_{t=1}^T$, $\{\mathcal{T}_1^{(t)}, \dots, \mathcal{T}_V^{(t)}\}_{t=1}^T$}
  \BlankLine\BlankLine
  Initialize $\bm{\gamma}^{(0)}_{1}, \dots, \bm{\gamma}^{(0)}_{V}$ and compute $\mathcal{T}^{(0)}_{1}, \dots, \mathcal{T}^{(0)}_{V}$ via \eqref{eq:gamma} in the manuscript\;
  \For{$t = 1, 2, \ldots, T$} {
    \For{$v = 1, 2, \ldots, V$} {
      Propose $\bm{\gamma}^{\star}_v$ via steps 3--19 of Algorithm~\ref{algo:2} using $\mathbf{D}^{(v)}$; compute proposal ratio $\beta_v$\;
      Compute proposed tessellation $\mathcal{T}^{\star}_{v}$ via \eqref{eq:gamma} in the manuscript\;
      Compute acceptance ratio $\alpha_v = \frac{\pi(\mathbf{D}^{(v)}\mid \mathcal{T}^{\star}_{v})\,\pi(\bm{\gamma}^{\star}_v)\,\prod_{u \neq v} C(\bm{\gamma}^{\star}_{v}, \bm{\gamma}_u)}{\pi(\mathbf{D}^{(v)}\mid \mathcal{T}^{(t-1)}_{v})\,\pi(\bm{\gamma}^{(t-1)}_{v})\,\prod_{u \neq v} C(\bm{\gamma}^{(t-1)}_{v}, \bm{\gamma}_u)} \times \beta_v$,
      where $\bm{\gamma}_u = \bm{\gamma}^{(t)}_u$ for $u < v$ and $\bm{\gamma}_u = \bm{\gamma}^{(t-1)}_u$ for $u > v$\;
      Set $\bm{\gamma}^{(t)}_{v} \gets \bm{\gamma}^{\star}_{v}$ with probability $\min(1, \alpha_v)$, else $\bm{\gamma}^{(t)}_v \gets \bm{\gamma}^{(t-1)}_{v}$\;
      Update $\mathcal{T}^{(t)}_{v}$ from $\bm{\gamma}_v^{(t)}$ via \eqref{eq:gamma} in the manuscript\; 
    }
  }
\end{algorithm}

Even if the parameters $\alpha_{uv}$ have been marginalized out, we can recover them from the posterior distribution. After obtaining the posterior samples of the medoid sets $(\bm{\gamma}_1, \dots, \bm{\gamma}_V)$, we can compute posterior samples of any $\alpha_{uv}$ through its conditional distribution, which only depends on the medoid sets of the two considered views:
\begin{equation*}
  \pi\left(\alpha_{uv} \mid \bm{\gamma}_u, \bm{\gamma}_v \right) \propto \exp \left\{ -\phi(\alpha_{uv}) d(\bm{\gamma}_{u},\bm{\gamma}_{v}) \right\}\, \alpha_{uv}^{a-1} (1-\alpha_{uv})^{b-1}
\end{equation*}

Finally, notice that also in this case a Gibbs sampler can be employed by considering the auxiliary variables $w_{vi}$ for $v \in [V]$ and $i \in [N]$, where $w_{vi} = 1$ if object $i$ in view $v$ is a medoid and it is $0$ otherwise. The full conditional probabilities are given by
\begin{equation*}
  w_{vi} = 1 \mid rest \propto \pi(\mathbf{D}^{(v)} \mid w_{vi} = 1, \bm{w}_{v\,-i})\,\pi(w_{vi} = 1, \bm{w}_{v\,-i})\,\prod_{u \neq v} C(\bm{w}_u, w_{vi} = 1, \bm{w}_{v\,-i})
\end{equation*}
where $\bm{w}_{v\,-i}$ is the medoid indicator vector in view $v$ after observation $i$ has been removed.

\section{Additional results for the simulation study}
\label{sect:sim_more_info}
In this section, we provide additional results and detailed sensitivity analyses complementing the simulation study in Section~\ref{sect:simulations} of the main manuscript. Section~\ref{subsec:multiview_alignment} evaluates the properties of joint multi-view inference, examining whether simultaneous estimation induces spurious inter-view concordance and highlighting the observation assignment flows across agreement rates and noise levels. Section~\ref{subsec:sensitivity_analysis} presents a comprehensive sensitivity analysis assessing the role of the repulsive term in preventing cluster degeneracy, alongside a comparative evaluation of the linear versus quadratic likelihood specifications. Finally, Section~\ref{subsec:scalability} benchmark computational runtimes across all models and analyzes algorithmic scalability as both the sample size $N$ and the number of views $V$ grow.

\subsection{Multi-view joint estimation and partition alignment}
\label{subsec:multiview_alignment}
We evaluate the impact of simultaneous multi-view estimation by comparing the joint formulations (\textit{Tessellation Joint} and \textit{PY Joint}) against their independent counterparts across all simulation settings. Figure~\ref{fig:ari_joint} illustrates within-view clustering accuracy (average ARI across the two views) as a function of the inter-view agreement rate $\alpha_s$ across overlap regimes: Low ($\sigma_s = 0.1$), Mid ($\sigma_s = 0.2$), and High ($\sigma_s = 0.3$). Across all conditions, joint estimation preserves within-view partition accuracy without degrading recovery when inter-view agreement is low.

Figure~\ref{fig:ari_between} assesses the reconstructed inter-view partition alignment by evaluating the empirical ARI computed directly between the two estimated view partitions. We observe that: $(i)$ the proposed joint loss function does not force spurious alignment when the true inter-view agreement is low; and $(ii)$ the \textit{Tessellation Joint} formulation recovers a higher inter-view alignment across all overlap regimes, outperforming the \textit{PY Joint} baseline under high overlap. This structural alignment is visually confirmed by the observation assignment flows in Figure~\ref{fig:waterfall_plots}, which compares \textit{Tessellation Joint} (Figure~\ref{subfig:waterfall_tessellation}) and \textit{PY Joint} (Figure~\ref{subfig:waterfall_py}) across agreement rates and noise levels, demonstrating how cluster allocations are preserved without artificial inter-view contamination.

\begin{figure}[htbp]
    \centering
    \includegraphics[width=\linewidth]{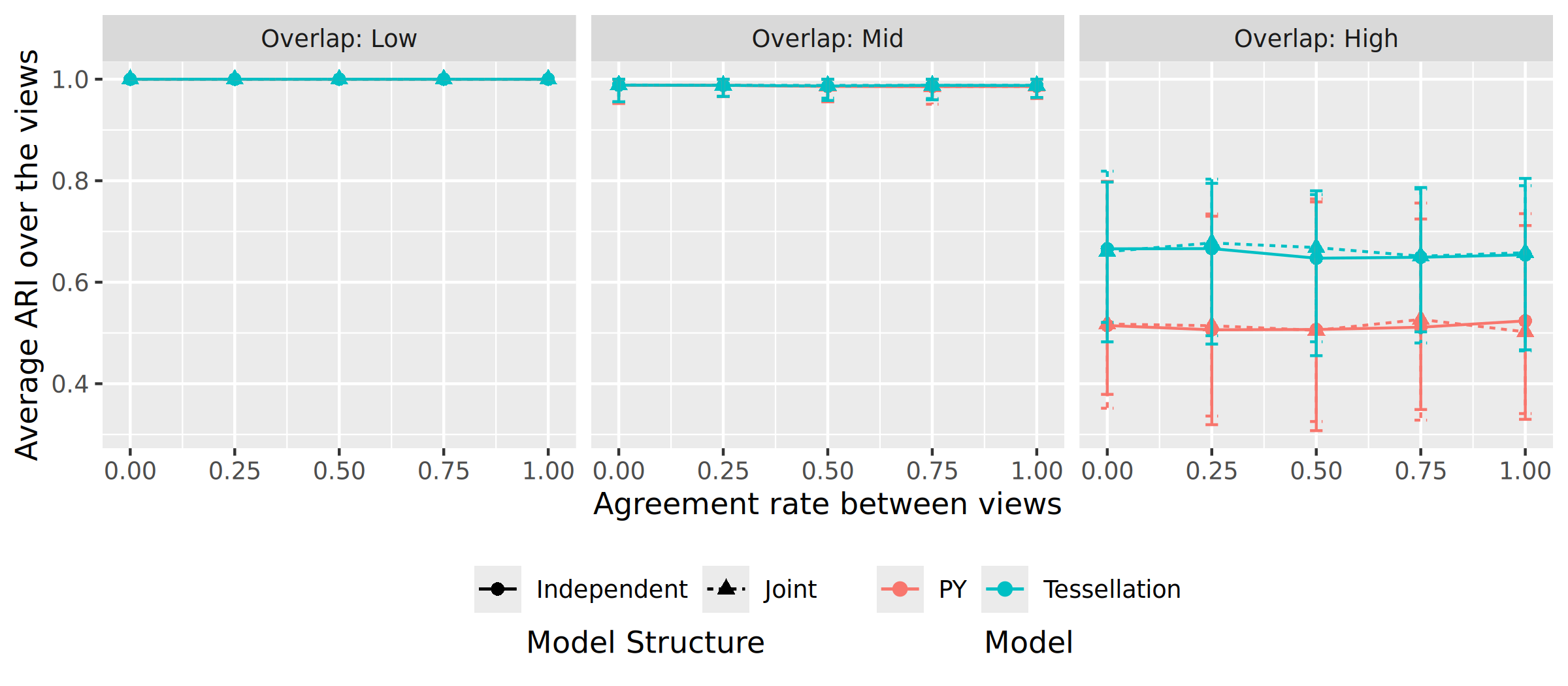}
    \caption{Within-view clustering accuracy across different models: average Adjusted Rand Index (ARI) over the two views as a function of the inter-view agreement rate ($\alpha_s$), comparing \textit{Independent} (solid lines) and \textit{Joint} (dashed lines) specifications for \textit{Tessellation} and \textit{PY} mixture models across overlap regimes: Low ($\sigma_s = 0.1$), Mid ($\sigma_s = 0.2$), and High ($\sigma_s = 0.3$). Points display medians across 50 replicates; error bars denote $95\%$ empirical intervals.}
    \label{fig:ari_joint}
\end{figure}

\begin{figure}[htbp]
    \centering
    \includegraphics[width=\linewidth]{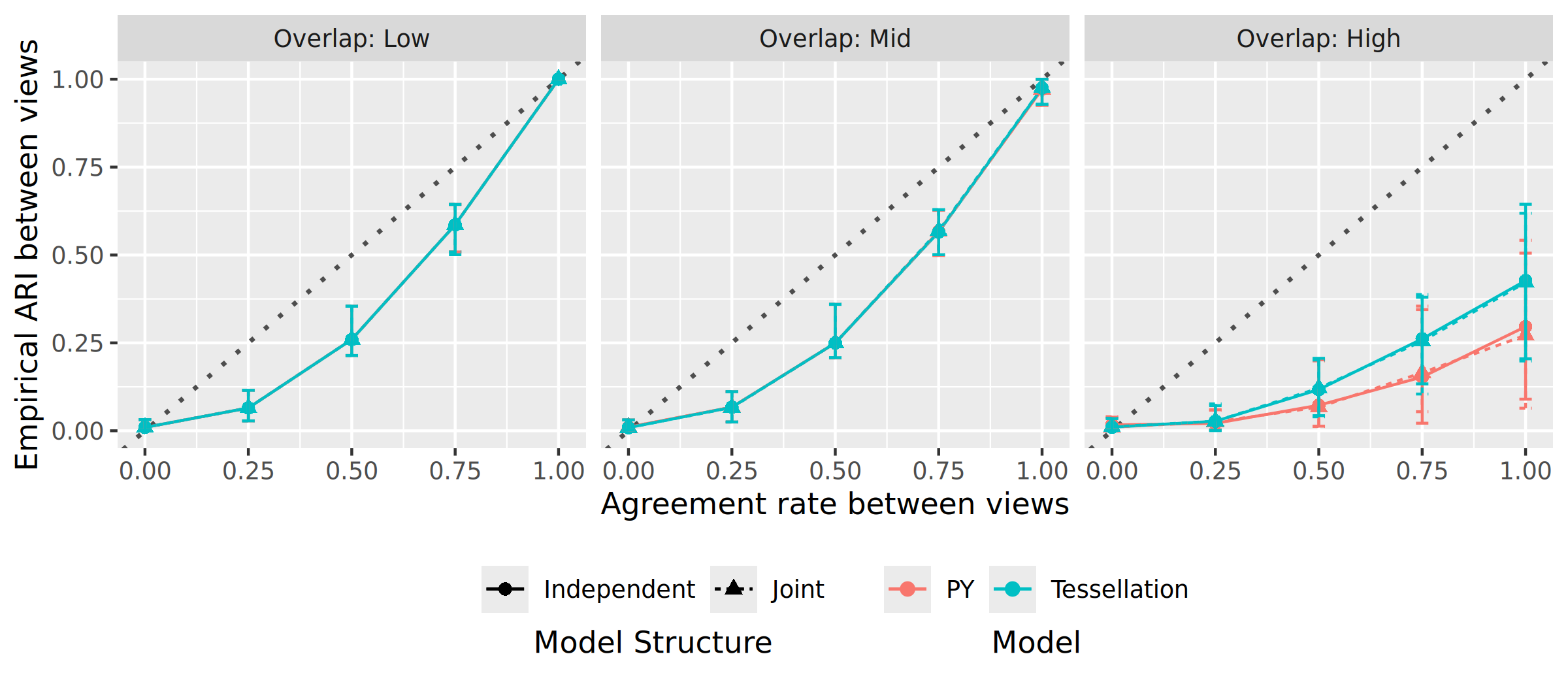}
    \caption{Reconstructed inter-view partition alignment: empirical ARI computed directly between the two estimated view partitions as a function of the agreement rate $\alpha_s$, comparing \textit{Independent} (solid lines) and \textit{Joint} (dashed lines) specifications for \textit{Tessellation} and \textit{PY} mixture models across overlap regimes: Low ($\sigma_s = 0.1$), Mid ($\sigma_s = 0.2$), and High ($\sigma_s = 0.3$) The dotted diagonal line represents theoretical identity.}
    \label{fig:ari_between}
\end{figure}

\begin{figure}[htbp]
    \centering
    \begin{subfigure}{\linewidth}
        \centering
        \includegraphics[width=.9\linewidth]{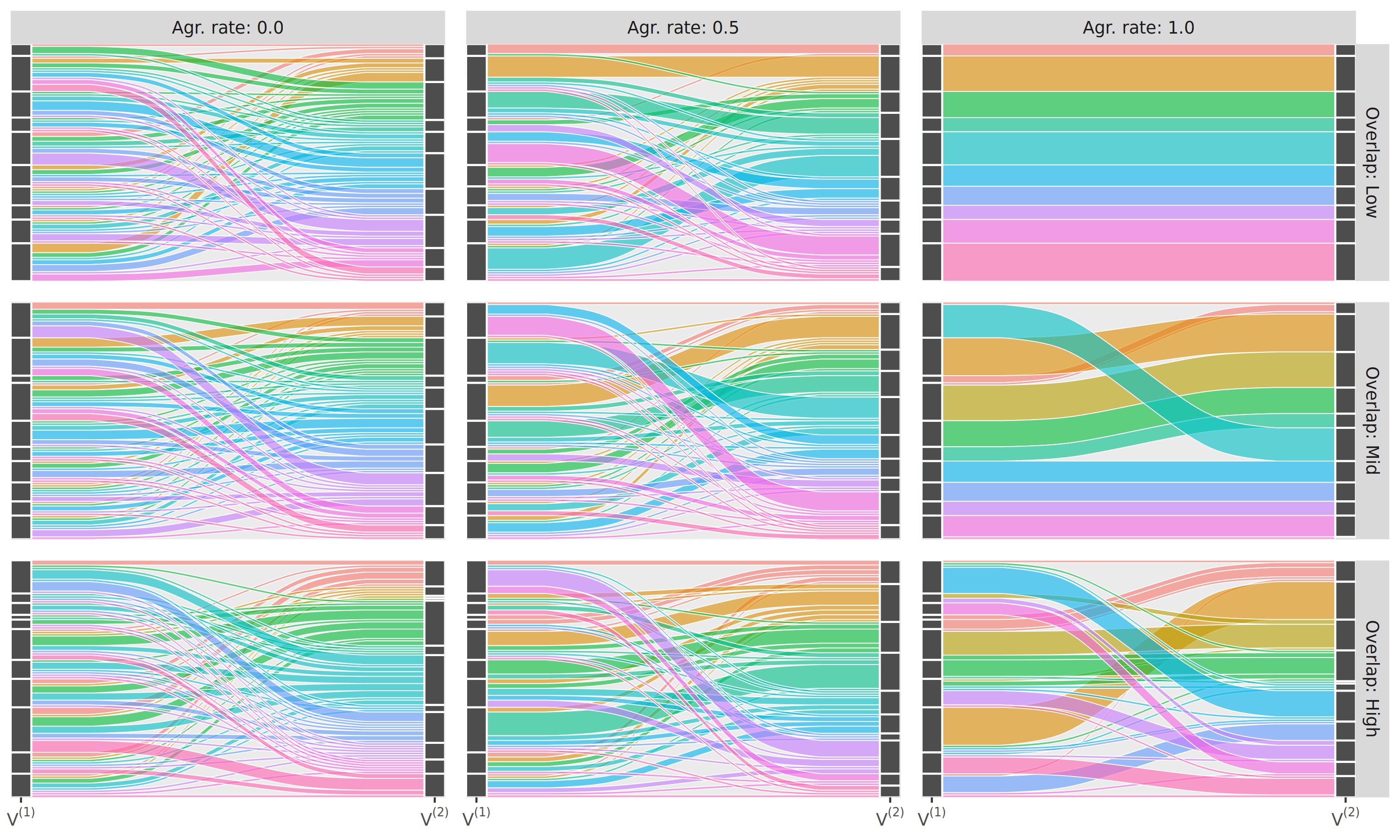}
        \vspace{-0.4cm}
        \caption{\textit{Tessellation Joint}}
        \label{subfig:waterfall_tessellation}
        \vspace{0.4cm}
    \end{subfigure}
    \begin{subfigure}[b]{\linewidth}
        \centering
        \includegraphics[width=.9\linewidth]{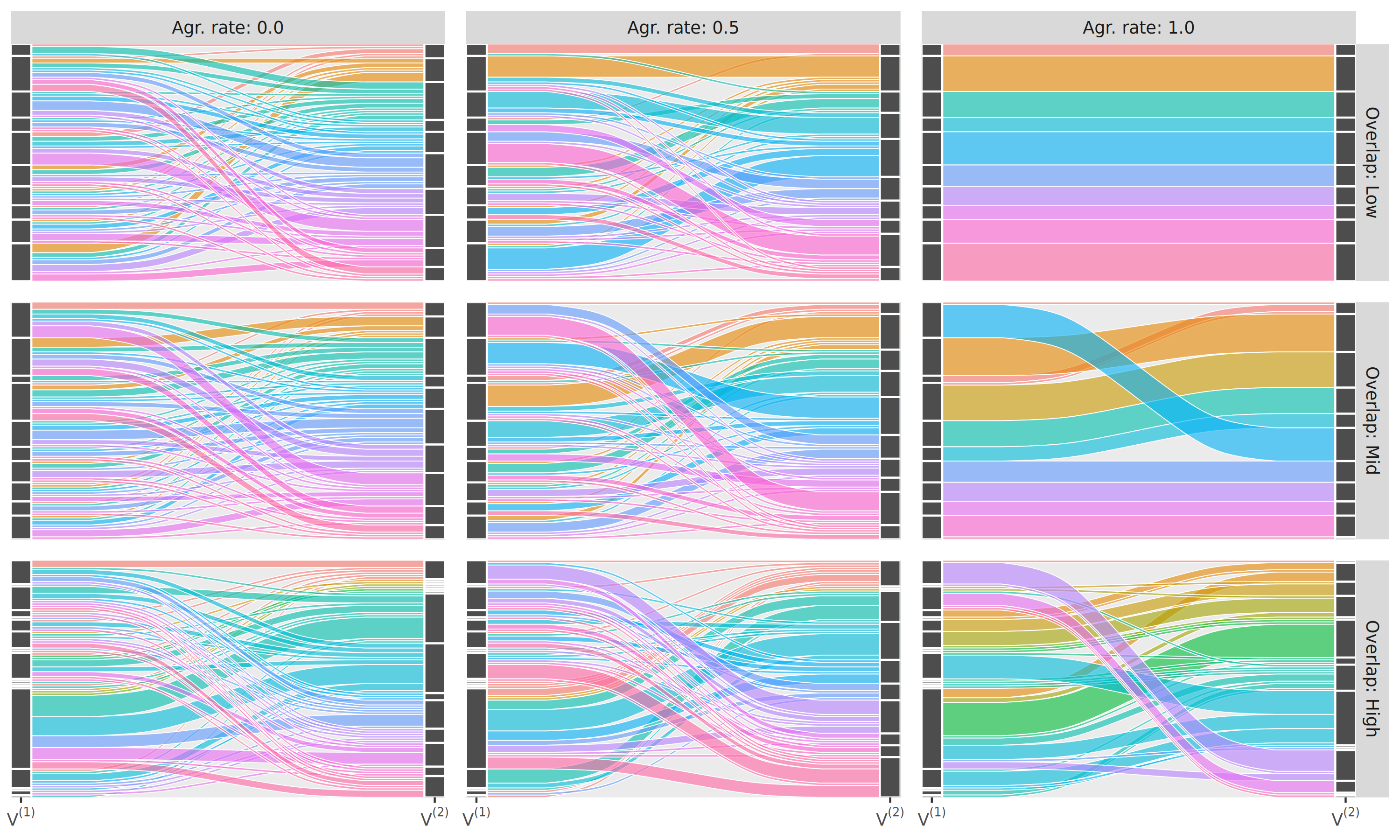}
        \vspace{-0.4cm}
        \caption{\textit{PY Joint}}
        \label{subfig:waterfall_py}
    \end{subfigure}
    \caption{Waterfall plots showing the observation assignment flows between the views across varying agreement rates and overlap levels. Block heights reflect the relative proportion of observations assigned to each cluster within a given view.}
    \label{fig:waterfall_plots}
\end{figure}

\subsection{Sensitivity w.r.t. likelihood specifications}
\label{subsec:sensitivity_analysis}
We assess the sensitivity of the proposed model with respect to different specifications of the likelihood. In particular, we examine: $(i)$ the role of the repulsion term in the likelihood function when performing model distance-based clustering, and $(ii)$ the trade-off between clustering accuracy and computational cost under the \textqt{quadratic} likelihood in \eqref{eq:mainequation} and the \textqt{linear} likelihood formulation in \eqref{eq:linear_simple} of the main manuscript. To quantify the role of the repulsion term in the likelihood function, we compare the quality of the cluster estimation for models fitted with and without the repulsion term under the quadratic likelihood specification. Figure~\ref{fig:repulsion_performance} reports the distribution of the mean ARI between the two views and the mean $\Delta K$ relative to the true cardinality $K^\star = 10$ across all overlap regimes as a function of presence / absence of the repulsion term.

\begin{figure}[t!]
  \centering
  \begin{subfigure}[b]{\linewidth}
    \includegraphics[width=\linewidth]{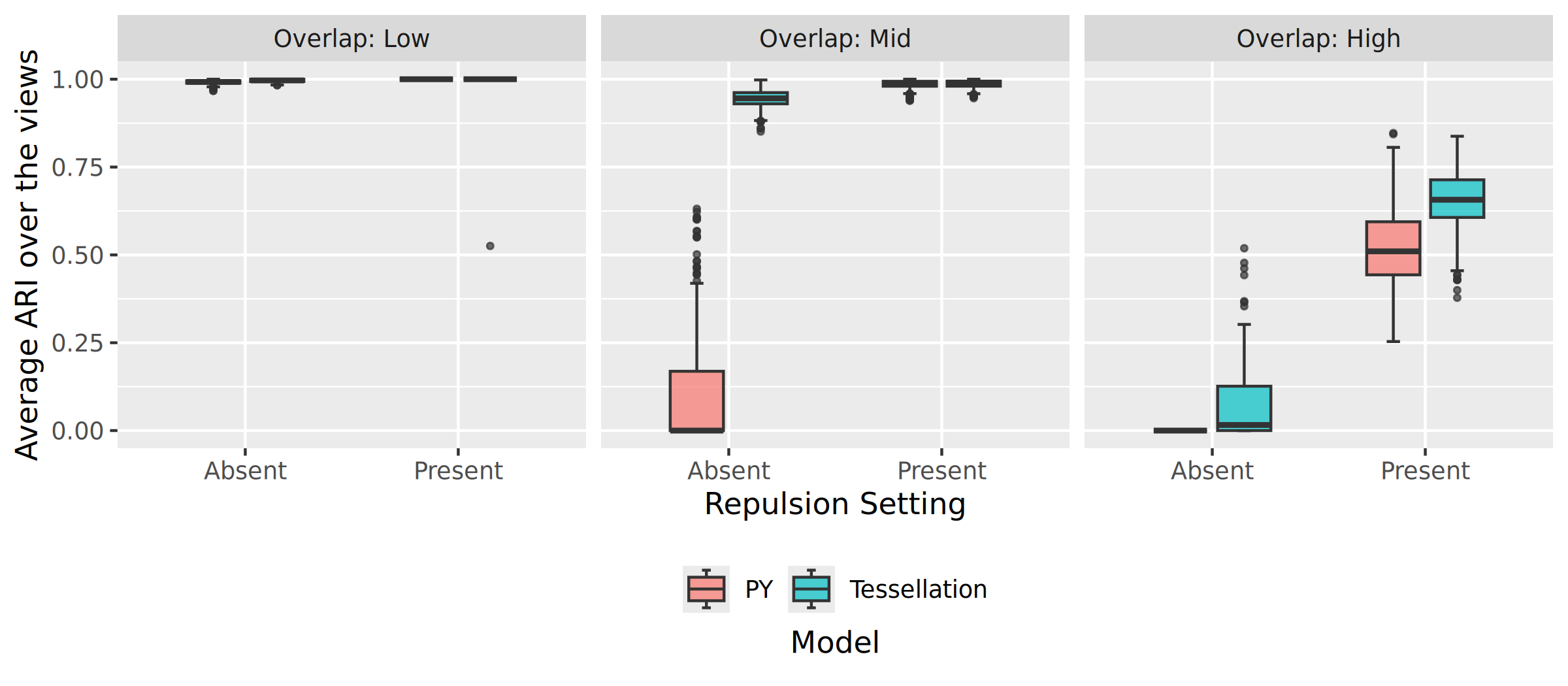}\vspace{-0.4cm}
    \caption{Average Adjusted Rand Index (ARI) over views.}
    \label{fig:mean_ari_repulsion}
  \end{subfigure}

  \vspace{0.3cm}
  
  \begin{subfigure}[b]{\linewidth}
    \includegraphics[width=\linewidth]{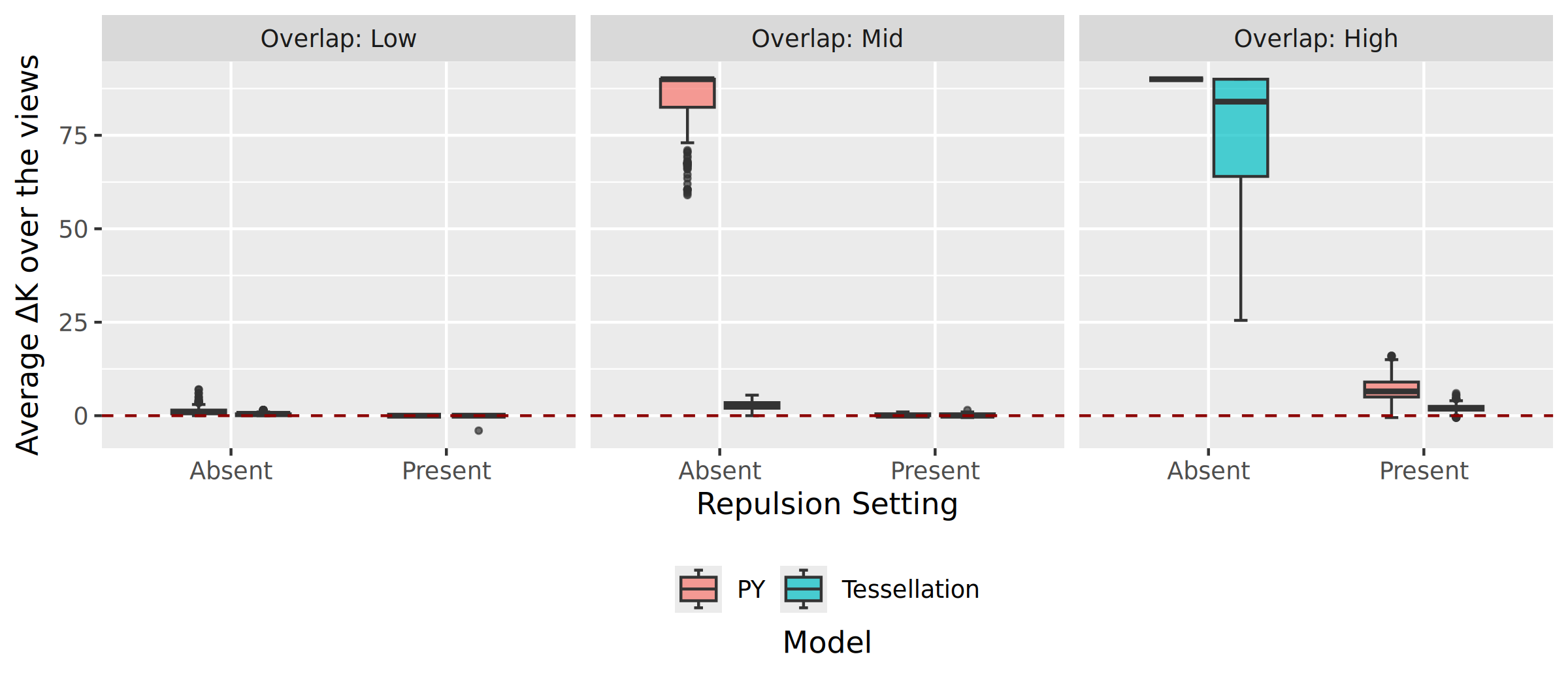}\vspace{-0.4cm}
    \caption{Average cardinality discrepancy $\Delta K = \frac{1}{2}(\hat{K}_1 - 10) + \frac{1}{2}(\hat{K}_2 - 10)$.}
    \label{fig:clust_recovery_repulsion}
  \end{subfigure}
  \caption{Comparison between non-repulsive (\textit{Absent}) and repulsive (\textit{Present}) likelihood specifications for \textit{Pitman-Yor} and \textit{Tessellation} models under quadratic likelihood across overlap regimes. Boxplots display empirical distributions across 50 replicates; the red dashed line indicates exact cardinality recovery.}
  \label{fig:repulsion_performance}
\end{figure}

When cluster centers are well separated, the role of the repulsion term is marginal, as all likelihood specifications achieve near-perfect partition recovery. As cluster overlap becomes moderate, the Pitman-Yor mixture model without a repulsion term suffers from posterior degeneracy, fragmenting into singleton clusters and causing the ARI to collapse near $0.00$; in contrast, the tessellation model remains more robust even without repulsion. Including the repulsion term in the likelihood fully restores both frameworks to optimal performance. The estimation gain provided by the repulsion term is even more pronounced in the high overlap scenario, where both models without repulsion suffer from severe cluster over-estimation. In this regime, the repulsion term is essential to prevent over-splitting, keeping estimated cardinality tightly bounded ($\Delta K \approx 1.9$ for \textit{Tessellation} vs.\ $\Delta K \approx 7$ for \textit{PY}) and increasing clustering accuracy to a median ARI of $0.66$ and $0.53$, respectively. The results suggest that the low ARI values observed in some simulation runs are mainly due to difficulties in cluster identifiability when cluster overlap is moderate or high, an effect that is substantially reduced by the inclusion of the repulsion term.

\begin{figure}[t!]
  \centering
  \includegraphics[width=\linewidth]{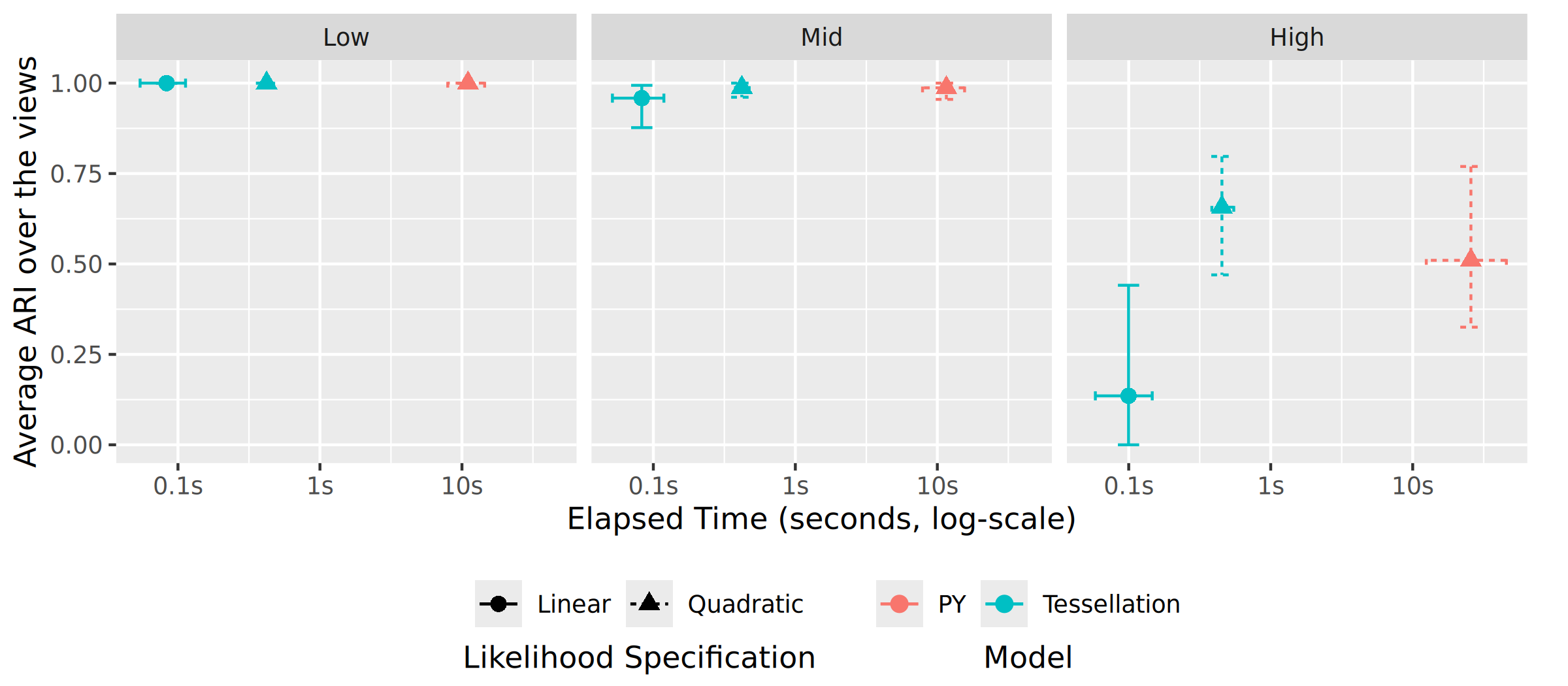}
  \caption{Trade-off between clustering accuracy and computational cost under alternative likelihood specifications: average Adjusted Rand Index (ARI) versus CPU elapsed time (in seconds, $\log_{10}$ scale) across 50 replicates, faceted by cluster overlap regime. Points represent the bivariate empirical means; horizontal and vertical error bars denote the $95\%$ empirical intervals for runtime and clustering accuracy, respectively.}
  \label{fig:ari_time_tradeoff_lik}
\end{figure}

We evaluate the influence of the likelihood function by comparing the quadratic likelihood $\pi\left(\mathbf{D} \mid \bm{\theta}, \bm{\lambda}, \bm{\gamma} \right)$ in \eqref{eq:mainequation} against the linear formulation $\pi^{\text{lin}}\left(\mathbf{D} \mid \bm{\lambda}, \bm{\gamma}\right)$ in \eqref{eq:linear_simple} of the main manuscript. Figure~\ref{fig:ari_time_tradeoff_lik} illustrates the trade-off between clustering accuracy (measured via average ARI over the two views) and computational cost (CPU elapsed time in seconds, $\log_{10}$ scale) across overlap regimes. In addition, Figure~\ref{fig:ari_lik} details the clustering accuracy across model structures (\textit{Tessellation} vs.\ \textit{PY}) and formulations (\textit{Independent} vs.\ \textit{Joint}) when the repulsion term is present. Finally, Table~\ref{tab:ablation_summary} provides the full numerical breakdown of clustering accuracy (ARI) and cluster cardinality discrepancy ($\Delta K$) across all model variants, likelihood specifications and repulsion configurations.

\begin{figure}[t!]
    \centering
    \includegraphics[width=\linewidth]{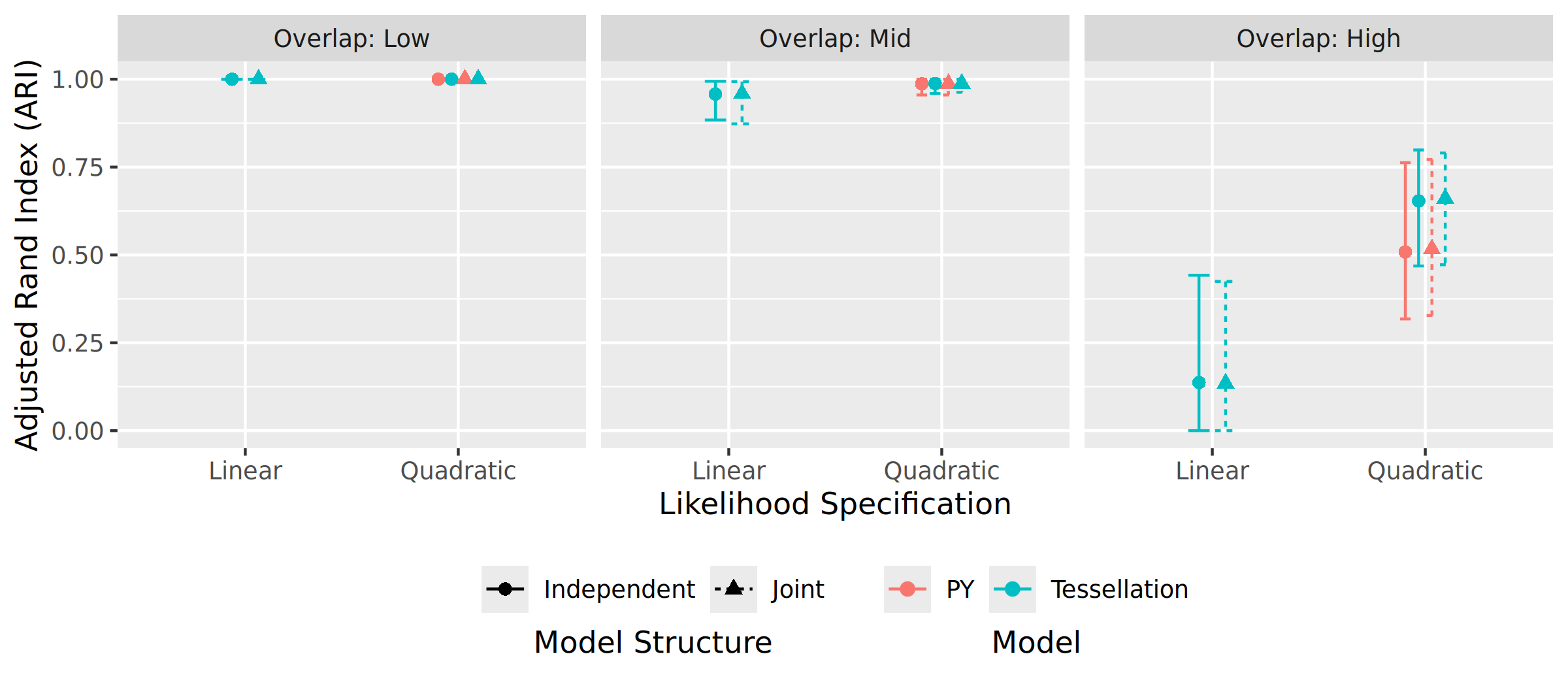}
    \caption{Clustering accuracy by likelihood specification and model structure: average ARI between views and $95\%$ empirical intervals for Linear vs.\ Quadratic likelihoods across overlap regimes, distinguishing \textit{Independent} (circles) and \textit{Joint} (triangles) model formulations. All specifications include the repulsive term in the likelihood.}
    \label{fig:ari_lik}
\end{figure}

We observe a clear inferential trade-off: in the low overlap scenario, the linear likelihood achieves identical clustering accuracy to the quadratic formulation ($\text{ARI} \approx 1.000$, $|\Delta K| \le 0.022$) while substantially reducing computational time. In the moderate overlap scenario, the linear likelihood maintains strong partition recovery at negligible computational cost, whereas the quadratic likelihood achieves near-optimal recovery ($\text{ARI} = 0.987$ with repulsion) at a higher computational expense. Under severe cluster overlap, the computational overhead of the quadratic formulation is necessary to achieve higher estimation accuracy: the linear likelihood with repulsion deteriorates markedly ($\text{ARI} = 0.162$, $\Delta K \approx -6$), whereas the quadratic likelihood with repulsion keeps cardinality tightly controlled ($\Delta K \approx 2$) and reaches the highest accuracy ($\text{ARI} \approx 0.66$). This suggests that the linear likelihood serves as a fast and reliable approximation for well-separated or moderately noisy data, while the quadratic likelihood with repulsion is the default choice to achieve more accurate results, especially in cases of more complex cluster structures.

\begin{table}[t!]
  \centering
  \resizebox{\linewidth}{!}{%
  \begin{tabular}{ccccccccc}
    \bottomrule[2pt]
    \multirow{2}{*}{Model} & \multirow{2}{*}{Likelihood} & \multirow{2}{*}{Repulsion~} & \multicolumn{2}{c}{Overlap: Low} & \multicolumn{2}{c}{Overlap: Mid} & \multicolumn{2}{c}{Overlap: High} \\ \cmidrule{4-9}
    & & & ARI (SD) & $\Delta K$ (SD) & ARI (SD) & $\Delta K$ (SD) & ARI (SD) & $\Delta K$ (SD) \\
    \midrule[2pt]
    \multirow{2}{*}{PY} & Quadratic & Absent & 0.992 (0.006) & 1.141 (0.966) & 0.085 (0.140) & 85.920 (6.939) & 0.000 (0.000) & 90.000 (0.000) \\
    & Quadratic & Present & 1.000 (0.000) & 0.000 (0.000) & 0.985 (0.012) & 0.208 (0.293) & 0.525 (0.115) & 7.005 (2.931) \\
    \midrule
    \multirow{4}{*}{Tessellation} & Linear & Absent & 0.999 (0.003) & 0.022 (0.103) & 0.972 (0.021) & 1.238 (0.808) & 0.482 (0.184) & 2.864 (5.225) \\
    & Linear & Present & 1.000 (0.001) & 0.001 (0.022) & 0.953 (0.033) & 0.406 (0.518) & 0.162 (0.124) & -6.363 (1.855) \\
    & Quadratic & Absent & 0.996 (0.004) & 0.330 (0.355) & 0.945 (0.023) & 2.699 (0.966) & 0.069 (0.091) & 75.782 (15.815) \\
    & Quadratic & Present & 0.999 (0.021) & -0.008 (0.179) & 0.987 (0.011) & 0.148 (0.255) & 0.655 (0.083) & 1.915 (1.000) \\
    \bottomrule[2pt]
  \end{tabular}}
  \caption{Detailed sensitivity analysis across MCMC models and likelihood specifications. Mean (and standard deviation) of the average ARI across the views and cluster cardinality error $\Delta K = \frac{1}{2}\sum_{v=1}^2 (\hat{K}_v - 10)$. Values are obtained across 50 replicates under Low, Mid, and High cluster overlap regimes.}
  \label{tab:ablation_summary}
\end{table}

\subsection{Computational runtime and scalability}
\label{subsec:scalability}
Table~\ref{tab:runtime_matrix} provides a comprehensive summary of computational runtimes across all evaluated models, formulations, and likelihood specifications, including empirical standard deviations, $[\min, \max]$ ranges, and relative speed-up factors compared to the standard joint Bayesian mixture baseline (\textit{PY Joint} with quadratic likelihood, averaging $29.8$ s). While deterministic heuristics ($k$-medoids via PAM and CLARA) execute in fractions of a second, they do not provide posterior uncertainty quantification. Among the probabilistic models, the proposed Tessellation framework demonstrates substantial computational efficiency: under the quadratic likelihood, both the independent and joint formulations are roughly two orders of magnitude faster than the baseline ($98.3\times$ and $85.0\times$ speed-up, respectively). Substituting the quadratic likelihood with the linear formulation accelerates computation further ($0.067$ s for independent, $0.124$ s for joint), representing up to a $445.2\times$ speed-up over the baseline MCMC implementation.

\begin{table}[t!]
  \centering
  \resizebox{\linewidth}{!}{%
  \begin{tabular}{cccccc}
    \bottomrule[2pt]
    Model Class & Formulation & Likelihood & Mean Elapsed Time (SD) & Range [Min, Max] & Relative Speed-up \\
    \midrule[2pt]
    \multirow{2}{*}{$k$-medoids} & PAM & -- & 0.006 (0.001) & [0.004, 0.013] & 4889.8x faster \\
    & CLARA & -- & 0.008 (0.001) & [0.006, 0.015] & 3654.4x faster \\
    \midrule
    \multirow{4}{*}{Tessellation} & Independent & Linear & 0.067 (0.011) & [0.050, 0.108] & 445.2x faster \\
    & Joint & Linear & 0.124 (0.014) & [0.093, 0.169] & 239.5x faster \\
    & Independent & Quadratic & 0.303 (0.110) & [0.165, 0.626] & 98.3x faster \\
    & Joint & Quadratic & 0.351 (0.109) & [0.203, 0.598] & 85.0x faster \\
    \midrule
    \multirow{2}{*}{PY} & Independent & Quadratic & 8.564 (6.887) & {}[1.729, 48.960] & 3.5x faster \\
    & Joint & Quadratic & 29.800 (19.527) & {}[7.355, 63.663] & Baseline (1x) \\
    \bottomrule[2pt]
  \end{tabular}}
  \caption{Computational runtime and speed-up comparison: Mean CPU elapsed time (seconds), standard deviation, $[{\min}, {\max}]$ observed interval across replications, and speed-up relative to the standard Joint Pitman-Yor baseline with quadratic likelihood. MCMC settings: 10,000 iterations, 2,000 of which are discarded as burn-in.}
  \label{tab:runtime_matrix}
\end{table}

Finally, to evaluate the computational scalability of the proposed algorithms, we measured execution time across two dimensions: the number of observations per view ($N$) and the number of distinct views ($V$). We compared the Joint formulations (\textit{Tessellation Joint} and \textit{PY Joint}) against their independent baselines (\textit{Tessellation Independent} and \textit{PY Independent}), as well as against standard non-Bayesian clustering algorithms (\textit{PAM} and \textit{CLARA}). We study the scalability of the algorithms under both linear and quadratic likelihood specifications, keeping the repulsion term present. Data has been generated as described at the beginning of Section~\ref{sect:simulations} of the main manuscript. We focus on a computationally inexpensive scenario with low cluster overlap ($\sigma_s = 0.1$) and medium inter-view agreement ($\alpha_s = 0.5$). To isolate algorithmic time complexity from hardware-induced memory bottlenecks, computational times were recorded sequentially on the same workstation described at the beginning of Section~\ref{sect:simulations} of the main manuscript.

Figure~\ref{fig:scalability_ndata} illustrates the elapsed execution time as $N$ increases from $100$ to $25,600$ (keeping $V$ fixed and equal to $2$), plotted on a log-log scale and faceted across the two likelihood specifications. On this scale, the slope of each line represents the empirical polynomial time complexity of the respective algorithm. The baseline \textit{PAM} algorithm exhibits severe scalability limitations as it relies on full pairwise dissimilarity evaluations. For a detailed breakdown of the computational time for each model, see Table~\ref{tab:scalability_ndata}.

\begin{figure}[t!]
  \centering
  \includegraphics[width=\textwidth]{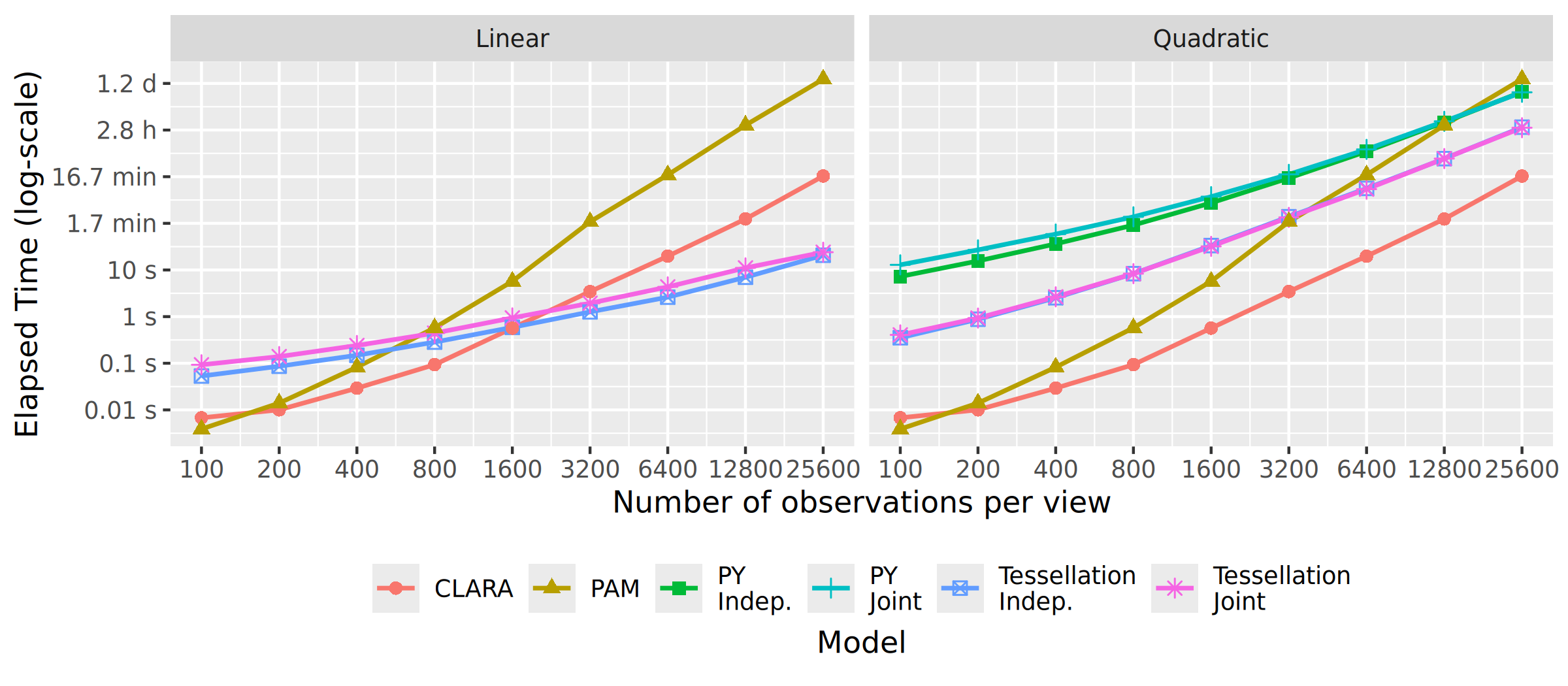}
  \caption{Computational scalability with respect to the number of observations $N$ in each view on a log-log scale and faceted across likelihood specifications. The number of views is fixed at $V = 2$.}
  \label{fig:scalability_ndata}
\end{figure}

\begin{table}[t!]%
  \centering
  \resizebox{\linewidth}{!}{%
  \begin{tabular}{ccccccc}
    \bottomrule[2pt]
    Model Class & Formulation & Likelihood & Time ($N=100$) & Time ($N=1600$) & Time ($N=25,600$) & Joint/Indep. Ratio\\
    \midrule[2pt]
    \multirow{2}{*}{$k$-medoids} & CLARA & -- & 0.01 s & 0.57 s & 17.2 min & --\\
    & PAM & -- & 0.00 s & 5.72 s & 1.43 d & -- \\ \midrule
    \multirow{2}{*}{PY} & Independent & Quadratic & 7.24 s & 4.6 min & 18.4 h & -- \\
    & Joint & Quadratic & 12.89 s & 6.2 min & 17.9 h & 0.98x \\ \midrule
    \multirow{4}{*}{Tessellation} & Independent & Linear & 0.05 s & 0.59 s & 20.71 s & -- \\
    & Joint & Linear & 0.09 s & 0.94 s & 24.04 s & 1.16x \\
    & Independent & Quadratic & 0.35 s & 33.37 s & 3.2 h & -- \\
    & Joint & Quadratic & 0.41 s & 32.09 s & 3.1 h & 0.97x \\
    \bottomrule[2pt]
  \end{tabular}}
  \caption{Computational scalability with respect to the number of observations $N$ in each view. The number of views is fixed at $V = 2$.}
  \label{tab:scalability_ndata}
\end{table}

We observe that \textit{PAM} crosses the computational cost of the proposed tessellation models around $N = 800$ for the linear likelihood and $N = 3200$ for the quadratic likelihood and becomes highly prohibitive for larger datasets. Conversely, \textit{CLARA} maintains a better empirical complexity, as it relies on data subsampling. All four Bayesian models show an excellent asymptotic scaling, which is comparable, or even better, than \textit{CLARA} empirical scaling. Since the proposed tessellation models explore the space of medoids and not the space of random partitions, \textit{Tessellation Independent} and \textit{Tessellation Joint} achieve a remarkably lower computational intercept if compared to the Pitman-Yor counterparts. For example, at $N = 25,600$, tessellation models retrieve $8,000$ posterior samples in around $3$ hours, while Pitman-Yor models require $\simeq 18$ hours. Finally, notice that, across both linear and quadratic likelihood specifications, Joint formulations are virtually superimposed on their Independent counterparts. This indicates that the iterative cross-view information sharing mechanism introduces a constant (and in this case almost negligible) overhead.

In a similar way, Figure~\ref{fig:scalability_nviews} displays execution time as a function of the number of views $V$, keeping the number of observations fixed at $N = 100$. As expected, the Independent formulations (\textit{CLARA}, \textit{PAM}, \textit{Tessellation Independent} and \textit{PY Independent}) scale linearly with respect to $V$, as each view is processed in complete isolation. By contrast, Joint formulations exhibit a steeper, super-linear scaling trajectory. This divergence is an expected consequence that directly derives from the joint modeling prior since we need to iterate across each view and compute $V-1$ interaction terms (see Algorithm~\ref{algo:joint}). Consequently, the cross-view update mechanism scales quadratically with respect to $V$ and such overhead can not be avoided. Nevertheless, the empirical scaling of the Joint formulations  remains tractable and, in practice, the computational overhead is negligible for a moderate number of views ($V \le 10$). For a detailed breakdown of the computational time for each model, see Table~\ref{tab:scalability_nviews}.

\begin{figure}[t!]
  \centering
  \includegraphics[width=\textwidth]{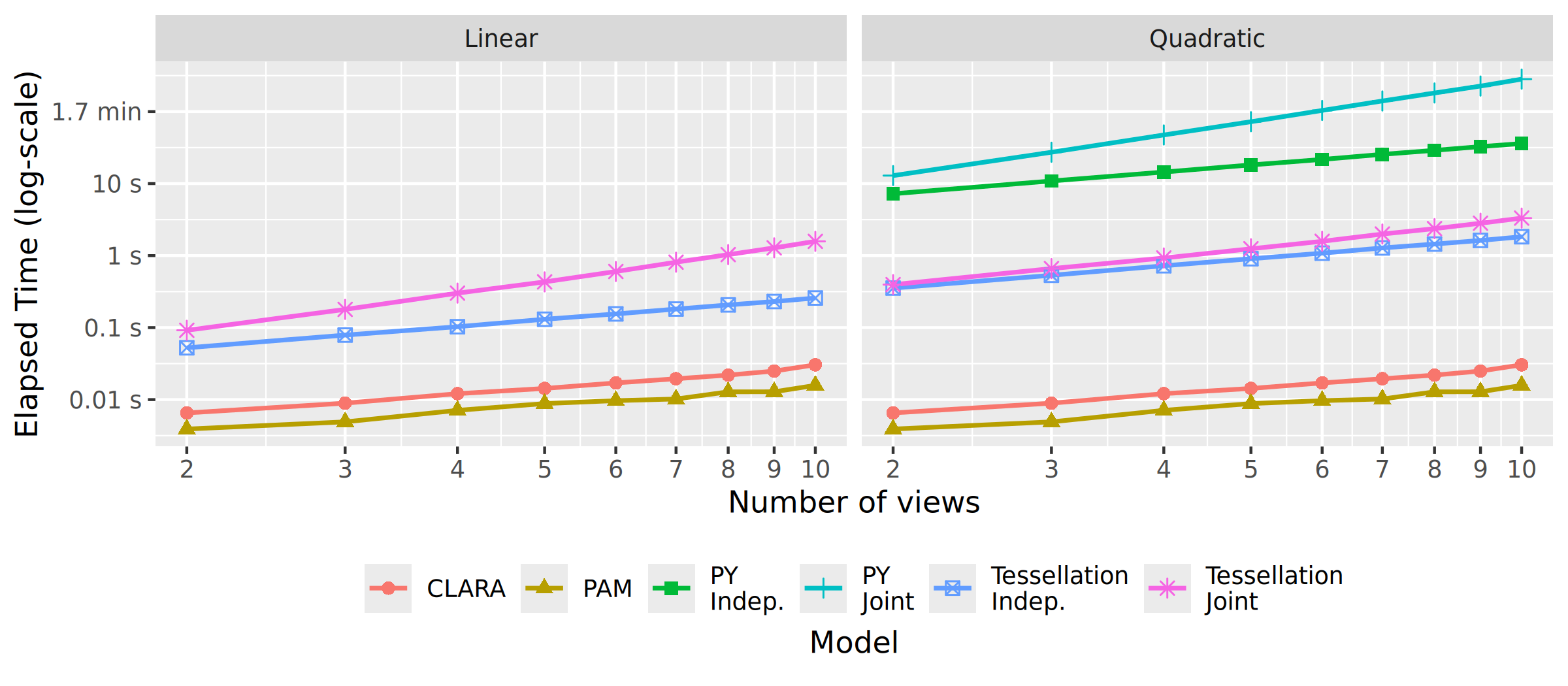}
  \caption{Computational scalability with respect to the number of views $V$ on a log-log scale and faceted across likelihood specifications. The number of observations in each view is fixed at $N = 100$.}
  \label{fig:scalability_nviews}
\end{figure}

\begin{table}[t!]
  \centering
  \resizebox{\linewidth}{!}{%
  \begin{tabular}{ccccccc}
    \bottomrule[2pt]
    Model Class & Formulation & Likelihood & Time ($V=2$) & Time ($V=6$) & Time ($V=10$) & Joint/Indep. Ratio\\
    \midrule[2pt]
    \multirow{2}{*}{$k$-medoids} & CLARA & -- & 0.01 s & 0.02 s & 0.03 s & -- \\
    & PAM & -- & 0.00 s & 0.01 s & 0.02 s & -- \\ \midrule
    \multirow{2}{*}{PY} & Independent & Quadratic & 7.23 s & 21.67 s & 36.03 s & -- \\
    & Joint & Quadratic & 12.92 s & 1.7 min & 4.7 min & 7.84x \\ \midrule
    \multirow{4}{*}{Tessellation} & Independent & Linear & 0.05 s & 0.15 s & 0.26 s & -- \\
    & Joint & Linear & 0.09 s & 0.60 s & 1.58 s & 6.13x \\
    & Independent & Quadratic & 0.35 s & 1.08 s & 1.83 s & -- \\
    & Joint & Quadratic & 0.40 s & 1.58 s & 3.32 s & 1.81x \\
    \bottomrule[2pt]
  \end{tabular}}
  \caption{Computational scalability with respect to the number of views $V$. The number of observations is fixed at $N = 100$.}
  \label{tab:scalability_nviews}
\end{table}

\section{Numismatics application}
\label{sect:num_appendix}

\subsection{Workflow for image preprocessing and distance computation}\label{app:workflow}
Besides the D2-Net \citep{dusmanu2019d2} employed in \cite{harris2024past} to match key points between the two coins, we also compute local self-similarity descriptors. Figure \ref{fig:ss_descriptor} shows the self-similarity descriptor for a given pixel $q$. A surrounding image region centered at $q$ (with a radius of 40 pixels) is compared with the inner region patch of 5-pixel radius centered on $q$ using the sum of square differences (SSD) in the color space. The resulting distance surface $SSD_q(x,y)$ is normalized to form a ``correlation surface" $S_q(x,y)$:
\begin{equation}
    S_q(x,y) = \exp\left(-\frac{SSD_q(x,y)}{\max(\text{var}_{\text{noise}},\text{var}_{\text{auto}}(q))}\right)
\end{equation}
where $\text{var}_{\text{noise}}$ is a constant for acceptable photometric variations and $\text{var}_{\text{auto}}(q)$ accounts for the pattern structure and contrast at $q$. This ensures that sharp edges are more tolerant to pattern variations than smooth patches. The correlation surface $S_q(x,y)$ is then converted into log-polar coordinates centered at $q$ and divided into 90 bins (15 angles, six radial intervals). The maximum correlation value in each bin forms the descriptor vector $d_q$, normalized to a range of $[0,1]$ to measure the similarity between patches while considering geometric and photometric variations. More details about this descriptor can be found in \cite{shechtman2007matching}. Afterward, we compute the Euclidean distance between the descriptor centered on $q$ between coins $i$ and $j$, leading to the distance $D^{(SS)}_{i,j,q}$.

\begin{figure}[t!]
    \centering
    \includegraphics[width=0.60\linewidth]{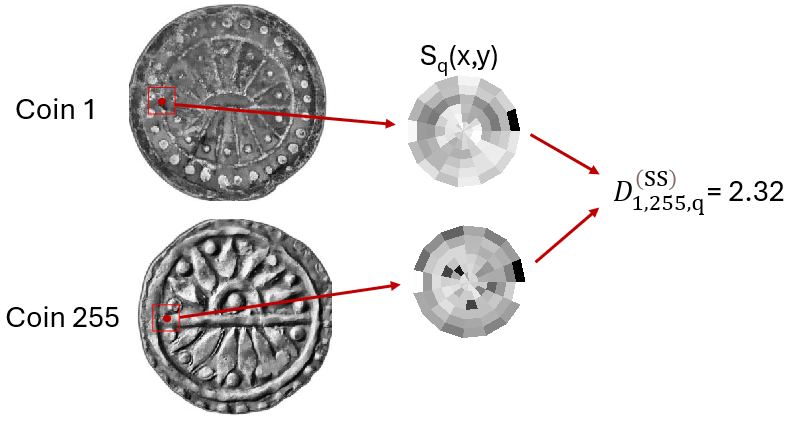}
    \caption{We show the self-similarity descriptors $S_q(x,y)$ in polar coordinates centered on the red dot for two coins and the Euclidean distance between those two descriptors.}
    \label{fig:ss_descriptor}
\end{figure}

We subsequently detail the workflow for image preprocessing and distance computation of Section \ref{section:numismatics} of the main manuscript. We mainly employed the R programming language, but steps 3,4, and 6 are performed in Mathematica \cite{Mathematica}, and steps 1, 2, 5, and 7 are performed in Matlab \cite{MATLAB}. The workflow is based on the workflow in \cite{harris2024past}, but we add steps 3, 4, 5, 6 and 8. We apply the preprocessing steps (1 to 5) and distance computation steps (6 to 9) separately for the coin obverse and reverses, and the steps are
\begin{enumerate}
    \item \textit{Conversion and Resizing:} We transform all images to grayscale and standardize their dimensions to 300x300 pixels.
    \item \textit{Noise Reduction and Restoration:} We perform total-variation image restoration \cite{rudin1992nonlinear} on each image to mitigate noise generated during image acquisition due to low lighting conditions. This process helps to smooth out minor imperfections while maintaining the sharpness of edges and corners. Next, we use contrast-limited adaptive histogram equalization \citep{pizer1987adaptive} to enhance local contrast and edge definition. Total-variation restoration is applied once more to address any artifacts these local transformations introduce.
    \item \textit{Background Removal:} We employ the \texttt{RemoveBackground} function in Mathematica to isolate the coins from the background.
    \item  \textit{Image Alignment:} We align the images by minimizing the mean squared error between the coin images and the reference coin image (the first one in the dataset). We employed the \texttt{ImageAlign} function in Mathematica with a mean square gradient descent method. The image alignment is necessary for steps 6 and 8.
    \item  \textit{Color histogram adjustment:} We adjust the histograms of each coin image to match the reference image, standardizing brightness and contrast across the dataset (needed for step 6). We employ the function \texttt{imhistmatch} in Matlab.
    \item \textit{Direct Distance and Similarity Measures:} We compute a variety of direct distances and similarity measures between the coin images, such as Euclidean distance, cosine distance, mean pattern intensity, gradient correlation, structural similarity index (SSIM), Wasserstein distance, and mutual information variation. These measures are computed using the \texttt{ImageDistance} function in Mathematica.
    \item \textit{D2-Net model:} The D2-Net model incorporates a 16-layer VGG16 convolutional neural network to detect and describe key points within the coin images that characterize their visual attributes. D2-Net yields a representation of the coin image that is simultaneously a detector (that is, a locator of key points in the image) and a descriptor (that is, represented by an array of features describing the image properties at the corresponding key points). A detailed introduction and implementation guidance of the D2-Net model are in \cite{dusmanu2019d2}. Following the feature extraction at the matched key points, the images are paired and analyzed using an efficient approach to nearest neighbor search \citep{muja2009fast}. This process matches landmarks identified among the D2-Net key points across different coin images. The procedure includes using Gaussian Processes landmarking \citep{gao2019gaussian} to rank the reliability of these landmarks for image comparison. From this procedure, we obtain the number of matches between coin images $i$ and $j$ (with a larger number of matches related to larger similarities),  $\text{S}^{(D2)}_{i,j}$.
    \item \textit{Local self-similarity descriptors:} We employ the  “local self-similarity descriptors” in \cite{shechtman2007matching}, which captures internal geometric patterns within images while accounting for small local affine deformations. It captures self-similarity of edges, repetitive patterns, and complex textures in a single unified way. For each coin image, we compute these descriptors in 3600 equally spaced locations, obtaining a vector $\mathbf{v}_{i,q}$, which is the descriptor at location $q$ of image $i$. We employ the C++ implementation in \cite{chatfield2009efficient}. Since the images have been aligned, we can directly compute the Euclidean distances between the descriptors of coin $i$ and $j$ at location $q$: $\text{D}^{(SS)}_{i,j,q} = ||\mathbf{v}_{i,q} - \mathbf{v}_{j,q} ||_2$. 
    \item \textit{Compute the final distances:} From the previous three steps, we obtain a variety of distances and similarity measures for each coin pair. We now have to select which similarities/distances are most relevant and how to weigh them appropriately. The similarity measures are given a negative weight (higher similarity means smaller distance). Most of the self-similarity descriptors $\text{D}^{(SS)}_{i,j,l}$ are not relevant since they pertained to locations outside the coins or locations inside the coins with slight variation between coin pairs. To select the best distances and similarity measures, we consider a subset of 50 coins. We fit a logistic regression model with a Lasso penalty, where the response is one if a coin pair was in the same die (or 0 otherwise), and the covariates are all the distances and similarities computed for each coin pair. The fitted regression coefficients from the logistic regression inform which of the distances/similarities are the most relevant, also allowing us to adequately compute the weight of each distance/similarity measure. Finally, from the weights in the two panels of Table~\ref{tab:weights}, we compute the final distances between each coin pair (for both the obverses and the reverses) by doing the weighted average of all the distances and similarity measures. After obtaining the final distance matrix $\mathbf{D}$, we further standardized it to obtain \( \tilde{\mathbf{D}} \). Subsequently, we apply the transformation \( Q(\Phi(D_{ij})) \), where \( Q \) is the quantile function of the Gamma distribution with a shape parameter of 3 and a rate parameter of 5. This final transformation is performed to ensure that the histograms of the distances would resemble a Gamma distribution more closely (see Figure \ref{fig:coins2}). This transformation enhances the fit to the model described in \eqref{eq:mainequation} and \eqref{eq:linear_simple} of the main manuscript, which assumes a Gamma distribution.
\end{enumerate}

\begin{table}[t]
  \centering\scriptsize
  \begin{subtable}[b]{0.48\linewidth}
    \centering
    \begin{tabular}{cc}
      \bottomrule[2pt]
      Distance/Similarity & Weight \\ \midrule[2pt]
      Number of Matches in D2-Net & -8 \\
      SSIM & +7 \\
      Self-similarity descriptor at location 56 & +4 \\
      Self-similarity descriptor at location 306 & +4 \\
      Mean Pattern Intensity & -3 \\
      Gradient correlation & -2 \\ \bottomrule[2pt]
    \end{tabular}
  \end{subtable}
  \hfill
  \begin{subtable}[b]{0.48\linewidth}
    \centering
    \begin{tabular}{cc}
      \bottomrule[2pt]
      Distance/Similarity & Weight \\ \midrule[2pt]
      Cosine Distance & +11 \\
      Number of Matches in D2-Net & -8 \\
      SSIM & +7 \\
      Normalized Mutual Variation Information & +5 \\
      Self-similarity descriptor at location 63 & +4 \\
      Self-similarity descriptor at location 199 & +4 \\ \bottomrule[2pt]
    \end{tabular}
  \end{subtable}
  \caption{Rounded weights used to compute the final distances of the coin obverses (left panel) and reverses (right panel).}
  \label{tab:weights}
\end{table}

\subsection{Additional results and figures for the numismatics application}\label{sect:numismatics_figures}
This section presents supplementary results and figures pertinent to the numismatics application discussed in Section \ref{section:numismatics} of the main manuscript. Prior to conducting the clustering analysis, it was anticipated that the majority of coins in this application would be singletons (specimens struck from dies observed only once in the sample). Including hundreds of isolated coins unnecessarily expands the combinatorial search space of partitions and degrades MCMC mixing. To address this, we establish an empirical threshold to classify a coin as belonging to a singleton cluster or not.

To determine a principled distance cut-off, we model the distribution of non-zero pairwise distances using a two-component Gamma mixture model. As shown in Figure~\ref{fig:coins2}, the two components clearly distinguish coins struck from the same die (the lower-mean component) from coins struck from distinct dies (the higher-mean component), with a boundary between the two groups at approximately $0.07$. To adopt a conservative approach and minimize the risk of falsely discarding true die matches, we select a higher distance threshold of $0.15$, which captures the whole \textqt{within-die} distribution. We then evaluate each coin $i$ by computing the 1\% quantile of its pairwise distances to all other coins, denoted $q_i^{(0.01)}$, which serves as a proxy for the distance to its nearest neighbors. Any coin satisfying $q_i^{(0.01)} > 0.15$ show no evidence of shared die characteristics and is pre-classified as a singleton. We then apply the clustering algorithm exclusively to the remaining non-singleton coins.

\begin{figure}[t!]
    \centering
    \includegraphics[width=0.75\linewidth]{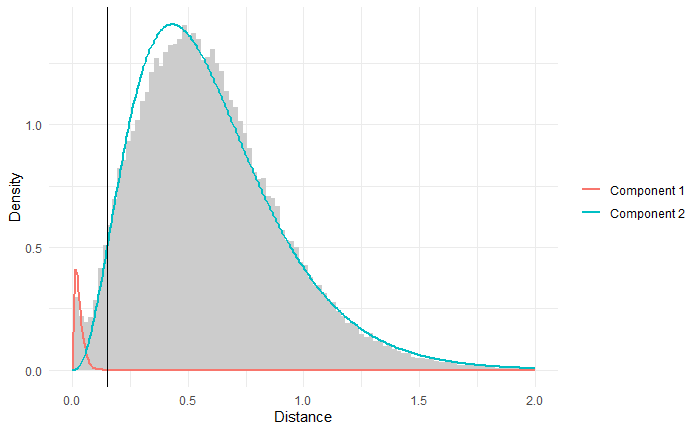}
    \caption{Two-component Gamma mixture fit to the empirical distribution of non-zero pairwise distances across the coins (grey histogram). The black vertical line indicates the conservative threshold at $0.15$.}
    \label{fig:coins2}
\end{figure}

\begin{figure}[ht!]
    \centering
    \begin{subfigure}[b]{0.45\textwidth}
        \centering
        \includegraphics[width=\linewidth]{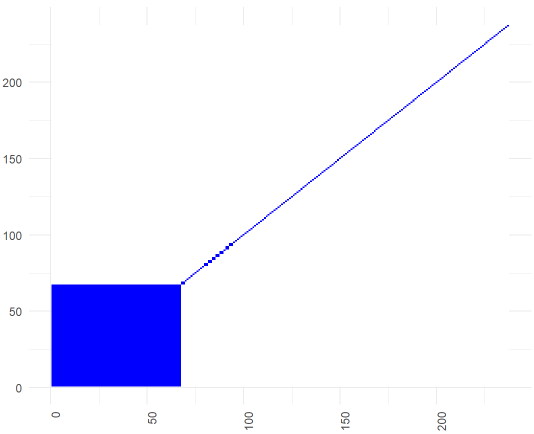}
    \end{subfigure}
    \hfill
    \begin{subfigure}[b]{0.45\textwidth}
        \centering
        \includegraphics[width=\linewidth]{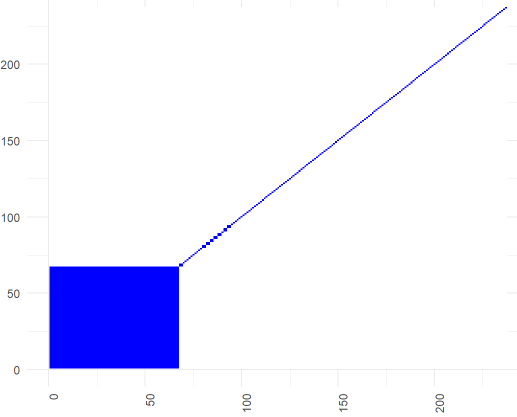}
    \end{subfigure}
    \caption{Adjacency matrix of the true clusters for obverses (left) and reverses (right).}
    \label{fig:adjacency_matrix}
\end{figure}

In Figure \ref{fig:adjacency_matrix}, we show the true adjacency matrices for the coin obverses and reverses. Figure \ref{tab:figure_table1} and \ref{tab:figure_table2} show the posterior co-clustering matrices for the independent and joint models, respectively. Finally, Figure \ref{fig:alpha2} shows the posterior distribution of the dependency parameter $\alpha$ for the joint PY and tessellation models, where the posterior median was 0.98 and 0.83, respectively. Figure \ref{fig:mcmc_diagnostics} presents posterior diagnostics for the numismatics application, comparing the Pitman-Yor model with the quadratic likelihood tessellation model. Specifically, the figure reports trace plots for the number of clusters, the corresponding autocorrelation functions, and the log-likelihood trajectories across MCMC iterations for the coin reverse data. 

\begin{figure}[htbp]
    \centering
    \begin{tabular}{ccc}
        \textbf{Obverse} & \textbf{Reverse} & \\ \midrule[2pt]
        \includegraphics[width=0.415\linewidth, height = 6cm]{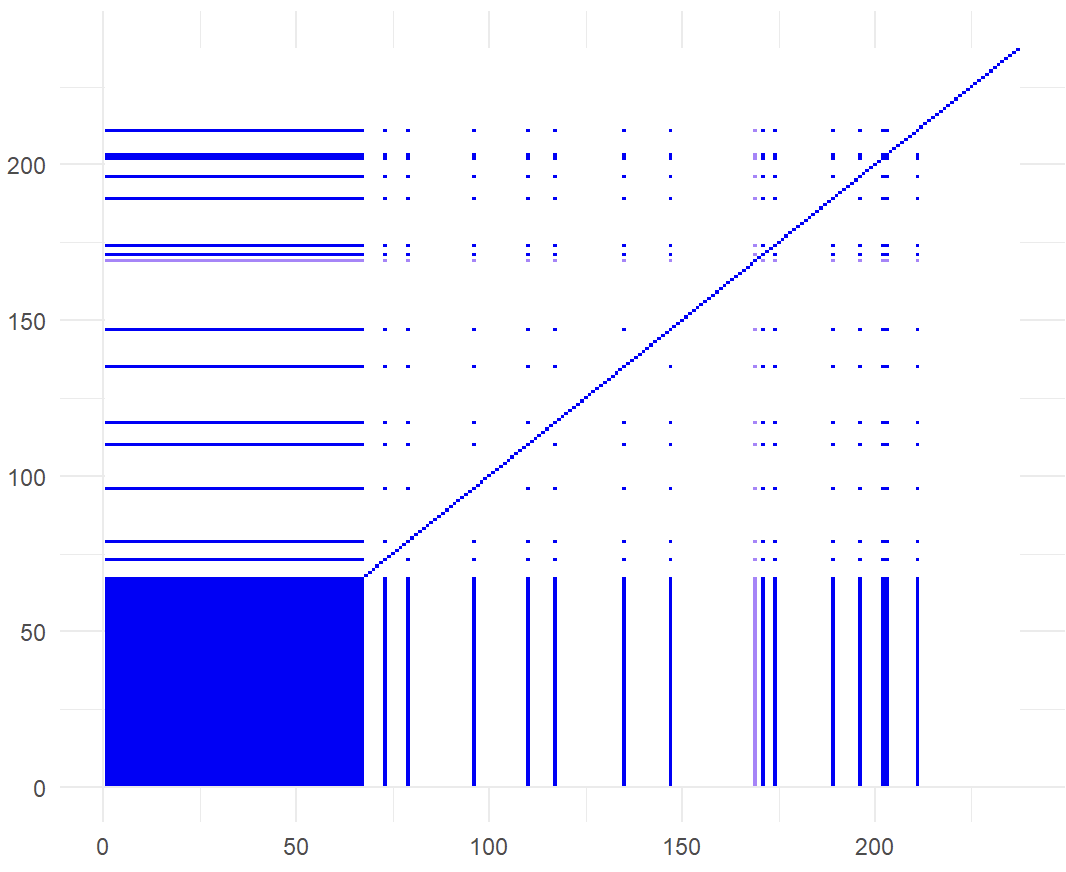} &
        \includegraphics[width=0.415\linewidth, height = 6cm]{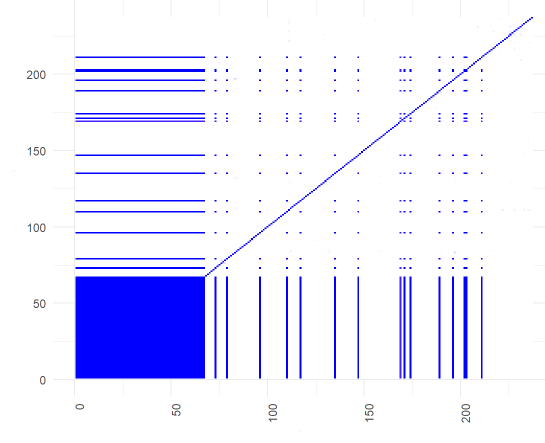} & 
        \includegraphics[width=0.07\linewidth, height = 6cm]{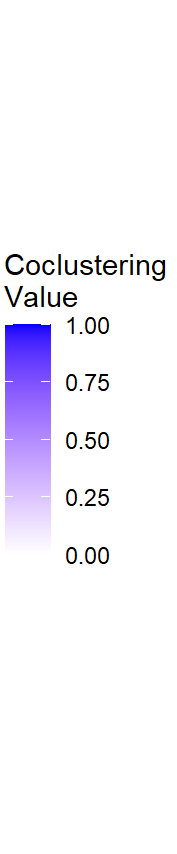} \\ \midrule[2pt]
        \includegraphics[width=0.415\linewidth, height = 6cm]{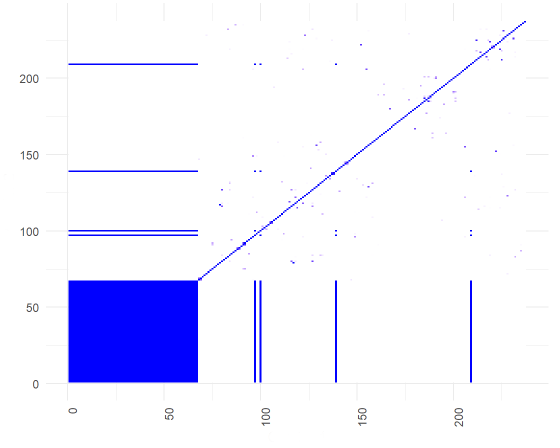} &
        \includegraphics[width=0.415\linewidth, height = 6cm]{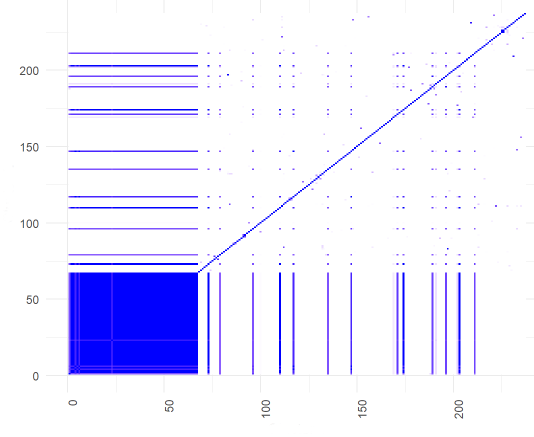} & 
        \includegraphics[width=0.07\linewidth, height = 6cm]{images/numismatic_application/scale.png} \\ \midrule[2pt]
        \includegraphics[width=0.415\linewidth, height = 6cm]{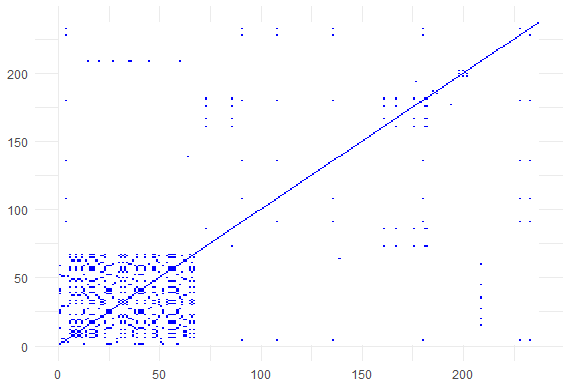} &
        \includegraphics[width=0.415\linewidth, height = 6cm]{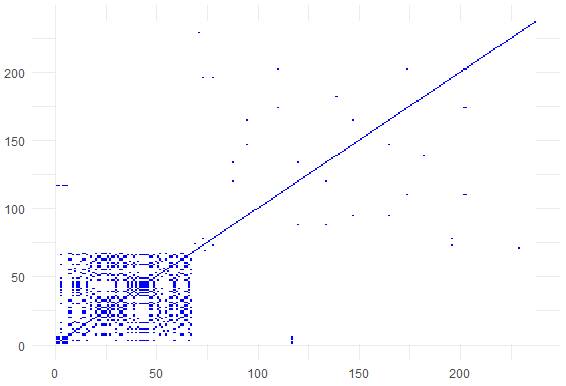} & 
        \includegraphics[width=0.07\linewidth, height = 6cm]{images/numismatic_application/scale.png} \\ \midrule[2pt]
    \end{tabular}
    \caption{Co-clustering matrices for the obverses and reverses for the independent Pitman-Yor model (first row), independent tessellation model (second row), and the PAM algorithm (third row).}
    \label{tab:figure_table2}
\end{figure}

\begin{figure}[htbp]
    \centering
    \begin{tabular}{ccc}
        \textbf{Obverse} & \textbf{Reverse} & \\ \midrule[2pt]
        \includegraphics[width=0.415\linewidth, height = 6cm]{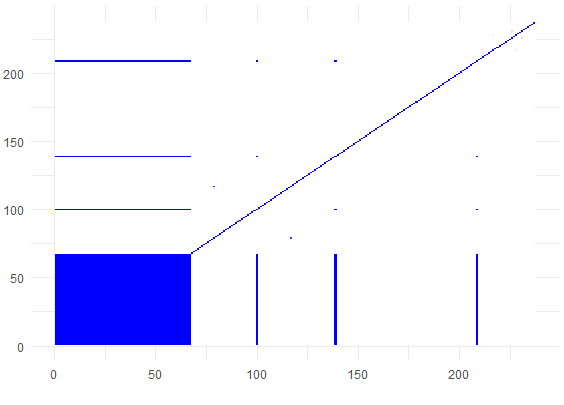} &
        \includegraphics[width=0.415\linewidth, height = 6cm]{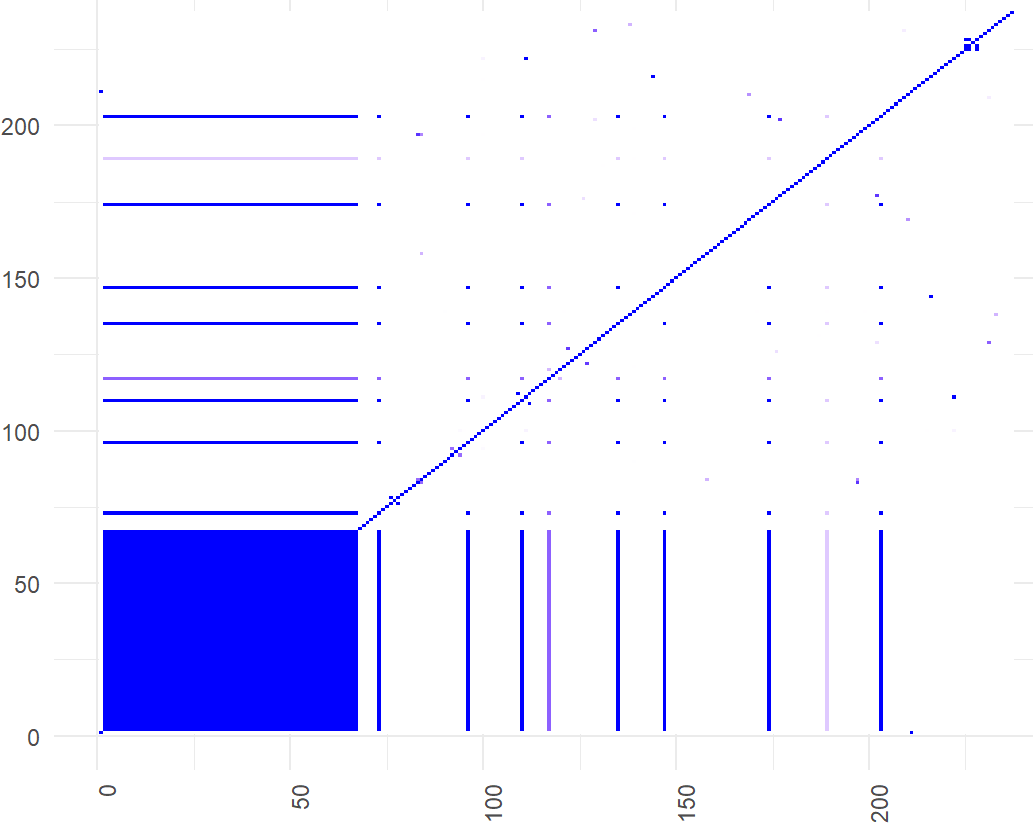} & 
        \includegraphics[width=0.07\linewidth, height = 6cm]{images/numismatic_application/scale.png}\\ \midrule[2pt]
        \includegraphics[width=0.415\linewidth, height = 6cm]{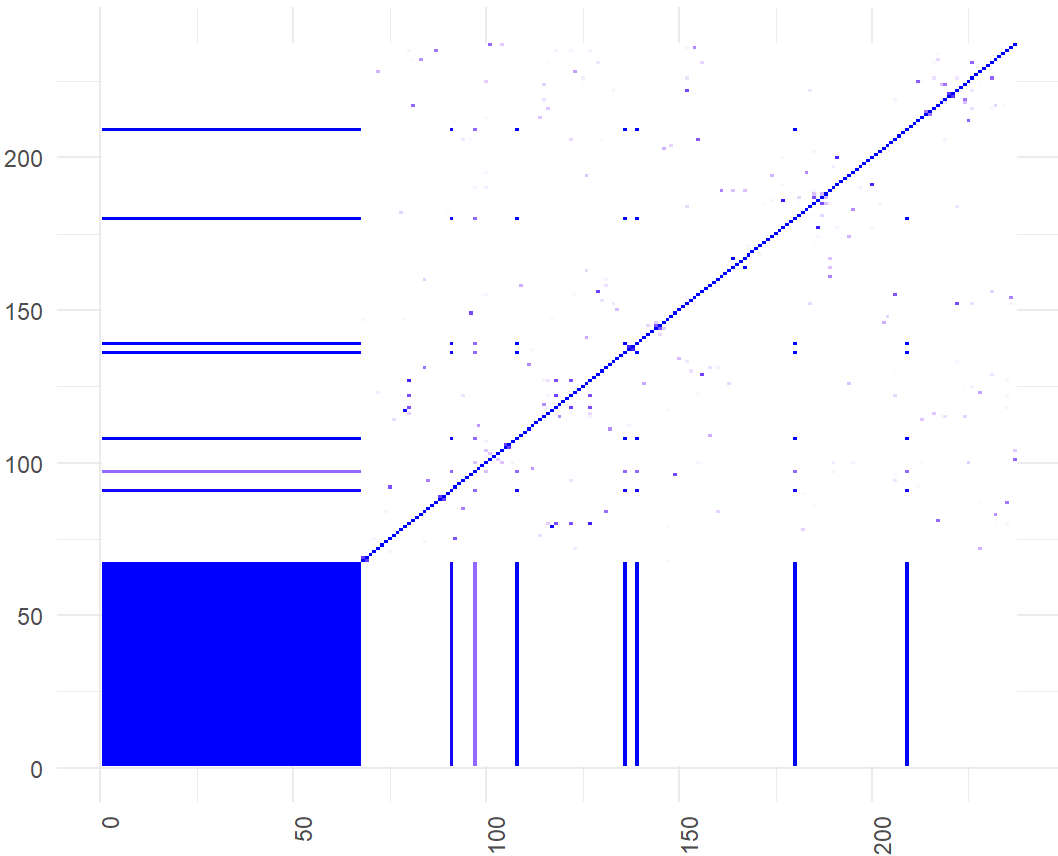} &
        \includegraphics[width=0.415\linewidth, height = 6cm]{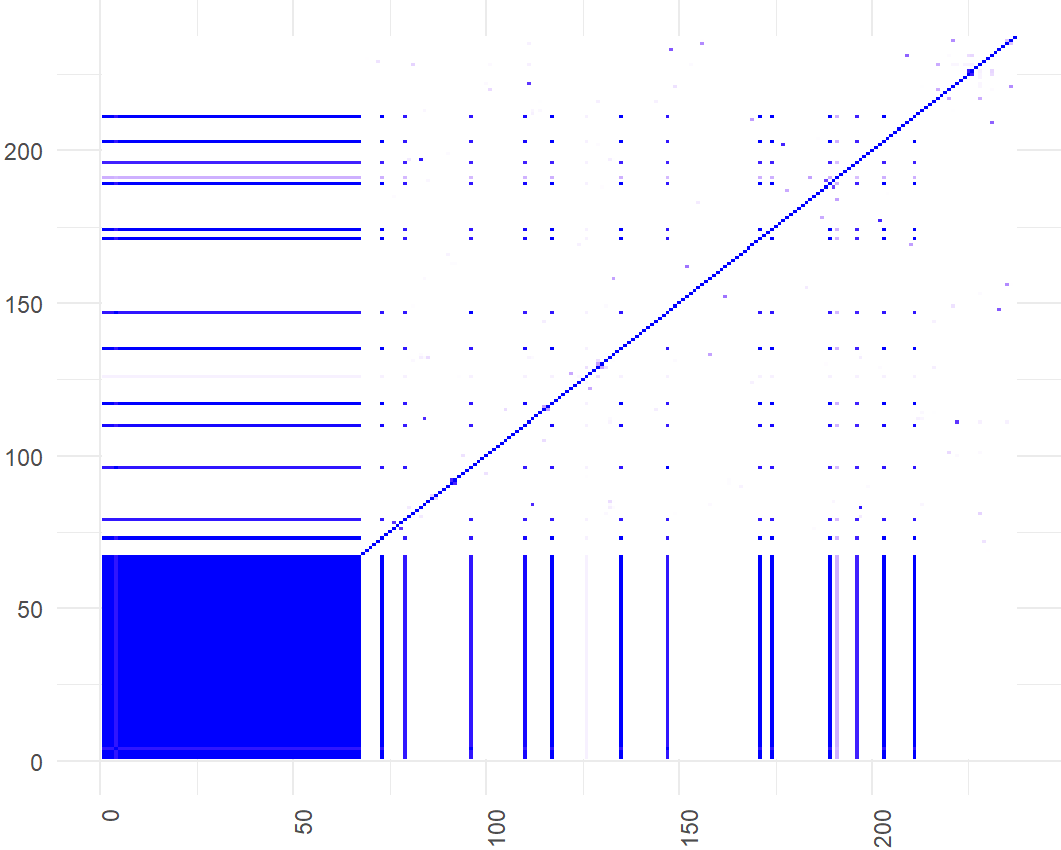} & 
        \includegraphics[width=0.07\linewidth, height = 6cm]{images/numismatic_application/scale.png} \\ \midrule[2pt]
    \end{tabular}
    \caption{Co-clustering matrices for the obverses and reverses for the joint Pitman-Yor model (first row) and joint tessellation model (sw row).}
    \label{tab:figure_table1}
\end{figure}

\begin{figure}[htbp]
    \centering
    \includegraphics[width=0.45\linewidth]{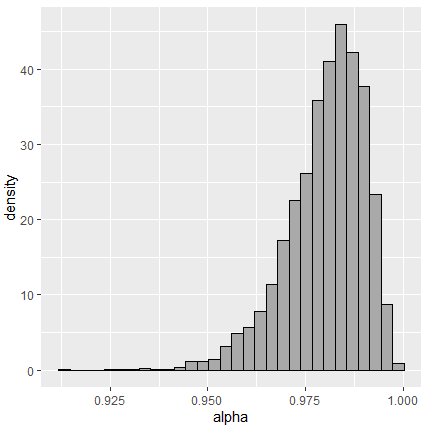}
    \includegraphics[width=0.45\linewidth]{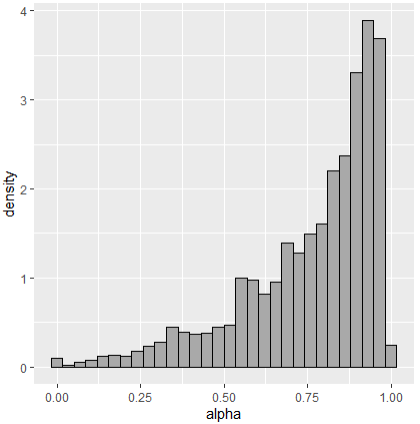}
    \caption{Dependency parameter $\alpha$ for the joint PY model (left) and joint tessellation model (right).}
    \label{fig:alpha2}
\end{figure}

\begin{figure}[htbp]
    \centering
    \begin{subfigure}[b]{0.32\textwidth}
        \centering
        \includegraphics[width=\linewidth]{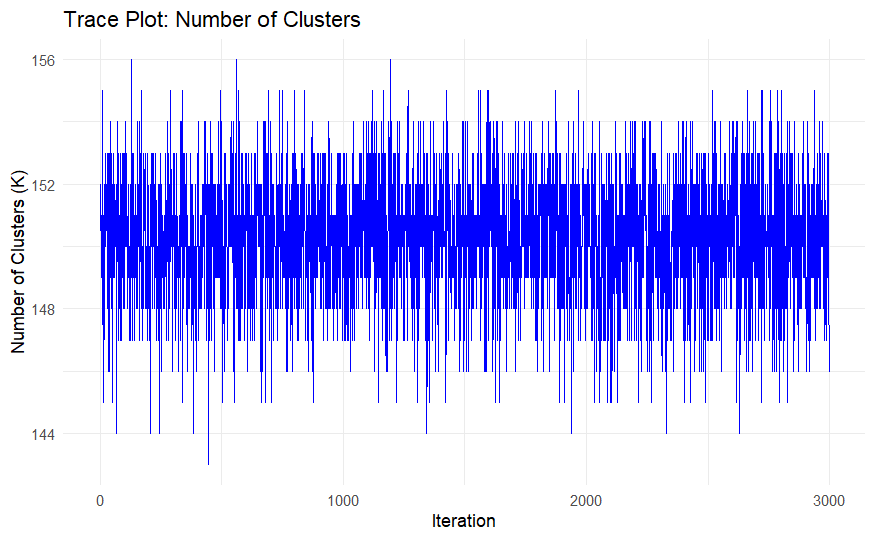}
    \end{subfigure}
    \hfill
    \begin{subfigure}[b]{0.32\textwidth}
        \centering
        \includegraphics[width=\linewidth]{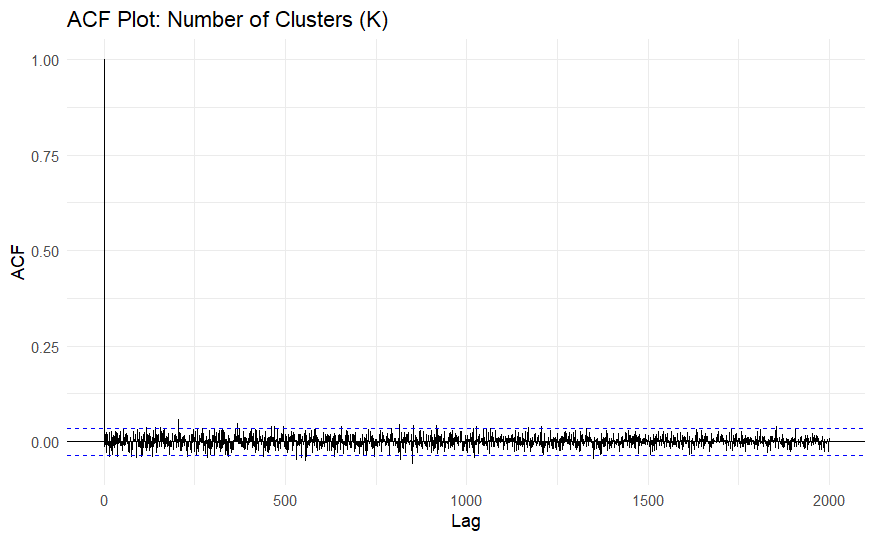}
    \end{subfigure}
    \hfill
    \begin{subfigure}[b]{0.32\textwidth}
        \centering
        \includegraphics[width=\linewidth]{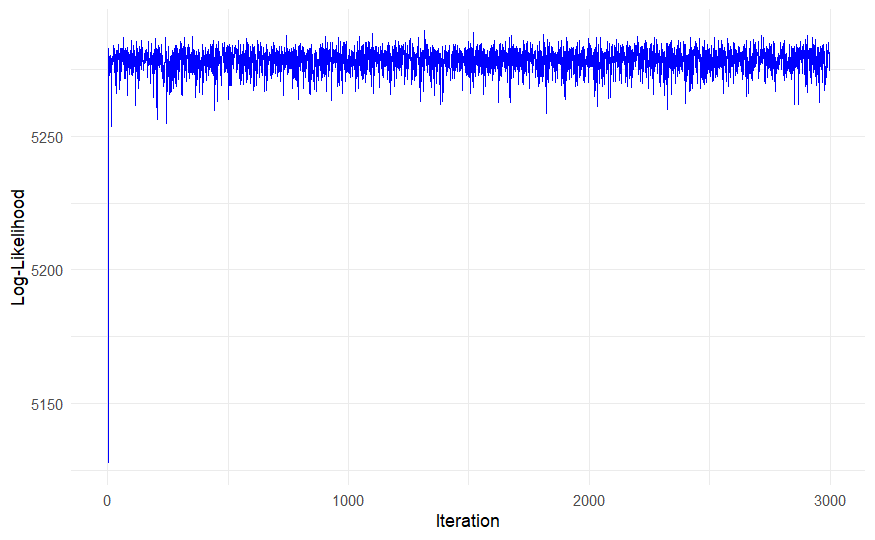}
    \end{subfigure}

    \vspace{1em} %

    \begin{subfigure}[b]{0.32\textwidth}
        \centering
        \includegraphics[width=\linewidth]{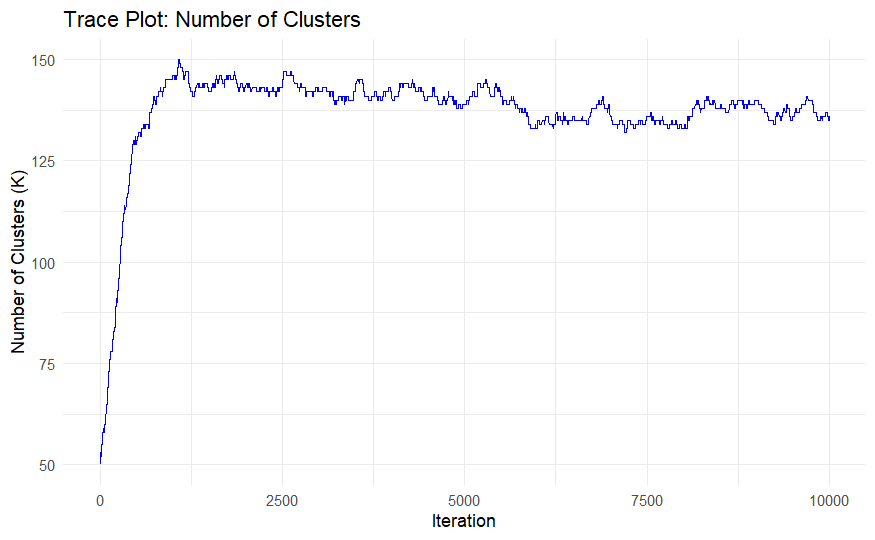}
    \end{subfigure}
    \hfill
    \begin{subfigure}[b]{0.32\textwidth}
        \centering
        \includegraphics[width=\linewidth]{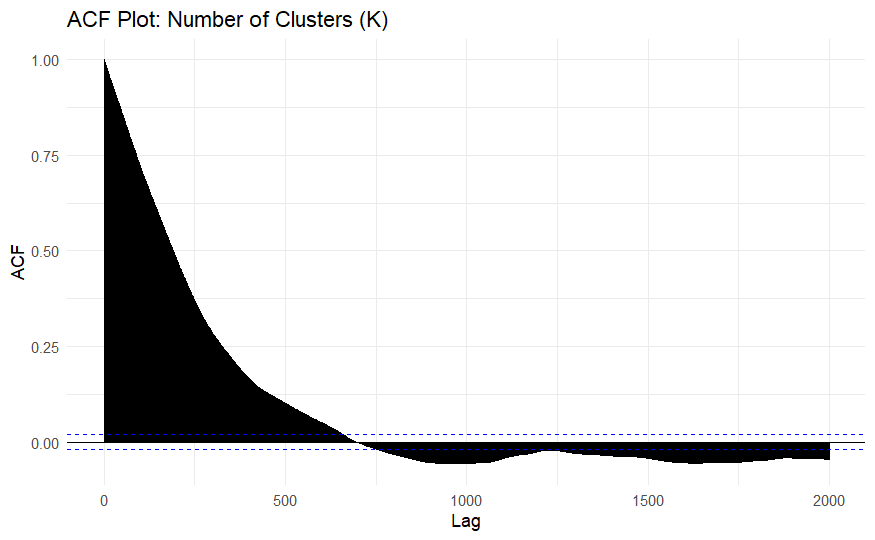}
    \end{subfigure}
    \hfill
    \begin{subfigure}[b]{0.32\textwidth}
        \centering
        \includegraphics[width=\linewidth]{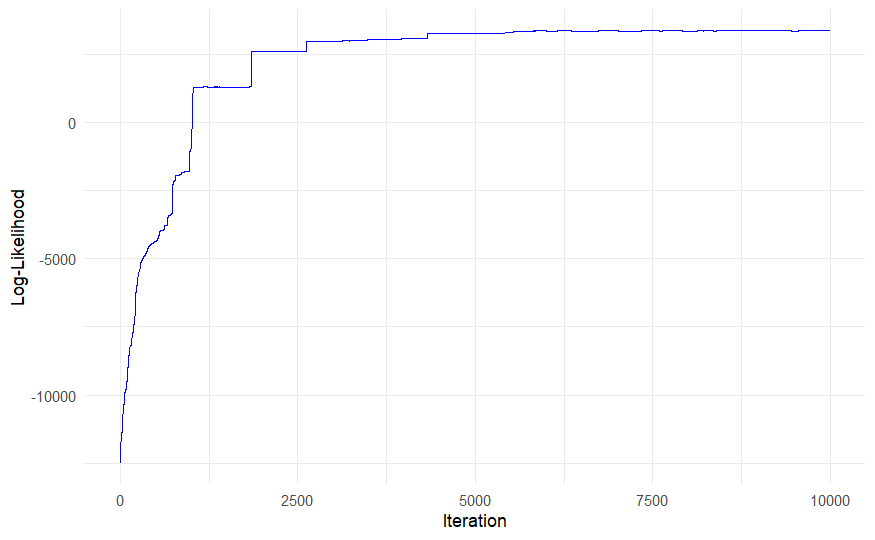}
    \end{subfigure}

    \caption{MCMC diagnostics for the numismatics application. The top row shows results for the PY model, and the bottom row shows the proposed quadratic likelihood tessellation model.}
    \label{fig:mcmc_diagnostics}
\end{figure}

\end{document}